\documentclass{article}

\usepackage{algorithm}
\usepackage[margin=30mm,includehead,includefoot]{geometry}
\usepackage[noend]{algpseudocode}
\usepackage{amsmath,amssymb,amsthm,amscd,color,comment}
\usepackage{autobreak}
\usepackage{color}
\usepackage{xcolor}
\usepackage{url}
\usepackage{blkarray}
\usepackage{jheppub}
\usepackage{empheq}
\usepackage{graphicx}

\newcommand{\CC}{\mathbb{C}}

\newcommand{\initial}{{\rm in}}

\newcommand{\NN}{\mathbb{N}}
\def\pd#1{\partial_{#1}}
\newcommand{\QQ}{\mathbb{Q}}

\newcommand{\ZZ}{\mathbb{Z}}
\newcommand{\RR}{\mathbb{R}}
\newcommand{\C}{\mathbb{C}}

\newcommand{\brk}[1]{(#1)}
\newcommand{\lrbrk}[1]{\left(#1\right)}
\newcommand{\bigbrk}[1]{\bigl(#1\bigr)}
\newcommand{\Bigbrk}[1]{\Bigl(#1\Bigr)}

\newcommand{\sbrk}[1]{[#1]}
\newcommand{\lrsbrk}[1]{\left[#1\right]}
\newcommand{\bigsbrk}[1]{\bigl[#1\bigr]}

\newcommand{\Bigsbrk}[1]{\Bigl[#1\Bigr]}

\newcommand{\brc}[1]{\{#1\}}

\newcommand{\bigbrc}[1]{\bigl\{#1\bigr\}}
\newcommand{\Bigbrc}[1]{\Bigl\{#1\Bigr\}}
\newcommand{\abs}[1]{|#1|}
\newcommand{\vev}[1]{\langle #1\rangle}

\newcommand{\dd}{\mathrm{d}}

\newcommand{\defas}{:=}
\newcommand{\safed}{=:}
\newcommand{\gen}{h}
\newcommand{\Std}{\mathrm{Std}}
\newcommand{\Der}{\mathrm{Der}}
\newcommand{\Mons}{\mathrm{Mons}}
\newcommand{\Ext}{\mathrm{Ext}}

\newcommand{\Ring}{\mathcal{R}}
\newcommand{\WeylR}{\mathcal{R}}

\newcommand{\Ideal}{\mathcal{I}}
\newcommand{\Jideal}{\mathcal{J}}
\newcommand{\dimIdeal}{\mathrm{d}}
\newcommand{\Degree}{D}

\newcommand{\eulerInt}{f_\Gamma}
\newcommand{\eulerIntResc}{g_\Gamma}
\newcommand{\Mext}{M_{\Ext}}
\newcommand{\Mstd}{M_{\Std}}
\newcommand{\bsigma}{\eta}

\newcommand{\Integers}{\mathbb{Z}}
\newcommand{\Rationals}{\mathbb{Q}}

\newcommand{\Complex}{\mathbb{C}}
\newcommand{\Field}{\Complex}
\newcommand{\Finite}{\mathbb{F}}

\newcommand{\soft}[1]{\textsc{#1}}
\newcommand{\mathematica}{\soft{Mathematica}}

\newcommand{\code}[1]{\texttt{#1}}

\newcommand{\eq}[1]{\begin{align} #1 \end{align}}
\newcommand{\D}{\mathcal{D}}
\newcommand{\G}{\mathcal{G}}
\newcommand{\e}{\epsilon}
\renewcommand{\d}{\delta}
\newcommand{\estd}{ e^{(\mathrm{Std})} }
\newcommand{\Pstd}{ P^{(\mathrm{Std})} }
\newcommand{\prop}{ \mathrm{D} }
\newcommand{\wtz}{z}

\theoremstyle{plain}
\newtheorem{theorem}{Theorem}[section]
\newtheorem{example}[theorem]{Example}
\newtheorem{lemma}[theorem]{Lemma}
\newtheorem{proposition}[theorem]{Proposition}
\newtheorem{remark}[theorem]{Remark}

\title{Macaulay Matrix for Feynman Integrals: \\
Linear Relations and Intersection Numbers}

\author[a,b]{Vsevolod Chestnov,}
\author[a,b]{Federico Gasparotto,}
\author[b]{Manoj K. Mandal,}
\author[a,b]{Pierpaolo Mastrolia,}
\author[c,d]{\\Saiei J. Matsubara-Heo,} 
\author[a,b]{Henrik J. Munch,}
\author[c]{and Nobuki Takayama}

\newcommand{\unipd}{Dipartimento di Fisica e Astronomia, Universit\`a degli Studi di Padova,
Via Marzolo 8, I-35131 Padova, Italy.}

\newcommand{\pdinfn}{INFN, Sezione di Padova,
Via Marzolo 8, I-35131 Padova, Italy.}

\newcommand{\kobe}{Department of Mathematics, Kobe University,
1-1, Rokkodai, Nada-ku, Kobe 657-8501, Japan.}

\newcommand{\kumamoto}{Faculty of Advanced Science and Technology, Kumamoto University, 2-39-1 Kurokami Chuo-ku Kumamoto
860-8555 Japan.}

\affiliation[a]{\unipd}
\affiliation[b]{\pdinfn}
\affiliation[c]{\kobe}
\affiliation[d]{\kumamoto}

\emailAdd{vsevolod.chestnov@pd.infn.it}
\emailAdd{federico.gasparotto@studenti.unipd.it}
\emailAdd{manojkumar.mandal@pd.infn.it}
\emailAdd{pierpaolo.mastrolia@unipd.it}
\emailAdd{saiei@educ.kumamoto-u.ac.jp}
\emailAdd{henrikjessen.munch@studenti.unipd.it}
\emailAdd{takayama@math.kobe-u.ac.jp}

\abstract{
We elaborate on the connection between  Gel'fand-Kapranov-Zelevinsky systems, de Rham theory for twisted cohomology groups, and Pfaffian equations for Feynman Integrals. We propose a novel, more efficient algorithm to compute Macaulay matrices, which are used to derive Pfaffian systems of differential equations. The Pfaffian matrices are then employed to obtain linear relations for ${\cal A}$-hypergeometric (Euler) integrals and Feynman integrals,
through recurrence relations and through projections by intersection numbers.
}

\definecolor{green1}{HTML}{3D792A}
\definecolor{cyan1}{HTML}{37cdaa}
\definecolor{blue1}{HTML}{5d7ac4}
\definecolor{red1}{HTML}{d0482a}
\definecolor{purple1}{HTML}{845ea8}
\definecolor{orange1}{HTML}{e07229}

\def\henrik#1{({\color{blue} Henrik: #1})}
\def\seva#1{({\color{green1} Seva: #1})}

\def\comment#1{({\color{red} Comments: #1})}

\def\done#1{ }

\newcommand{\namedref}[2]{\hyperref[#2]{#1~\ref*{#2}}}
\newcommand{\secref}[1]{\namedref{Section}{#1}}

\newcommand{\appref}[1]{\namedref{Appendix}{#1}}
\newcommand{\Appref}[1]{\namedref{Appendix}{#1}}

\newcommand{\Algref}[1]{\namedref{Algorithm}{#1}}
\newcommand{\prpref}[1]{\namedref{Proposition}{#1}}
\newcommand{\thmref}[1]{\namedref{Theorem}{#1}}
\newcommand{\exref}[1]{\namedref{Example}{#1}}

\makeatletter
\def\mr@ignsp#1 {\ifx\:#1\@empty\else #1\expandafter\mr@ignsp\fi}%
\newcommand{\multiref}[1]{\begingroup
\xdef\mr@no@sparg{\expandafter\mr@ignsp#1 \: }%
\def\mr@comma{}%
\@for\mr@refs:=\mr@no@sparg\do{\mr@comma\def\mr@comma{,\,}\ref{\mr@refs}}%
\endgroup}
\makeatother
\renewcommand{\eqref}[1]{(\multiref{#1})}

\begin{document}
\maketitle

\section{Introduction}

\subsection{Background}

Within the perturbative approach to quantum as well as classical field theory, the evaluation of multi-loop Feynman integrals is necessary for the determination of scattering amplitudes and related quantities (see \cite{Weinzierl:2022eaz} for a recent review).
The by-now standard evaluation techniques of Feynman integrals (in momentum-space representation) exploit loop-momentum shift invariance to establish \emph{integration-by-parts} (IBP) relations \cite{Chetyrkin:1981qh,Laporta:2001dd} among integrals whose integrands are built from products of the same set of denominators (and scalar products), but raised to different powers.
IBP identities play a crucial role in the calculation of multi-loop integrals because they can be used to identify a minimal set of elements, dubbed \emph{master integrals} (MIs), that can be used as a basis for the decomposition of multi-loop amplitudes.
At the same time, IBP relations can be exploited to build systems of equations solved by the MIs: differential equations \cite{Barucchi:1973zm,KOTIKOV1991158,KOTIKOV1991123,Bern:1993kr,Remiddi:1997ny,Gehrmann:1999as,Henn:2013pwa,Henn:2014qga,Papadopoulos:2014lla,Mastrolia:2014wca}, dimensional recurrence relations \cite{Tarasov:1996br,Lee:2009dh}, and finite difference equations \cite{Laporta:2001dd,Laporta:2003jz} are 
{\it linear relations} 
emerging from the IBP reduction of integrals which spring 
from the action of special (polynomial and differential) operators
on the integrands of the MIs.
Solving such equations amounts to the actual determination of the MIs themselves, as an alternative to direct integration.

The derivation of IBP-decomposition formulas requires the solution of a system of linear relations, generated by imposing that integrals of a total differential (w.r.t.\,integration variables) vanish on the integration boundary \cite{Laporta:2001dd,Larsen:2015ped}.
For multi-loop, multi-scale scattering amplitudes, solving the system of IBP relations may, however, represent a formidable task.
This problem has motivated the development of important techniques based on reconstruction of rational functions using finite fields \cite{vonManteuffel:2014ixa,Peraro:2016wsq,Klappert:2019emp,Peraro:2019svx}.

Already during the early developments of $S$-matrix theory, it was recognized that topology and cohomological methods offer a connection between analytic properties of Feynman integrals and the geometry of their singularity structure, which emerge from  
 the varieties associated to their graph polynomials. 
 In more recent studies, {\it intersection numbers} for twisted de Rham cohomology groups \cite{cho1995,matsumoto1994,matsumoto1998,OST2003,doi:10.1142/S0129167X13500948,aomoto2011theory,yoshida2013hypergeometric,goto2015,goto2015b,Yoshiaki-GOTO2015203,Mizera:2017rqa,matsubaraheo2019algorithm}
 have been exploited to uncover 
 the vector space structure of Feynman integrals and to derive novel algorithms for the direct {\it projection} of multi-loop integral onto MIs \cite{Mastrolia:2018uzb,Frellesvig:2019kgj,Mizera:2019gea,Frellesvig:2019uqt,Frellesvig:2020qot}
 (see also \cite{Mizera:2020wdt,Frellesvig:2021vem,Mandal:2022vok,Mastrolia:2022tww}).
 Important developments have been carried out in 
  \cite{Mizera:2019vvs,Weinzierl:2020xyy,Kaderli:2019dny,Weinzierl:2020nhw,Chen:2020uyk,Caron-Huot:2021xqj,Caron-Huot:2021iev}
 and example of applications to non-linear relations can be found in \cite{fresn2020quadratic,fresn2020quadratic2}.
 See \cite{Cacciatori:2021nli,Weinzierl:2022eaz,Mastrolia:2022tww,Abreu:2022mfk}, for recent reviews. \\
Within this approach, the number of MIs corresponds naturally to the dimension of the vector space of Feynman integrals,
and can be related to topological quantities such as the dimension of the cohomology groups, the
number of certain critical points \cite{Lee:2013hzt,Mastrolia:2018uzb}, Euler characteristics \cite{Aluffi:2009aa,Aluffi:2011ep,Bitoun:2017nre,Bitoun:2018afx,Frellesvig:2019uqt}, as well as to the dimension of quotient rings of polynomials for zero dimensional ideals (Shape lemma)~\cite{Frellesvig:2020qot}.
Linear relations (for integrals with shifted indices), Pfaffian systems of differential equations, finite difference equations, as well as quadratic relations, such as Riemann twisted period relations, can be derived by means of intersection numbers for differential (twisted) forms. \\ 

The pivotal role that twisted de Rham theory seems to have in controlling the algebra of Feynman integrals, and more generally of Euler integrals, has been stimulating us to elaborate on the isomorphism between cohomology groups, whose elements are differential forms, and ${\cal D}$-modules, whose elements are partial differential operators in a Weyl algebra. 

Within momentum-space representation, differential operators acting on multi-loop integrands and integrals are used to establish IBP relations, system of partial differential equations (SPDE) for MIs, and Lorentz invariance identities \cite{Gehrmann:1999as}.
In the former type of relations, differentiation is carried out w.r.t.\,the integration momenta, while for the latter two ones, w.r.t.\, the momenta of external particles and masses. 
Also, within parametric representations of Feynman integrals, it is possible to study the action of differential operators acting on multi-fold integrands and integrals, distinguishing between the cases for which the differentiation is carried out w.r.t.\,integration or external variables.

In \cite{Bitoun:2017nre,Bitoun:2018afx}, it was shown that linear relations between Feynman integrals having shifted indices, similar to IBPs, can arise from the action of parametric annihilators of the integrand, within the Lee–Pomeransky (LP) representation.
The algebraic properties of these special partial differential operators, w.r.t.\,the integration variables, are
derived by studying the ideals of parametric annihilators within the language of $\D$-modules.
Remarkably, the number of master integrals, known to be finite \cite{Smirnov:2010hn}, was found to be identical to the Euler characteristic of the complement of the hypersurface determined by the LP polynomial \cite{Bitoun:2017nre,Bitoun:2018afx}.

While the previous study focused on $\D$-module theory for differential operators in the {\it internal} variables, 
in this work, we investigate the properties of differential operators in the {\it external} variables.
In particular, we exploit the properties of GKZ-hypergeometric systems, 
introduced by Gel'fand, Kapranov, Zelevinsky \cite{GKZ-1989},
also known as ${\cal A}$-hypergeometric systems, to derive Pfaffian equations for the generators of the corresponding $\D$-module. 
As established in \cite{GKZ-Euler-1990}, GKZ systems provide a dictionary between Euler integrals and differential operators, which was later made fully algorithmic in \cite{Matsubara-Heo-Takayama-2020b}.
Feynman integrals can be considered as special cases of Euler integrals, therefore the Pfaffian matrices derived in the context of $\D$-module theory correspond to the matrices of the SPDE satisfied by the MIs \cite{Nasrollahpoursamami:2016}. 

The connection between Feynman integrals and the generalized hypergeometric functions was first proposed by Regge~\cite{Regge:1968rhi} and it was better understood with the knowledge of the Feynman integrals satisfying a holonomic differential equation~\cite{Kashiwara:1977nf}, where the singularities of the differential equation were governed by the Landau singularities. The relation between Feynman integrals and hypergeometric system were studied in~\cite{Hosono:1993qy, Hosono:1994ax,  Fleischer:2003rm, Kalmykov:2008gq, Vanhove:2018mto, Kalmykov:2020cqz,
Blumlein:2021hbq} (and reference therein).
The isomorphism between the GKZ system and the Feynman integral was established in~\cite{Nasrollahpoursamami:2016} and later realized within the 
LP representation in~\cite{delaCruz:2017zqr, Klausen:2019hrg}. \\

 Within the proposed approach, many problems, including the derivation of differential equations, are translated into a ring theoretic computation where we may utilize various notions of computational ring theory such as {\it Gr\"obner bases}.
 As Gr\"obner basis calculations 
 may be computationally expensive,
 we propose a different approach based on the {\it Macaulay matrix}, which is of central interest in this work. 
 Usually, relations among integrals are employed to derive systems of partial differential equations. Instead, in the current work, we reverse the perspective, and show how
 Pfaffian matrices, built from Maculay matrices, can be used to derive linear relations for integrals. Moreover, through the {\it secondary equation} introduced in \cite{matsubaraheo2019algorithm}, 
 we exploit the role of Pfaffian matrices for building {\it cohomology intersection matrices} to be used in the {\it master decomposition formula} presented in \cite{Mizera:2019gea,Mizera:2019vvs,Mizera:2020wdt,Mastrolia:2018uzb, Frellesvig:2019kgj,Frellesvig:2019uqt}, for the direct integral decomposition. \\ 

\subsection{Macaulay matrix for GKZ system}

GKZ systems can be rewritten into systems of differential equations dubbed Pfaffian systems.
By extension, the same holds true for Feynman integrals.
A Pfaffian system for a holonomic function $f$ in variables $z=(z_1, \ldots, z_N)$ is of the form
\begin{equation} \label{eq:pfaffianeq}
   \pd{i} F = P_i \cdot F , \quad F=(s_1 f, \ldots, s_r f) , \quad i=1,\ldots,N
   \ ,
\end{equation}
where $\pd{i} \defas \partial/\partial z_i$, $s_i$ are differential operators acting on $f$, and $P_i$ are $r \times r$ matrices with entries being rational functions of $z$.
The $P_i$ are called {\it Pfaffian matrices} (or simply {\it Pfaffians}) in this paper.
They satisfy the integrability condition
\begin{equation} \label{eq:CI}
    \pd{j}P_i + P_i \cdot P_j = \pd{i} P_j + P_j \cdot P_i \, .
\end{equation}

Note that a Gr\"obner basis for the annihilating ideal $\mathcal{I}$ of the function $f$ gives a Pfaffian system.
Conversely, a Pfaffian system, for instance obtained by Macaulay matrices, gives a Gr\"obner basis for the annihilating ideal: we have $\pd{i} e - P_i e \equiv 0 \ {\rm mod} \ \mathcal{I}$, $e \defas (s_1, \ldots, s_r)$.
Thus, computing the Macaulay matrix naturally leads to a Gr\"obner basis, as alternative to the Buchberger algorithm.

The Macaulay matrix is a generalization of the Sylvester matrix.
The method was regarded as less efficient than the Buchberger algorithm in the ring of polynomials until J.C.Faug\'ere proposed the algorithm F4 in 1999 \cite{Faugere-1999}.
F4 constructs Gr\"obner bases in the ring of polynomials by a variation of Macaulay matrix together with symbolic preprocessing of $S$-pairs.
This algorithm is extremely efficient when combined with linear algebra
optimization methods, for which reason the Macaulay matrix method is currently attracting more attention%
\footnote{
    Note that there is no unique definition of a Macaulay matrix:
    by a Macaulay matrix, we usually mean variations of the matrix used by Macaulay in \cite{Macaulay-1903}.
}.

There are several advantages of utilizing Macaulay matrix
in the ring of differential operators.
One is that once we construct a relevant Macaulay matrix, we do not need expensive computation in a non-commutative ring.
Moreover, in applications to GKZ-hypergeometric systems, standard monomials can be obtained \emph{without} computing Gr\"obner bases in the ring of differential operators \cite{Hibi-Nishiyama-Takayama-2017}.
Hence, we may assume that a set of standard monomials is given before constructing a Macaulay matrix.
These advantages were utilized in \cite{Ohara-Takayama-2015}, where they used a Macaulay matrix to evaluate normalizing constants for a statistical distribution by the holonomic gradient method.
This calculation was not possible with the original holonomic gradient method which utilizes Gr\"obner bases in the ring of differential operators (see e.g. \cite{url-hgm-refs}).

In this article, we propose a method for constructing {\it smaller Macaulay} matrices in the ring of differential operators compared to \cite{Ohara-Takayama-2015}.
Using smaller Macaulay matrices, one can construct larger Gr\"obner bases or Pfaffian systems.
By means of smaller Macaulay matrices, we provide new methods for computing recurrence relations, which may be called IBP identities for GKZ-hypergeometric systems.
In this setting, Feynman integrals are viewed as Mellin transforms of a certain (graph) polynomial $\G$ in Schwinger parameters, where the coefficients in front of the monomials of $\G$ are lifted to independent, indeterminate variables $z = (z_1, \ldots, z_N)$
\cite{delaCruz:2019skx,Klausen:2019hrg,Klausen:2021yrt,Tellander:2021xdz,Feng:2019bdx,Nasrollahpoursamami:2016,Vanhove:2018mto,Weinzierl:2022eaz}. \\

\noindent
The main results of our investigation can be summarised in three points:
\begin{itemize}
    \item 
    We obtain Pfaffian systems of partial differential equations for GKZ-hypergeometric functions, proposing a novel and more efficient algorithm to obtain it from Macaulay matrices, therefore improving the algorithm of~\cite{Ohara-Takayama-2015}. 

\item We use the Pfaffians generated with the Macaulay matrices to obtain linear (contiguity) relations for the GKZ systems, using the matrix factorial, within the holonomic gradient method~\cite{Tachibana-Goto-Koyama-Takayama-2020}; 

\item 
We exploit the rational solution of the secondary equation for intersection matrices, controlled by the Pfaffians~\cite{matsubaraheo2019algorithm}, to derive a novel version of the integral decomposition formula \cite{Mizera:2019gea,Mastrolia:2018uzb, Frellesvig:2019kgj,Frellesvig:2019uqt} based on intersection theory.
\end{itemize}

\noindent
Our results apply to Feynman integrals, and, more generally to 
GKZ-hypergeometric functions. \\

Our presentation is organized as follows.
In \secref{sec:gkz-system}, we review basic notions of the GKZ
hypergeometric systems and their Euler integral representation.
In \secref{sec:pfaffian-systems}, we discuss Pfaffian systems of
differential equations, which are intimately related to GKZ systems.
We present the Macaulay matrix algorithm, based only on linear algebra, to compute Pfaffian matrices in \secref{sec:pfaffian-from-macaulay}.
We show its application to examples of differential equations for
Feynman integrals in \secref{sec:macaulay-feynman}.
In \secref{sec:linear-relations}, we show how Pfaffians can be used to derive linear relations for GKZ systems, similar to IBP identities for Feynman integrals.
Finally, in \secref{sec:intersection}, we present the integral decomposition via intersection numbers, using Pfaffians to compute the required intersection matrices.

All algorithms in this paper are implemented in
the computer algebra system \soft{Risa/Asir} \cite{url-asir}, \soft{Maple} \cite{url-maple} and
\soft{Mathematica} \cite{url-Mathematica} with \soft{FiniteFlow} \cite{url-FiniteFlow},
while the calculations involving Feynman integrals are checked with \soft{LiteRed} \cite{Lee:2012cn, Lee:2013mka}.
Programs used in this paper and machine readable data can be obtainable
from \cite{url-mm-data}.

\section{GKZ hypergeometric systems}
\label{sec:gkz-system}
\done{\seva{In this section we \ldots}}
In this section, we briefly review some basic properties of the
GKZ-hypergeometric systems to fix our notation.
\secref{subsec:Integral representation} introduces a particular integral
representation related to the GKZ systems we work with,
and \secref{subsec:GKZ and de Rham} covers its relation to the algebraic de Rham cohomology groups.
In \secref{ssection:ints_in_D-mod}, we describe how to represent a cohomology class by an element of Weyl algebra.
Finally, in \secref{ssec:rescaling}, we discuss the homogeneity property of GKZ
systems, which allows us to reduce the number of independent variables.

\subsection{Integral representation of GKZ-hypergeometric system}\label{subsec:Integral representation}

In this work, we consider Euler integrals of the form
\eq{
    \label{f_Gamma(z)}
    f_\Gamma(z) = \int_\Gamma g(z;x)^{\beta_0} \, x_1^{-\beta_1} \cdots x_n^{-\beta_n} \, \frac{\dd x}{x}
    \quad , \quad
    \frac{\dd x}{x} \defas \frac{\dd x_1}{x_1} \wedge \cdots \wedge \frac{\dd x_n}{x_n} \, .
}
Here $\Gamma$ is a {\it twisted cycle}%
\footnote{
    A twisted cycle is an integration contour with no boundary, along which the branch of the integrand is specified. For details, see \cite[Chapter 3]{aomoto2011theory}
},
$\beta = (\beta_0, \ldots, \beta_n) \in \mathbb{C}^{n+1}$ are complex parameters, and $g(z;x)$ is a Laurent polynomial in $x$
\eq{
    \label{g(z;x)}
    g(z;x) = \sum_{i=1}^N z_i \, x^{\alpha_i} \, .
}
The monomials above are written in multivariate exponent notation:
given an integer vector $\alpha_i \in \mathbb{Z}^n$ we set
\begin{align}
    x^{\alpha_i} \defas 
    x_1^{\alpha_{i,1}} 
    \cdots 
    x_n^{\alpha_{i,n}} \ ,
    \label{eq:multivar-exp}
\end{align}
where $\alpha_{i,j}$ stands for the $j$-th component of the vector $\alpha_i$.
Crucially, in \eqref{g(z;x)} we regard each coefficient $z_i$ as an {\it independent} variable of $f_\Gamma(z)$.

Let us construct the $(n+1) \times N$ matrix
\eq{
    \label{A_mat}
    A = \begin{pmatrix} a_1 & \ldots & a_N \end{pmatrix} \ ,
}
whose columns $a_i$ are built from the monomial exponents
$\alpha_i$ as $a_i := (1, \alpha_i)$, 
with the assumption that  $\text{Span}\{a_1,\ldots,a_N\} = \ZZ^{n+1}$.
Moreover, we introduce the (left) kernel of $A$, defined as,
\eq{
    \text{Ker}(A) =
    \big\{
    u = (u_1,\ldots,u_N) \in \mathbb{Z}^N \, \big | \, {\sum_{j=1}^N u_j \, a_j} = \mathbf{0}
    \big\} \, .
}
Then, by using $A$ and $\beta$ as input, 
we build the following set of differential operators:
\begin{alignat}{2}
    \label{E_j}
    E_j &=
    \sum_{i=1}^N a_{j,i} \, z_i \, \frac{\partial}{\partial z_i} - \beta_j
    \, , \quad
    && j= 1, \ldots, n+1 
    \\
    \Box_u &=
    \prod_{u_i > 0} \left( \frac{\partial}{\partial z_i} \right)^{u_i} -
    \prod_{u_i < 0} \left( \frac{\partial}{\partial z_i} \right)^{-u_i}
    \, , \quad
    && \forall u \in \text{Ker}(A) \, .
    \label{Box_u}
\end{alignat}
The function $f_\Gamma(z)$, 
defined in \eqref{f_Gamma(z)}, satisfies the system of partial differential equations (PDE)
\eq{
    E_j \, f_\Gamma(z) &= 0
    \ ,
    \\
    \Box_u \, f_\Gamma(z) &= 0
    \ ,
}
therefore it is dubbed an ${\cal A}$-hypergeometric function \cite{GKZ-1989}.

\subsection{GKZ \texorpdfstring{$\D$}{}-modules and de Rham cohomology}\label{subsec:GKZ and de Rham}

The operators in \eqref{E_j}-\eqref{Box_u} can be regarded as elements of a
Weyl algebra
\eq{
    \D_N = \mathbb{C}[z_1,\ldots,z_N]\langle \partial_1,\ldots,\partial_N\rangle
    \quad , \quad
    [\partial_i,\partial_j] = 0 \quad , \quad [\partial_i,z_j] = \delta_{ij} \, .
}
In multivariate exponent notation, the elements of $\D_N$ take the form $\sum_{k \in K} h_k(z) \partial^k$ for some finite collection of sets $K = \{K_i \in \mathbb{N}_0^{N} \}_i $, where the $h_k(z)$ are polynomials in $z$ with complex coefficients.
The symbol $\pd{i}$ is an alias of $\frac{\partial}{\partial z_i}$.

We introduce the \emph{GKZ system} as the left $\D_N$-module $\D_N/H_A(\beta)$, where
$H_A(\beta)$ 
is the left ideal generated by
$E_j$ and $\Box_u$,
\eq{
    \label{HA(beta)}
    H_A(\beta) =
    \sum_{j=1}^{n+1} \D_N \cdot E_j + \sum_{u \in \text{Ker}(A)} \D_N \cdot \Box_u \, .
}
Further details on $\D$-modules theory can be found in the \Appref{sec:D-modules-appendix}.

Let us list a few important properties of GKZ systems and their relations
to de Rham cohomology groups, which are expressed through the following
theorems and propositions. \\

\noindent
$\square$
First, we recall a theorem on the number of solutions to GKZ systems.
Let $\Delta_A$ denote the convex polytope spanned by the columns of $A$.
We say $\beta$ is \emph{non-resonant} when it does not belong to any set of the form
${\rm span}_{\C}\{a_i\mid a_i\in F\}+\ZZ^{n+1}$ where $F$ is a facet of $\Delta_A$. \\

\noindent
\begin{minipage}{11.5cm}
For example, if we take a $2\times 2$ matrix
$$A=\left(
\begin{array}{cc}
1 & 1 \\
0 & 1 \\
\end{array}
\right) \ ,
$$
the polytope $\Delta_A$ is a segment, and there are exactly two facets,
respectively indicated by the thick segment and the black dots in the nearby picture.
Each facet defines a linear subspace on which $\beta$ is resonant.
\\
\end{minipage}
\begin{minipage}{3cm}
\hspace*{0.5cm}
\setlength{\unitlength}{1.5mm}
\begin{picture}(20,20)
\put(0,10){\line(1,0){20}}
\put(0,0){\line(1,1){20}}
\put(9,8){(0,0)}
\put(15,10){\circle*{1}}
\put(14,8){(1,0)}
\put(15,15){\circle*{1}}
\put(16,13){(1,1)}
\linethickness{0.6mm}
\put(15,10){\line(0,1){5}}
\thinlines
\put(0,0){\circle{1}}\put(0,5){\circle{1}}\put(0,10){\circle{1}}
                     \put(0,15){\circle{1}}\put(0,20){\circle{1}}
\put(5,0){\circle{1}}\put(5,5){\circle{1}}\put(5,10){\circle{1}}
                     \put(5,15){\circle{1}}\put(5,20){\circle{1}}
\put(10,0){\circle{1}}\put(10,5){\circle{1}}\put(10,10){\circle{1}}
                      \put(10,15){\circle{1}}\put(10,20){\circle{1}}
\put(15,0){\circle{1}}\put(15,5){\circle{1}}\put(15,10){\circle{1}}
                      \put(15,15){\circle{1}}\put(15,20){\circle{1}}
\put(20,0){\circle{1}}\put(20,5){\circle{1}}\put(20,10){\circle{1}}
                      \put(20,15){\circle{1}}\put(20,20){\circle{1}}
\end{picture}
\end{minipage}
\begin{theorem}[\cite{Adolphson-1994}]
    \label{thm:Ado}
    \hfill
    \begin{enumerate}

        \item $H_A(\beta)$ is a holonomic ideal\,%
        \footnote{
            For the definition of a holonomic ideal, see \Appref{sec:D-modules-appendix} or p. 31 of \cite{SST}.
        }

        \item When $\beta$ is non-resonant, the holonomic rank $r$ of $H_A(\beta)$ is given by the volume%
        \footnote{
            ${\rm vol}$ stands for the Lebesgue measure and can be calculated with software such as \texttt{Polymake} \cite{polymake:2000}. The holonomic rank is the number of
            standard monomials of $\Ring H_A(\beta)$ \brk{see \Appref{sec:D-modules-appendix}}.
        }
        \eq{
            \label{holonomic_rank}
            r = n! \cdot {\rm vol}(\Delta_A) \, .
        }

    \end{enumerate}
\end{theorem}

\noindent
The holonomic rank equals the number of independent solutions to the system of PDEs \eqref{E_j}-\eqref{Box_u} at a generic point $z\in\mathbb{C}^N$.
The first statement ensures that the rank is finite, while the second statement gives an exact formula for computing it in terms of combinatorial data. \\

\noindent
$\square$
Next, letting $\mathbb{G}_m$ (resp. $\mathbb{A}$) stand for the complex torus (resp. complex Affine line) equipped with the Zariski topology%
\footnote{
    As a set, the torus $\mathbb{G}_m$ (resp. the complex Affine line $\mathbb{A}$) is equivalent to $\mathbb{C}^\ast := \mathbb{C} \setminus \{0\}$ (resp. $\mathbb{C}$).
}
and
\eq{
    X \defas
    \big\{(z,x)\in\mathbb{A}^N \times (\mathbb{G}_m)^n    \big| \, g(z;x) \neq 0 \big\}
    \quad , \quad
    Y \defas
    \mathbb{A}^N \, ,
}
we denote by $\pi : X \ \to \ Y$ the natural projection from the space of GKZ and integration variables to the space of GKZ variables only.

Setting $\mathcal{O}(X)\defas \C[z_1,\dots,z_N,x_1^{\pm 1},\dots,x_n^{\pm 1},\frac{1}{g}]$, we define an action of $\D_N$ on
$f=f(z,x) \in \mathcal{O}(X)$ by
\begin{align}
        \frac{\partial}{\partial z_i}\bullet f &=
     \frac{\partial f}{\partial z_i} + \beta_0
     \left(
     {1 \over g(z;x)}
     {\partial g(z;x) \over \partial z_i}
     \right)
     f
    \ ,
     \\
    \frac{\partial}{\partial x_i}\bullet f &=
    \frac{\partial f}{\partial x_i} +
    \beta_0
    \left(
     {1 \over g(z;x)}
     {\partial g(z;x) \over \partial x_i}
    \right)f
    - \beta_i \frac{f}{x_i}
    \ .
\end{align}
The symbol
$\mathcal{O}(X)\, g^{\beta_0}x_1^{-\beta_1}\dots x_n^{-\beta_n}$ denotes the left $\D_N$-module $\mathcal{O}(X)$ endowed with this action.
Formally, we have the identities
\begin{eqnarray}
&&
\frac{\partial}{\partial z_i} \Big(g^{\beta_0} \, x_1^{-\beta_1}\dots x_n^{-\beta_n} \,
f
\Big) = g^{\beta_0} \, x_1^{-\beta_1}\dots x_n^{-\beta_n}
 \left(\frac{\partial}{\partial z_i}\bullet f
 \right)
    \ ,
    \\
&&
\frac{\partial}{\partial x_i} \Big(g^{\beta_0} \, x_1^{-\beta_1}\dots x_n^{-\beta_n} f \,
\Big) = g^{\beta_0}\, x_1^{-\beta_1}\dots x_n^{-\beta_n}
 \left(\frac{\partial}{\partial x_i} \bullet f
 \right)
    \ .
\end{eqnarray}
The direct image $\D$-module $\int_\pi\mathcal{O}(X) \, g^{\beta_0} \, x_1^{-\beta_1}\dots x_n^{-\beta_n}$ is defined canonically as in \Appref{sec:D-modules-appendix}.

\begin{theorem}[\cite{GKZ-Euler-1990}]\label{thm:Gel}
Suppose that $\beta$ is non-resonant.
Then there is a canonical isomorphism of left $\D_N$-modules
\begin{equation}
    \D_N / H_A(\beta) \simeq \int_\pi \mathcal{O}(X) \,  g^{\beta_0} \,  x_1^{-\beta_1}\dots x_n^{-\beta_n} \ .
\end{equation}
\end{theorem}

\noindent
Let us make this isomorphism explicit \cite{AdoSpa}.
We let
\eq{
    \Omega_{X/Y}^k =
    \bigoplus_{J \subset \{1, \ldots, n\}, \, |J| = k}
    {\mathcal{O}(X)} \, \dd x^J
}
represent the space of relative $k$-forms%
\footnote{
    $X/Y$ does not denote a quotient space. Rather, it is a symbol for relative differential forms, by which we mean that we do not consider differential forms $\dd z^J$ where $z \in Y$.
}
onto which we act with the covariant derivative in integration variables
\eq{
\label{covariant_derivative_GKZ}
    \nabla_x  =
    \dd_x + \beta_0 \frac{\dd_x g(z;x)}{g(z;x)} \wedge -
    \sum_{i=1}^n \beta_i \frac{\dd x_i}{x_i}  \wedge \, .
}
We then obtain a chain complex
\eq{
    \cdots \overset{\nabla_x}{\longrightarrow} \
    \Omega_{X/Y}^k \
    \overset{\nabla_x}{\longrightarrow} \
    \Omega_{X/Y}^{k+1} \
    \overset{\nabla_x}{\longrightarrow} \cdots \, .
}
The $k$-th \emph{relative de Rham cohomology group} is defined as follows:
\eq{
\label{de_Rham_cohomology_group}
    \mathbb{H}^k &\defas
    \mathrm{Ker}
    \Bigbrk{
        \nabla_x : \, \Omega_{X/Y}^k
        \longrightarrow
        \Omega_{X/Y}^{k+1}
    }
    \ \Big / \
    \mathrm{Im}\Bigbrk{
        \nabla_x : \, \Omega_{X/Y}^{k - 1}
        \longrightarrow
        \Omega_{X/Y}^{k}
    } \ .
}

\noindent
It can be shown that the direct image $\mathcal{D}$-module
$\int_\pi\mathcal{O}(X) \, g^{\beta_0} \, x_1^{-\beta_1}\dots x_n^{-\beta_n}$
is isomorphic to the $n$-th relative de Rham cohomology group $\mathbb{H}^n$, for which reason the latter is a left $\D_N$-module by \thmref{thm:Gel}.
In fact, \thmref{thm:Gel} can be rephrased as
\begin{proposition}\label{prop:D-mod_deRham_iso}
Suppose that $\beta$ is non-resonant.
Then there is a unique isomorphism of $\D_N$-modules
\begin{equation}
    \D_N / H_A(\beta) \simeq \mathbb{H}^n
\end{equation}
such that $[1]\in \D_N/H_A(\beta)$ is sent to 
$\left[ \frac{dx}{x} \right]\in \mathbb{H}^n$ \ .
\end{proposition}
\noindent

A consequence of \prpref{prop:D-mod_deRham_iso}, which will be essential for
our application of $\D_N$-module theory to Feynman integrals, is the following:
given a cohomology class $[\omega(z)] \in \mathbb{H}^n$, there exists a
differential operator $P \in \D_N$, which is unique modulo $H_A(\beta)$, such
that
\eq{
    P \left[ \frac{\dd x}{x} \right] = [\omega(z)] \, .
}
The partial differential operators $\partial_i$ in $P$ act on a cohomology class $[\omega(z)] \in \mathbb{H}^{n}$ via
\eq{
    \partial_i\bullet \, [\omega(z)] =
    \left[
            \partial_i \, \omega(z) + \beta_0 \frac{x^{\alpha_i}}{g(z;x)} \omega(z)
    \right].
    \label{eqn:GM-Der}
}
The action \eqref{eqn:GM-Der} comes from differentiation under the integral sign:
\eq{
\frac{\partial}{\partial z_i}
\int_\Gamma g(z;x)^{\beta_0} \, x_1^{-\beta_1} \cdots x_n^{-\beta_n} \, \omega(z)
=
\int_\Gamma g(z;x)^{\beta_0} \, x_1^{-\beta_1} \cdots x_n^{-\beta_n} \, \partial_i\bullet\omega(z)
\ .
}

Since a Feynman integral can be represented by a cohomology class
\cite{Mastrolia:2018uzb}, \emph{we may equally well consider the operator $P$
as representing that integral}. An algorithm for computing $P$ was developed
in \cite{Matsubara-Heo-Takayama-2020b} and will be outlined in the following
section.
Moreover, in the view of the relation of GKZ-systems and Feynman integrals (to be elaborated on in \secref{sec:macaulay-feynman}), we observe that the finiteness of the rank, established
by the first statement of \thmref{thm:Ado}, can be related to the finiteness of the number of
master integrals \cite{Smirnov:2010hn}.  The formula for its
evaluation, given in the second statement of \thmref{thm:Ado}, offers an alternative way
of determining ${\rm dim}({\mathbb H}^n)$, which is ordinarily computed in terms of Betti
numbers, by counting the number of certain critical points, or by Euler
characteristics - all of which are related to the number of master integrals.

\subsection{Representing integrals inside \texorpdfstring{$\D_N$}{}-modules}
\label{ssection:ints_in_D-mod}

Let us define the following family of differential forms:
\eq{
  \label{omega_q}
  \omega_q = g(z;x)^{-q_0} \, x^{q'} \, \frac{\dd x}{x}
  \quad , \quad
  q = (q_0 , \, q') \in \ZZ \times \ZZ^n.
}
Pairing $\omega_q$ with a cycle $\Gamma$ yields the integral
\eq{
    \label{omega_q_Gamma}
    \langle \omega_q \rangle_\Gamma =
    \int_\Gamma g(z;x)^{\beta_0-q_0} \,  x_1^{-\beta_1+q'_1} \cdots x_n^{-\beta_n+q'_n} \frac{\dd x}{x}
    \ .
}
We observe that differentiating the integral w.r.t $z_i$ shifts the $q$ vector
by $+a_i$ (defined after \eqref{A_mat}):
\eq{
    \label{step_down}
   \partial_i \langle \omega_q \rangle_\Gamma =
   (\beta_0-q_0) \langle \omega_{q+a_i} \rangle_\Gamma \, .
}

It is also possible to construct an operator which does the opposite, i.e. shifting
$q$ by $-a_i$. In \cite{SSTip} it was shown that there exists a so-called
step-up (or creation) operator $U_i$ and polynomial $b_i(\beta)$ such that
\eq{
    U_i \partial_i - b_i(\beta) = 0 \quad \text{mod} \quad H_A(\beta) \, .
}
The operator $U_i$ then acts on the integral \eqref{omega_q_Gamma} via
\eq{
    \label{step_up}
    U_i \langle \omega_q \rangle_\Gamma = b_i(\beta-a_i) \langle \omega_{q-a_i}\rangle_\Gamma
    \ .
}
The polynomial $b_i(\beta)$ is called the $b$-function \cite{SSTip}. Its
computation is algorithmic and is implemented in computer algebra packages such
as \code{mt\_gkz.rr} \cite{Matsubara-Heo-Takayama-2020b}.

Now, assume an integral of the form \eqref{omega_q_Gamma} is given for some
choice of the integer vector $q$. Since the $a_i$'s span $\ZZ^{n+1}$ by
assumption, we may write
\eq{\label{eq:q_sum_definition}
    q = \sum_{i=1}^N r_i \, a_i \, .
}
The coefficients $r_i \in \ZZ$ allow us to balance the shifts by $\pm a_i$ in the $q$ of
$\langle \omega_q \rangle_\Gamma$ via equations \eqref{step_down} and
\eqref{step_up}. We may therefore construct the operator
\eq{
    P(q) = \prod_{r_i<0} U_i^{-r_i} \prod_{r_i>0} \frac{1}{B(\beta)B'(\beta)}
    \> \partial^{r_i}
    \ ,
}
with the property that
\eq{
    \label{P(q)}
    P(q) \langle \omega_0 \rangle_\Gamma =
    \langle \omega_q \rangle_\Gamma
    \quad , \quad
    \omega_0 = \frac{\dd x}{x}
    \ .
}
The two functions $B(\beta)$ and $B'(\beta)$ are built from the prefactors on
the RHSs of \eqref{step_down} and \eqref{step_up}, and their explicit expressions
can be found in \cite{Matsubara-Heo-Takayama-2020b}. The operator $P(q)$ will
be used extensively in our applications to Feynman integrals, in which case the
vector $q$ will be related to the propagator powers and space-time dimension of
a given integral.
Let us observe that the vector $(r_1, \ldots, r_N)$ in \eqref{eq:q_sum_definition} is not unique, because we can shift it by adding a term like $\sum_i c_i u_i$ for some choice of $c_i \in \ZZ$, and $u_i \in \ZZ^N$ span $\text{Ker}(A)$.
This freedom can be exploited to simplify the expression of the operators $P\brk{q}$ appearing in \eqref{P(q)}, and 
for practical purposes, we find it convenient 
to choose $r_i>0$ for all $i$, in order to avoid the appearance of $U_i$ in \eqref{P(q)}, because it may contain monomials in $\partial_i$ of large degree.

\subsection{Homogeneity and integrand rescaling}
\label{ssec:rescaling}
Under rescaling of the $z_i$ variables,
Euler integrals scale as
\eq{
    \label{f_Gamma_homogeneity}
    f_\Gamma(t^{a_1} z_1, \ldots, t^{a_N} z_N) =
    t_0^{\beta_0} t_1^{\beta_1} \cdots t_n^{\beta_n} f_\Gamma(z)
    \ ,
}
where we used the multivariate exponent notation from \eqref{eq:multivar-exp} on the LHS.
Indeed, \eqref{f_Gamma_homogeneity} is equivalent to homogeneity equations
\begin{equation}
E_j f_\Gamma = 0 \ , \quad
j = 1, \dots, n + 1  \ .
\end{equation}
This property can be exploited to work on a simpler set of integrals obtained
from the original definition by {\it freezing} $(n+1)$ out of $N$ variables $z_i$ to e.g. $1$. 
The residual functional dependence is then given by $N - \brk{n + 1}$ ratios of
$z_i$-variables. More details can be found in \Appref{sec:rescaling-details}.

\begin{example}\rm
    Consider the following Euler integral:
    \begin{align}
        \eulerInt\brk{z} = \int_\Gamma \brk{z_1 + z_2 x_1 + z_3 x_2 + z_4 x_1 x_2}^{\beta_0} \, x_1^{-\beta_1} x_2^{-\beta_2} \, \frac{\dd x}{x} \ ,
        \label{eq:euler-int-example}
    \end{align}
    which, according to eq. \eqref{f_Gamma(z)}, corresponds to a $3 \times 4$
    matrix:
    \begin{align}
        A = \lrbrk{
            \begin{array}{cccc}
                1 & 1 & 1 & 1 \\
                0 & 1 & 0 & 1 \\
                0 & 0 & 1 & 1
            \end{array}
        } \ .
        \label{eq:rescale-example-A}
    \end{align}
    The homogeneity property from \eqref{f_Gamma_homogeneity} involves
    three rescaling parameters $\brc{t_0, t_1, t_2}$:
    \begin{align}
        \eulerInt\brk{t_0 \, z_1, t_0 t_1 \, z_2, t_0 t_2 \, z_3, t_0 t_1 t_2 \, z_4}
        = t_0^{\beta_0} t_1^{\beta_1} t_2^{\beta_2} \> \eulerInt\brk{z_1, z_2, z_3, z_4}
        \ .
    \end{align}
    To derive this relation we need to rescale integration variables as
    $x_1 \mapsto x_1 / t_1$ and $x_2 \mapsto x_2 / t_2$.
    The homogeneity property allows us to factorize the dependence on 
    $z_1$, $z_2$, and $z_3$ by the following choice of the 
    rescaling parameters $t_i$:
    \begin{align}
        t_0 = 1 / z_1 \ , \quad
        t_1 = z_1 / z_2 \ , \quad
        t_2 = z_1 / z_3 \ .
    \end{align}
    The original Euler integral from \eqref{eq:euler-int-example} then
    becomes
    \begin{align}
        \eulerInt\brk{z_1, z_2, z_3, z_4}
        = z_1^{-\beta_0 + \beta_1 + \beta_2} z_2^{-\beta_1} z_3^{-\beta_2} \>
        \eulerInt\brk{1, 1, 1, w} \ , \quad w = \tfrac{z_1 \, z_4}{z_2 \, z_3}
        \ ,
        \label{eq:rescale-example}
    \end{align}
    and so we have effectively rescaled away $\brk{n + 1}$ of the $z_i$
    variables.

    Finally, we show the differential equation obeyed by the rescaled Euler
    integral $\eulerInt\brk{1, 1, 1, z}$. The generator \eqref{Box_u} of the
    original GKZ system with the $A$-matrix shown in
    \eqref{eq:rescale-example-A} reads as
    \begin{align}  
        \Box_u \eulerInt =
        \big(\pd{1} \pd{4} - \pd{2} \pd{3} \big) \>
        \eulerInt\brk{z_1, z_2, z_3, z_4} = 0
        \ .
        \label{eq:rescale-box}
    \end{align}
    Upon substitution of the representation \eqref{eq:rescale-example}
    on the RHS, we can determine what the two terms in $\Box_u$ map to when changing variables to $w$:
    \begin{align}
        \pd{1} \pd{4} &\longrightarrow -w^{-1}
        \brk{\beta_0 - \beta_1 - \beta_2 - w \, \pd{w}} \brk{w \, \pd{w}}
        \ ,
        \\
        \pd{2} \pd{3} &\longrightarrow
        \brk{\beta_1 + w \, \pd{w}} \brk{\beta_2 + w \, \pd{w}}
        \ .
    \end{align}
    These expressions can also be derived from \prpref{prp:simplex}. Thus the
    original equation \eqref{eq:rescale-box} becomes
    \begin{align}
        \Bigsbrk{
            w \brk{1 - w} \, \pd{w}^2
            + \bigbrk{1 - \beta_0 + \beta_1 + \beta_2 - \brk{1 + \beta_1 + \beta_2}w} \, \pd{w}
            - \beta_1 \beta_2
        } \> \eulerInt\brk{1, 1, 1, w} = 0 \, ,
    \end{align}
    which we recognize as the differential equation for the Gau\ss{}
    hypergeometric function.
\end{example}

This concludes our short overview of the basic properties of GKZ systems and
Euler integrals. In the following section, we describe how to obtain a Pfaffian system of first-order PDEs given a basis of Euler integrals.

\section{Pfaffian systems}
\label{sec:pfaffian-systems}
In this section, we introduce notation related to Pfaffian systems to be used in the rest of this work.
In \secref{ssec:D-module to DEQ}, we recall the passage from a holonomic $\mathcal{D}_N$-module to a Pfaffian system, while in \secref{ssec:basis-change} and \secref{ssec: basis change without derivatives} we present two methods for basis change.

For any positive integer $N$, we write $\Field\brk{z} = \Field\brk{z_1, \dots, z_N}$
for the field of rational functions in variables $z_1,\dots,z_N$.
Let $\mathcal{R}:=\mathbb{C}(z_1,\dots,z_N)\otimes_{\mathbb{C}[z_1,\dots,z_N]}\mathcal{D}_N$ be the rational Weyl algebra%
\footnote{
    The subscript of the tensor product $\otimes$ denotes the ring over
    which the product is defined. For example, it implies the identity:
    $r \otimes p \, Q = r \, p \otimes Q$, for
    $r \in \Field\brk{z}$, $p \in \Field\sbrk{z}$, and $Q \in \D_N$
    .
};
note that $\Ring$ also has the structure of a non-commutative ring.
Any element $Q \in \Ring$ can be uniquely written in the so-called
normally ordered form
with all the partial derivatives commuted to the right:
\begin{equation}
    Q = \sum_{k} \, q_k \cdot \partial^k \, , \quad
    q_k \in \Field\brk{z} \ .
    \label{eq:normal-form}
\end{equation}

As we review in \appref{sec:D-modules-appendix}, given a Gr\"obner basis $G$ of a left ideal $\Ideal$ in $\Ring$, we denote by%
\footnote{When considering $\Std$ as a basis, we instead write $\estd$.}
$\Std$ the standard monomials w.r.t $G$ and a term order $\prec$.
For our purposes it is enough to consider a zero-dimensional left ideal $\Ideal$ generated by a finitely many elements:
\begin{align}
    \Ideal \defas \vev{\gen_1,\gen_2,\dots,\gen_\dimIdeal} \ , \quad \gen_i \in \Ring
    \ .
    \label{eq:ideal-def}
\end{align}
The ideal $\Ideal$ is zero-dimensional iff the corresponding set of standard monomials
$\Std$ is finite \brk{see \cite[\S 6]{dojo}}.

\subsection{From \texorpdfstring{$\mathcal{D}_N$}{}-module to DEQ system}\label{ssec:D-module to DEQ}

We consider a left $\mathcal{R}$-module $M$
which is finite dimensional over $\mathbb{C}(z)$.  Such a left
$\mathcal{R}$-module arises as an extension
$\mathcal{R}\otimes_{\mathcal{D}_N}M^\prime$ for some other holonomic
$\mathcal{D}_N$-module $M^\prime$.
When $\{u_1,\dots,u_r\}\subset M$ is a
$\mathbb{C}(z)$-basis, one can find an $r\times r$ matrix $P_j(z)$ with entries
in $\mathbb{C}(z)$ such that
\begin{equation}\label{eqn:Pfaffian}
    \partial_j (u_1,\dots,u_r)^T = P_j(z) \cdot (u_1,\dots,u_r)^T \, .
\end{equation}
We call $P_j(z)$ the Pfaffian matrix in direction $z_j$.
Let us now consider the GKZ $\mathcal{D}_N$-module.
Namely, we consider a left $\mathcal{R}$-module $M=\mathcal{R}\otimes_{\mathcal{D}_N} (\D_N / H_A(\beta))$, where $H_A(\beta)$ is as in \eqref{HA(beta)}.
In view of \prpref{prop:D-mod_deRham_iso}, each $u_i$ corresponds to a
cohomology class $[\omega_i]\in\mathbb{H}^n$.
Integrating \eqref{eqn:Pfaffian} over a cycle $\Gamma$, we obtain an identity
\begin{equation}
    \partial_j
    \big( \langle \omega_1\rangle_\Gamma,\dots,\langle \omega_r\rangle_\Gamma \big)^T
    = P_j(z) \cdot
    \big( \langle \omega_1\rangle_\Gamma,\dots,\langle \omega_r\rangle_\Gamma \big)^T,
\end{equation}
which is precisely the system of differential equations that the master Euler
integrals are subject to.
A standard method of computing the Pfaffian matrix $P_j(z)$ is to use a
Gr\"obner basis of $\mathcal{R}$ (\cite[Chapter 6]{dojo}).
Unfortunately, this approach is far from effective.  We develop a different
approach based on the Macaulay matrix in \secref{sec:pfaffian-from-macaulay}.

Our construction of Pfaffian matrices by the Macaulay matrix method, to be given in \secref{sec:pfaffian-from-macaulay}, will assume that a set of standard monomials $S$ for GKZ ideal $\mathcal{R}H_A(\beta)$ is given.
Note that the set $S$ gives rise to a $\CC(z)$-basis of $M$.
Although one usually computes a Gr\"obner basis of the ideal $\mathcal{R}H_A(\beta)$ to find the set $S$, we claim that the set $S$ can be found by computing a Gr\"obner basis of an ideal in a \underline{commutative} subring of $\mathcal{R}$.
This reduces the cost of computation in a significant way.

For this purpose, we fix a term order $\prec$ on the commutative ring $\CC[\pd{1}, \ldots, \pd{N}]$.
It naturally induces a term order $\prec$ on the ring $\mathcal{R}$ because a term order on the ring $\mathcal{R}$ is a total order on the set of monomials $\pd{}^k$ (cf. \cite[\S6.1]{dojo}).
We denote by $\mathrm{in}_\prec(\Ideal_\Box)$
the ideal of $\mathbb{C}[\pd{1},\ldots,\pd{N}]$ generated by initial terms of $\Ideal_\Box = \langle \Box_u \rangle$ by the order $\prec$, where the box operators are given in \eqref{Box_u}.
Any generator of $\mathrm{in}_\prec(\Ideal_\Box)$ can be written as
$\pd{}^k=\prod_{i=1}^N \pd{i}^{k_i}$.
Let us write $\theta_i$ for the Euler operator $z_i\pd{i}$ and consider a subring of $\mathcal{R}$ generated by $\theta_1,\dots,\theta_N$ over $\CC$.
We will denote this ring by $\CC[\theta]$, which is clearly a commutative ring.
By the correspondence $\theta^k\leftrightarrow\partial^k$, we equip a term order $\prec$ on the ring $\CC[\theta]$ induced from $\prec$.
We then define the {\it distraction} of the element $\pd{}^k$ (\cite[p.68]{SST}) by
\begin{equation}
  \prod_{i=1}^N \theta_{i} (\theta_{i}-1) \cdots (\theta_{i}-k_i+1)
    \ .
\end{equation}
We consider the ideal $\Ideal'$ of the ring $\CC[\theta]$ generated by the distractions of the generators of $\mathrm{in}_\prec(\Ideal_\Box)$ and the $E_1,\dots,E_{n+1}$ from \eqref{E_j}.
The following theorem is essential whose proof is based on a technique called {\it Gr\"obner deformation} as we will see in \appref{sec:D-modules-appendix}.

\begin{theorem}\label{thm:3.1}
The set of standard monomials of $\Ideal'$ coincides with that of GKZ ideal $\mathcal{R}H_A(\beta)$ when $\beta$ is generic.
More concretely, if $\{ \theta^{a_1},\dots,\theta^{a_r}\}$ is a set of standard monomials of $\Ideal'$, the set $\{ \pd{}^{a_1},\dots,\pd{}^{a_r}\}$ is a set of standard monomials of $\mathcal{R}H_A(\beta)$.
\end{theorem}
The ideal $\Ideal_\Box$ is known as {\it toric ideal}.
Computation of a Gr\"obner basis of such an ideal is effective as it is generated by binomials.
Thus, the computation of a set of generators of the ideal $\Ideal'$ is efficient.
Running Buchberger's algorithm for the commutative ideal $\Ideal'$, we obtain a set of standard monomials of GKZ ideal $\mathcal{R}H_A(\beta)$ by \thmref{thm:3.1}.
We emphasize that the process explained above does not involve any computation in the non-commutative ring $\mathcal{R}$.
Finally, we remark that the process above is implemented as a function \texttt{cbase\_by\_euler} in the \texttt{Risa/Asir} package \texttt{mt\_gkz.rr}.

\begin{example}
    Let us consider a matrix
$$A=\left(\begin{array}{cccc}
 1&1&1&1 \\
 0&1&3&4 \\
\end{array}\right).
$$
Let $\prec$ be the lexicographic order such that $\pd{1} \succ \cdots \succ \pd{4}$.
The ideal ${\rm in}_\prec(\Ideal_\Box)$ is generated by monomials
$$ \pd{2}\pd{4}^2, \pd{1}\pd{4}, \pd{1}\pd{3}^2, \pd{1}^2\pd{3}. $$
The distractions of them are
$$ \theta_{2}\theta_{4} (\theta_{4}-1), \theta_{1}\theta_{4}, \theta_{1}\theta_{3}(\theta_{3}-1), 
\theta_{1}(\theta_{1}-1)\theta_{3}. $$
We compute the Gr\"obner basis of an ideal generated by these distractions and
$ \theta_{1} + \theta_{2} + \theta_{3} + \theta_{4} - b_1$ and $
\theta_{2} + 3 \theta_{3} + 4\theta_{4} - b_2
$ with respect to $\succ$.
The set of the standard monomials is
$\{ \theta_{3}, \theta_{4}^2, \theta_{4}, 1 \}$.
Thus, the set of the standard monomials of GKZ ideal $\mathcal{R}H_A(\beta)$ is given by 
$$\{ \pd{3}, \pd{4}^2, \pd{4}, 1 \}.$$
\end{example}

\if
Our construction of Pfaffian matrices by the Macaulay matrix method, to be given in \secref{sec:pfaffian-from-macaulay}, will assume that a set of standard monomials is given.
We briefly summarize the result of \cite{Hibi-Nishiyama-Takayama-2017},
which gives an algorithm to effectively obtain standard monomials of the GKZ ideal $\Ideal = \mathcal{R} H_A(\beta)$.
We fix a term order $\prec$ in the commutative ring $\CC[\pd{1}, \ldots, \pd{N}]$.
We denote by $\mathrm{in}_\prec(\Ideal_\Box)$
the ideal of $\mathbb{C}[\pd{1},\ldots,\pd{N}]$ generated by initial terms of $\Ideal_\Box = \langle \Box_u \rangle$ by the order $\prec$, where the box operators are given in \eqref{Box_u}.
Any generator of $\mathrm{in}_\prec(\Ideal_\Box)$ can be written as
$\prod_{i=1}^N \pd{i}^{k_i}$.
We then define the {\it distraction} of the generator \cite[p.68]{SST} by
\begin{equation}
  \prod_{i=1}^N \pd{i} (\pd{i}-1) \cdots (\pd{i}-k_i+1)
    \ .
\end{equation}
We consider the ideal $\Ideal'$ generated by the distractions of the generators
of $\mathrm{in}_\prec(\Ideal_\Box)$
and the $E_j|_{z=(1, \ldots, 1)}$ from \eqref{E_j}.
It is shown in \cite{Hibi-Nishiyama-Takayama-2017} that
the set of the standard monomials of $\Ideal'$ gives
the set of standard monomials for $\Ideal=\mathcal{R} H_A(\beta)$.
This algorithm only involves
computation in commutative rings, which makes it more efficient than the calculation of standard monomials via Buchberger's algorithm.
\fi

\subsection{Basis change}\label{Basis_change}
\label{ssec:basis-change}
Assume that we know Pfaffians $\Pstd_i$ in a basis of standard monomials $\estd$.
In this section we show how to gauge transform the system into an arbitrary
basis $e$.

Consider an ideal $\mathcal{I}$ over a rational Weyl algebra $\Ring$.
Let a basis
$e = (e_1, \ldots, e_r)^T \subset \Ring/\mathcal{I}$
be given. In the following, the symbol $\circ$ denotes multiplication inside
of $\Ring$, i.e. we write $\partial_i \circ h$ when $h \in \Ring$. We
introduce $\circ$ to distinguish composition in the ring from
the action $\partial_i H$ when, say, $H$ is a matrix with entries in the
field $\mathbb{C}(z)$.

We seek the Pfaffian matrices $P_i \in \mathbb{C}(z)^{r \times r}$ satisfying
the differential equation for the basis $e$:
\eq{
     \partial_i \circ e = P_i \cdot e \, .
}
Assume that we know $\Pstd_i \in \mathbb{C}(z)^{r \times r}$ in
\eq{
    \label{Ps}
     \partial_i \circ \estd = \Pstd_i \cdot \estd
     \ ,
}
where $\estd = \{\partial^\alpha\} \subset \Ring/\mathcal{I}$ are the standard
monomials. If we could find a matrix $G \in \mathbb{C}(z)^{r \times r}$
relating the two bases,
\eq{
    e = G \cdot \estd
    \ ,
}
then we could relate the Pfaffian matrices in different bases via a gauge
transformation:
\eq{
    \label{Ps_to_Pe}
    P_i = \left( \partial_i G + G \cdot \Pstd_i \right) \cdot G^{-1}
    \ .
}
Note that $G$ is invertible by assumption, since both $\estd$ and $e$ are bases.

By \eqref{Ps_to_Pe}, our problem is reduced to finding the
transformation matrix $G$. To this end, we begin by expanding $e$ in
terms of $\estd$:
\eq{
    e = \sum_{k \in K} G^{(1)}_k \cdot \left( \partial^k \circ \estd \right) \ , \
    K \subset \mathbb{N}_0^N
    \ ,
}
for some matrix $G^{(1)}_k \in \mathbb{C}(z)^{r \times r}$.
Here we relied on the fact that any monomial of $e$ is divisible by an entry of $\estd$.
Next we transform each $\partial^k \circ \estd$ into the form $G^{(2)}_k
\cdot \estd$ via repeated application of \eqref{Ps}.
Let us illustrate the last step with an example.
\begin{example}
    \rm
    Given $k = \{1, 1, 0, \ldots, 0\}$, we can rewrite $\partial^k \circ \estd$ as
    \eq{
        \partial_1 \circ \partial_2 \circ \estd
        &=
        \partial_1 \circ \left( \Pstd_2 \cdot \estd \right) \\
        &=
        \partial_1 \Pstd_2 \cdot \estd + \Pstd_2 \cdot \left( \partial_1 \circ \estd \right) \\
        &=
        \partial_1 \Pstd_2 \cdot \estd + \Pstd_2 \cdot P_1^{(s)} \cdot \estd
        \ .
    }
    We hence have that $G^{(2)}_k = \partial_1 \Pstd_2 + \Pstd_2 \cdot \Pstd_1$.
\end{example}

We conclude that the gauge transformation matrix $G$ decomposes as follows:
\eq{
    \label{gauge_mat}
    G = \sum_{k \in K} G^{(1)}_k \cdot G^{(2)}_k \, ,
}
where each $G^{(2)}_k$ is built from the matrix products between Pfaffians
and derivatives thereof. Thus, equations \eqref{Ps_to_Pe, gauge_mat} constitute one
way to translate the Pfaffian matrices between different bases.

\subsection{Basis change without derivatives}\label{ssec: basis change without derivatives}
The matrix $G^{(2)}_k$ appearing in the gauge transformation matrix
\eqref{gauge_mat} involves derivatives of Pfaffian matrices. If we were to
work over a finite field, there would be no $z_i$ variables with respect to
which we could take the derivative. This situation is amended, in the GKZ case, by the following
proposition.
\begin{proposition}
The derivative of a Pfaffian matrix can be computed solely via addition and matrix multiplication from
\eq{
    (\partial_i P_j)(\beta) =
    \big( P_j(\beta-a_i) - P_j(\beta) \big) P_i(\beta)
    \ .
}
\end{proposition}
\begin{proof}
    Let
    \eq{
        F(\beta) \defas
        \frac{1}{\Gamma(\beta_0+1)} \big(s_1 f_\Gamma(\beta), \ldots, s_r f_\Gamma(\beta) \big)
    }
    be a solution to a GKZ Pfaffian system given some monomial basis
    $s_i = \partial^{k_i}$, $k_i \in \NN_0^N$, for which $s_i \partial_j = \partial_j s_i$.
    The following relations hold true:
    \eq{
        \label{P_F}
        \partial_j F(\beta) &= P_j(\beta) F(\beta)
        \ , \\
        \partial_j F(\beta) &= F(\beta-a_j)
        \ .
        \label{F(b-a)}
    }
    Differentiating \eqref{P_F} w.r.t $z_i$, we get
    \eq{
        \partial_i \partial_j F &=
        \big( \partial_i P_j \big) F + P_j \big( \partial_i F \big) \\&=
        \big( \partial_i P_j \big) F + P_j P_i F
        \ ,
        \label{di_dj_F_1}
    }
    where we omitted the argument $\beta$ for clarity.
    On the other hand, differentiating \eqref{F(b-a)} we get
    \eq{
        \partial_i \partial_j F(\beta) &=
        \partial_i F(\beta - a_j) \\&=
        F(\beta - a_j - a_i) \\&=
        P_j(\beta - a_i) P_i(\beta) F(\beta)
        \ ,
        \label{di_dj_F_2}
    }
    where we applied the identity $P_k(\beta) F(\beta) = F(\beta-a_k)$ twice in
    the last step.
    The proposition follows upon equating
    \eqref{di_dj_F_1} and \eqref{di_dj_F_2} and isolating $\big( \partial_i P_j
    \big) (\beta)$.
\end{proof}

Pfaffian systems introduced above are systems of {\it linear} partial differential
equations (SPDE), satisfied by the solutions of a given GKZ system.
As we will see later on, these equations are extremely useful in physical
applications. Next we present an efficient way to calculate the Pfaffian systems,
essentially via linear algebra.

\section{Constructing Pfaffian systems from Macaulay matrices}
\label{sec:pfaffian-from-macaulay}
In this section we describe a method for building the Pfaffian systems defined
in eq. \eqref{eqn:Pfaffian}. The method amounts to first building an auxiliary
matrix $M$ called the Macaulay matrix, and then solving a special system
of linear equations.
In \secref{ssec:pfaffian-macaulay} we derive the Macaulay matrix
\eqref{eq:macaulay-matrix} and the linear system \eqref{eq:C1sys, eq:C2sys} that
it satisfies. In \secref{ssec:macaulay-pfaffian} we then present
\Algref{alg:macaulay-matrix} for calculation of Pfaffian systems. In
\secref{ssec:macaulay-discussion} we give several remarks about the algorithm
and its efficiency. We close this section with several examples in
\secref{ssec:examples}, showcasing the steps and runtime statistics of the
algorithm in practice.

\subsection{From Pfaffian to Macaulay matrix}
\label{ssec:pfaffian-macaulay}
We present how the Macaulay matrix arises from a Pfaffian system in the basis of standard monomials.
Since we will focus our discussion on the case of GKZ systems later on, based on the comments at the end of \secref{ssec:D-module to DEQ} we may safely assume that
\begin{center}
    a set of standard monomials 
    $
    \Std \defas \brc{ \partial^{k} } 
    $ is given,
\end{center}
and that its size equals the holonomic rank $|{\rm Std}| = r$, defined in \eqref{holonomic_rank}.
We remind that $\partial^k$ denotes a monomial in derivatives, while $\partial_i$ denotes a single derivative w.r.t.\ $z_i$.

Now, the Pfaffian matrix $P_i$ in the direction $z_i$, defined in \eqref{eqn:Pfaffian},
is specified by the following relation:
\begin{equation}
    \partial_i \circ \Std_\alpha = \sum_{\beta} \, \brk{P_i}_{\alpha \beta} \cdot \Std_\beta \quad \text{in $\Ring/\Ideal$}
    \ ,
    \label{eq:pfaff-eq-quotient}
\end{equation}
where $\circ$ denotes composition in a
ring $\Ring$, $\alpha \in \brc{1, \ldots, \abs{\Std}}$ and $\Ideal$ is a general holonomic ideal for now.
Using the expression \eqref{eq:ideal-def} for $\Ideal$, we carry over the relation \eqref{eq:pfaff-eq-quotient}
from the quotient $\Ring/\Ideal$ to the whole ring $\Ring$ as
\begin{equation}
    \partial_i \circ \Std_\alpha = \sum_{j} \, Q_{\alpha j} \circ \gen_j
    + \sum_{\beta} \, \brk{P_i}_{\alpha \beta} \cdot \Std_\beta \quad \text{in $\Ring$}
    \ ,
    \label{eq:pfaff-eq-ring}
\end{equation}
where each $Q_{\alpha j}$ is an element of $\Ring$ and $j \in \brc{1, \ldots, \dimIdeal}$ labels the ideal
generators $\gen_j$ from \eqref{eq:ideal-def}. Now we bring the
operators $Q_{\alpha j}$ into normal ordered form as in \eqref{eq:normal-form}:
\begin{align}
    Q_{\alpha j} = \sum_{k} \, q_{\alpha j k} \cdot \partial^k, \quad \partial^k \in \Der
    \ ,
    \label{eq:Q-normal-form}
\end{align}
where $q_{\alpha j k } \in \Field\brk{z}^{\abs{\Std} \times \dimIdeal \times
\abs{\Der}}$ is a newly introduced rank 3 tensor%
\footnote{
    The notation $a_{n m} \in \Field\brk{z}^{N \times M}$ means that the
    indices range over the sets: $n \in \brc{1, \ldots, N}$ and $m \in
    \brc{1, \ldots, M}$.
}, and
we denote by $\Der$ the set of all partial derivatives $\partial^k$ appearing on the RHS.

Now we are ready to introduce the main player of this section: the Macaulay matrix.
It arises due to the normal ordered form of the first sum in \eqref{eq:pfaff-eq-ring}.
Substitution of the $Q_{\alpha j}$ operators from \eqref{eq:Q-normal-form} back into
\eqref{eq:pfaff-eq-ring} results in the action of $\partial^k \in \Der$ on the ideal generators $h_j$,
which we bring to normal ordered form as follows:
\begin{align}
    \partial^k \circ \gen_j = \sum_{L} \, M_{k j L} \cdot \partial^L, \quad \partial^L \in \Mons
    \ .
    \label{eq:macaulay-tensor}
\end{align}
Here we introduced a rank 3 tensor $M_{k j L} \in \Field\brk{z}^{\abs{\Der} \times \dimIdeal \times \abs{\Mons}}$
and a set $\Mons$, which collects the partial derivatives $\partial^L$ emerging on the RHS.
Upon concatenation%
\footnote{
    In \mathematica $,$ this operation can be done with \code{Flatten[\#, \{\{1, 2\}, \{3\}\}]\&}.
} of the first two indices $k$ and $j$, the rank 3 tensor $M_{k j L}$ turns
into the Macaulay matrix
\begin{align}
    M_{R L} \defas M_{\brk{k j} L} \in \Field\brk{z}^{\brk{\abs{\Der} \cdot \dimIdeal} \times \abs{\Mons}}
    \ ,
    \label{eq:macaulay-matrix}
\end{align}
where the combined index $R \defas \brk{k j}$ runs over a product of sets $R \in \Der \times \brc{1, \ldots, \dimIdeal}$.
Plugging equations \eqref{eq:Q-normal-form, eq:macaulay-tensor, eq:macaulay-matrix} back into the main relation \eqref{eq:pfaff-eq-ring}
for the Pfaffian, we observe that the first term can be written as a product of matrices:
\begin{align}
    \sum_{j} \, Q_{\alpha j} \circ \gen_j = \sum_{R L} \, q_{\alpha R} \, M_{R L} \cdot \partial^L
    \equiv \brk{C \cdot M \cdot \Mons}_\alpha
    \ ,
    \label{eq:macaulay-gen}
\end{align}
where $C \in \mathbb{C}(z)^{|\Std| \times (|\Der| \cdot \dimIdeal)}$ stands for the coefficient matrix $C_{\alpha R} \defas q_{\alpha \brk{j k}}$ of \eqref{eq:Q-normal-form}, and we regard the set $\Mons$ as a column vector.

The set $\Mons$, containing all monomials in $\partial_i$, can be partitioned into two disjoint sets consisting of standard and, what we call, exterior monomials:
\begin{align}
    \Mons = \Ext \sqcup \Std
    \ , \quad
    \Ext \defas \Mons \setminus \Std
    \ .
\end{align}
This partition naturally induces a column partition of the
Macaulay matrix:
\begin{align}
    M = \brk{\Mext | \Mstd}
    \ ,
    \label{eq:macaulay-split}
\end{align}
with the two column blocks $\Mext \in \Field\brk{z}^{\brk{\abs{\Der} \cdot \dimIdeal} \times \abs{\Ext}}$ and
$\Mstd \in \Field\brk{z}^{\brk{\abs{\Der} \cdot \dimIdeal} \times \abs{\Std}}$. Similarly, the LHS of \eqref{eq:pfaff-eq-ring}
can be decomposed as
\footnote{
    Note that $C^\prime, C_\Ext^\prime$ and $C_\Std^\prime$ depend on the direction $z_i$, since these matrices are born from the expression $\partial_i \circ \Std$.
}
\begin{align}
    \partial_i \circ \Std & \safed
    C^\prime \cdot \Mons \\& \safed
    C_\Ext^\prime \cdot \Ext + C_\Std^\prime \cdot \Std
    \ ,
    \label{eq:partial-std-split}
\end{align}
where we introduced the coefficient matrix $C^\prime \in \Field\brk{z}^{\abs{\Std} \times \abs{\Mons}}$
in addition to its column blocks
$C_\Ext^\prime \in \Field\brk{z}^{\abs{\Std} \times \abs{\Ext}}$ and
$C_\Std^\prime \in \Field\brk{z}^{\abs{\Std} \times \abs{\Std}}$, such that
$C^\prime = \brk{C_\Ext^\prime | C_\Std^\prime}$ in accordance with
\eqref{eq:macaulay-split}.

Finally substituting equations \eqref{eq:macaulay-gen, eq:macaulay-split, eq:partial-std-split} into the main relation \eqref{eq:pfaff-eq-ring}, we arrive at the following matrix relation:
\begin{align}
    \brk{C_\Ext^\prime - C\cdot \Mext} \cdot \Ext
    + \brk{C_\Std^\prime - C\cdot \Mstd} \cdot \Std
    = P_i \cdot \Std
    \ .
    \label{eq:pfaff-eq-matrix}
\end{align}
 Due to the linear independence of $\Ext$ and $\Std$, this is equivalent to a system
\begin{empheq}[box=\fbox]{align}
    C_\Ext^\prime - C\cdot \Mext &= 0
    \ ,
    \label{eq:C1sys}
    \\
    C_\Std^\prime - C\cdot \Mstd &= P_i
    \ .
    \label{eq:C2sys}
\end{empheq}
This system relates the Macaulay matrix $M$ to the Pfaffian $P_i$ in the direction $z_i$, and is central to our computational method presented next.

\subsection{From Macaulay matrix to Pfaffian}
\label{ssec:macaulay-pfaffian}
Now we reverse the logic in the derivation above: although we started from \eqref{eq:pfaff-eq-ring}, we would like to regard it as the goal.
We view \eqref{eq:pfaff-eq-ring} and its matrix version \eqref{eq:pfaff-eq-matrix} as a system of linear equations
for the unknown matrix of coefficients $C$. Given the Macaulay matrix $M$, the idea is to solve first the system
\eqref{eq:C1sys}, and then substitute the solution into \eqref{eq:C2sys} in order to obtain the desired Pfaffian matrix $P_i$.

Let us now turn this idea into a practical algorithm.
Say we would like to compute $P_i$ in direction $z_i$ given a zero-dimensional ideal $\Ideal$.
First, given a $\Degree \in \Integers_{\ge 0}$ we construct the set of all derivatives whose degree is bounded by $\Degree$:
\begin{align}
    \Der_\Degree \defas \brc{\partial^k}_{\abs{k} \leq \Degree}
    \ .
    \label{eq:der-degree-def}
\end{align}
Next, similarly to \eqref{eq:macaulay-tensor}, we write $\Mons_\Degree$ for the
set of monomials in derivatives appearing in the normal ordered form of
$\brc{\partial^k \circ \gen_j}_{\partial^k \in \Der_\Degree, \, j = 1, \dots, \dimIdeal}$.
We have
\eq{
    \brc{\partial^k \circ \gen_j}_{\partial^k \in \Der_\Degree, \, j = 1, \dots, \dimIdeal}
    &= \brk{M_\Degree}_{\brk{k j} L} \cdot \partial^L \\
    &\equiv M_\Degree \cdot \Mons_\Degree
    \ ,
    \label{eq:macaulay-degree}
}
where we call $M_\Degree$ the Macaulay matrix of degree $\Degree$ \brk{compare to \eqref{eq:macaulay-matrix}}.
We adjust the maximal degree $\Degree$ and the sets $\Der_\Degree$ and $\Mons_\Degree$ until the linear system \eqref{eq:C1sys}
can be resolved w.r.t.\,the unknown matrix $C$.
Finally, the desired Pfaffian follows from using $C$ in \eqref{eq:C2sys}.
We summarize this procedure in \Algref{alg:macaulay-matrix}.

\begin{algorithm}
    \begin{tabbing}
        \underline{Input}: Standard monomials $\Std$, generators $\gen_1,\gen_2,\dots$ of a zero-dimensional ideal $\Ideal \subset \Ring$,  direction $i$.
        \\[3pt]
        \underline{Output}: Pfaffian matrix $P_i$.
    \end{tabbing}
    \begin{algorithmic}[1]
        \State $\Degree:=0$;
        \State $M:=M_0=(\Mext|\Mstd)$;
        \State ${\rm Mons}:={\rm Mons}_0$;
        \While{
            ($\partial_i \circ \Std$ cannot be written as $C^\prime\cdot{\rm Mons}$)
            or ($C_\Ext^\prime=C\cdot \Mext$ is not solvable w.r.t. $C$)
        }
            \State $\Degree++;$
            \State $M\leftarrow M_\Degree;$
            \State ${\rm Mons}\leftarrow{\rm Mons}_\Degree$;
        \EndWhile
        \State Solve $\partial_i \circ \Std=C^\prime\cdot{\rm Mons}$ and $C_\Ext^\prime=C\cdot \Mext$;
        \State Return $P_i:=C_\Std^\prime-C\cdot \Mstd$;
    \end{algorithmic}
    \caption{: Pfaffian matrix by Macaulay matrix method}
    \label{alg:macaulay-matrix}
\end{algorithm}

\begin{proposition}
    \Algref{alg:macaulay-matrix} outputs the Pfaffian matrix $P_i$ in a finite number of steps.
    \label{prp:alg-finite}
\end{proposition}

\begin{proof}
    Since the left ideal $\Ideal$ is zero-dimensional, there exists a matrix $P_i$ with entries in $\Field\brk{z}$ such that
    $\pd{i} \circ \Std - P_i \cdot \Std \equiv 0 \> {\rm mod} \, \Ideal$ (see e.g. \cite[Chap 6]{dojo}).
    This means that the LHS can be expressed in terms of
    the generators $\brc{\gen_1, \dots, \gen_\dimIdeal}$ inside the ring $\Ring$, as in \eqref{eq:pfaff-eq-ring}.
    Let $N_0$ be the maximal degree of the monomials $\partial^{k} \in \Der$ from \eqref{eq:Q-normal-form}.
    Then \Algref{alg:macaulay-matrix} outputs the Pfaffian matrix $P_i$ until the degree parameter $\Degree$
    reaches $N_0$.
\end{proof}

\subsection{Discussion}
\label{ssec:macaulay-discussion}
Some comments about \Algref{alg:macaulay-matrix} are in order.

\paragraph{Other directions.} 
The condition on line 4 can be extended by checking $\pd{i}
\circ \Std$ for all $i = 1, \ldots, N$
to generate a Macaulay matrix $M_\Degree$ large enough for calculating 
Pfaffian matrices in all directions $\brc{P_i}_{i = 1, \dots, N}$.
\paragraph{Rank of $\Mext$.} The Macaulay matrix $\Mext$ produced by the
algorithm may not be of full rank. This makes the matrices $\Mext, \Mstd$
smaller and therefore the whole algorithm more efficient than the algorithm
of \cite{Ohara-Takayama-2015}.
\paragraph{Degree of freedom.} The unknown coefficient matrix $C$, which solves
the system \eqref{eq:C1sys}, can be chosen up to the left nullspace of the
Macaulay matrix $M$. Indeed, suppose that both $C=C^{(1)}$ and $C=C^{(2)}$ solve
the system \eqref{eq:C1sys}. Clearly, we have $\brk{C^{(1)} - C^{(2)}} \cdot
\Mext = 0$. Uniqueness of the Pfaffian matrix alongside 
\eqref{eq:C2sys} imply $\brk{C^{\brk{1}} - C^{\brk{2}}} \cdot \Mstd =
0$ as well. Therefore, we deduce
\begin{align}
    \brk{C^{\brk{1}} - C^{\brk{2}}} \cdot M = 0
    \ ,
    \label{eq:C1sys-freedom}
\end{align}
according to the partition \eqref{eq:macaulay-split} of the Macaulay matrix.
\paragraph{Solvability condition.} The system \eqref{eq:C1sys} has a solution
when the rows of the $C^\prime_1$ matrix lie in the row space of the $\Mext$
block of the Macaulay matrix. This can be formulated as a condition on the rank of
the following matrix:
\begin{align}
    \renewcommand{\arraystretch}{1.4} \mathrm{rank}
    \lrsbrk{
        \begin{array}{c}
            \Mext \\
            \hline
            C_\Ext^\prime
        \end{array}
    } = \mathrm{rank}\bigsbrk{\Mext}
    \ .
    \label{eq:C1sys-solvability}
\end{align}
\paragraph{Probabilistic method.} Checking the solvability condition
\eqref{eq:C1sys-solvability} can be computationally demanding. In order to
reduce the computation time and memory usage we may employ the so-called
probabilistic method. In the GKZ case, it amounts to setting the all the variables
$z_i$ and parameters $\beta$ \brk{see \eqref{f_Gamma(z)}} to numbers in
$\Rationals$ or a finite field
$\Finite_p$, which greatly simplifies the row reduction required for the
solvability check.
\paragraph{Smaller Macaulay matrix.} Once the condition
\eqref{eq:C1sys-solvability} is proved \brk{almost surely, if one employs the
probabilistic method}, we further simplify the system \eqref{eq:C1sys} by
selecting a subset of rows of the Macaulay matrix necessary for solving the
system. In practice, we perform such a selection using the probabilistic method
together with a determination of the left nullspace of the matrix on the LHS of \eqref{eq:C1sys-solvability}.
\paragraph{Finite fields.} For fixed values of the variables $z_i$
and the exponents $\beta$, numerical values of the Pfaffian matrix $P_i$
can be efficiently obtained via the Macaulay matrix method. Note that one
only needs numerical values of $P_i$ to obtain the values of 
normalizing constants for a certain class of statistical distributions, see
\cite{Ohara-Takayama-2015}. In our applications, however, we would like to
\brk{re}construct the full rational dependence of $P_i$ on all the
variables $z_i$ and $\beta$. When it becomes too demanding to solve the system
\eqref{eq:C1sys} analytically, the probabilistic approach can be
used in with tandem with rational function reconstruction over finite
fields $\Finite_p$. Indeed, using modular arithmetic (over prime numbers) and the
Chinese remainder theorem, the rational entries of $P_i \in \mathbb{C}(z)^{r \times r}$
can be efficiently obtained by means of functional reconstruction
algorithms for multivariate, rational functions, which were
developed in the context of scattering amplitudes in \cite{Peraro:2016wsq}. This
algorithm reconstructs analytic expressions for polynomials and rational
functions via repeated numerical evaluation over finite fields by means
of iterative application of Newton's polynomial representation and Thiele's
interpolation formula. The reconstruction algorithm is 
implemented in the public code \soft{FiniteFlow} \cite{Peraro:2019svx}, which
uses so-called dataflow graphs for building numerical algorithms over
finite fields.
\paragraph{Tests.} Once we reconstruct the Pfaffian matrices from
their numerical values over the finite fields, we check correctness as follows.
Recall that Pfaffian equations
\eqref{eq:pfaff-eq-quotient} can be regarded as a Gr\"obner basis in $\Ring$
\brk{see the discussion below eq. \eqref{eq:CI}}.  Therefore, if the
solution rank of the left ideal \eqref{eq:ideal-def} is known to be $r$, the
correct $r \times r$ Pfaffian matrices should satisfy the integrability
condition \eqref{eq:CI}, and should also reduce the generators $\gen_i$ of the
ideal \eqref{eq:ideal-def} to $0$ via the Pfaffian equations
\eqref{eq:pfaff-eq-quotient}.
\paragraph{Further improvements.} Let us comment on potential ways to
improve the efficiency of the \Algref{alg:macaulay-matrix}. By virtue of 
\eqref{eq:macaulay-degree}, to generate the Macaulay matrix of degree
$\Degree$ we are instructed to act with $\Der_\Degree$, the set of all derivatives of
degree $\le \Degree$ \brk{see \eqref{eq:der-degree-def}}, on the generators of ideal. The set $\Der_\Degree$ might be larger than necessary, so one way to improve the
\Algref{alg:macaulay-matrix} is to find a smaller set of derivatives needed for
the construction of the Macaulay matrices.
Another important direction of study is the prediction of singularities for the 
Pfaffian matrix in question, since such knowledge would greatly
simplify the solution of \eqref{eq:C1sys} using rational reconstruction. \\

\done{In conclusion, the Macaulay matrix method for determination of Pfaffian
matrices is an efficient alternative to the direct computation of the Gr\"obner
basis. In this section we presented the linear system of matrix equations
\eqref{eq:C1sys, eq:C2sys} as well as the \Algref{alg:macaulay-matrix},
which will be used for calculation of the Pfaffian matrices in various non-trivial
examples later on.
\henrik{Does this add anything in addition to the introduction to section 4?}
\seva{Nah, it's just repetition at this point, let's remove it from here.}
}

\subsection{Examples}
\label{ssec:examples}
\begin{example}
\label{ex:mm-F1}
\rm
%
Let us apply the Macaulay matrix method to compute the Pfaffian system for the
Appell function $F_1\brk{\alpha,\beta,\beta',\gamma;z_1, z_2}$.

\paragraph{Setup}
Consider the following integral representation:
%
%
\begin{align}
    &\frac{\Gamma\brk{\beta} \, \Gamma\brk{\beta'} \, \Gamma\brk{\gamma - \beta - \beta'}}{\Gamma\brk{\gamma}} \>
    F_1\brk{\alpha, \beta, \beta', \gamma; z_1, z_2}\nonumber\\
    =& \int_{\Delta^2}
    \brk{1 - x_1 - x_2}^{\gamma - \beta - \beta' - 1}
    \brk{1 - z_1 x_1 - z_2 x_2}^{-\alpha} \,
    x_1^\beta x_2^{\beta'} \,
    \frac{\dd x_1 \, \dd x_2}{x_1 x_2}\\
    =& \frac{\Gamma(1-\gamma+\alpha+\beta+\beta^\prime)}{\Gamma(1-\gamma+\beta+\beta^\prime)\Gamma(\alpha)}
    \int_{(0,+\infty)\times\Delta^2} g\brk{z; x}^{\gamma-\alpha-\beta-\beta^\prime-1} \,
    x_0^{\alpha}x_1^{\beta}x_2^{\beta^\prime} \,
    \frac{\dd x_0 \, \dd x_1 \, \dd x_2}{x_0 x_1 x_2} \, ,
\end{align}
where\,%
\footnote{
    Although this polynomial is not written in the standard form \eqref{g(z;x)}, it can be shown that $F_1$ does indeed satisfy a GKZ system.
}
\begin{equation}
    g\brk{z; x} = 1-x_1-x_2-x_0(1-z_1x_1-z_2x_2) \, ,
\end{equation}
and $\Delta^2$ is the two-dimensional simplex
$\Delta^2 \defas \brc{(x_1,x_2)\in\mathbb{R}^2\mid x_1,x_2,1-x_1-x_2\geq 0}$.
The ideal of the corresponding GKZ system is generated by three
independent differential operators $\Ideal = \vev{\gen_1, \gen_2, \gen_3}$.
However, to compute the Pfaffian matrices with generic parameters
$\brc{\alpha, \beta, \beta^\prime, \gamma}$ and $\brc{z_1, z_2}$,
it is enough to consider only the following pair:
\begin{equation}
    \begin{split}
        \gen_1 &= p_2 \, \pd{1}\pd{2} + p_3 \, \pd{2}^2 + p_4 \, \pd{1} + p_5 \, \pd{2} + p_6 \, ,\\
        \gen_2 &= q_2 \, \pd{1}\pd{2}                   + q_4 \, \pd{1} + q_5 \, \pd{2} \, ,
    \end{split}
    \label{eq:appell-gens}
\end{equation}
where the abbreviated coefficients are
%
%
\begin{alignat}{3}
    p_2&=z_1 (1-z_2)\, , && p_3=z_2(1-z_2) \, , \quad && p_4=-\beta'z_1 \, ,\nonumber\\
    p_5&=\gamma-z_2(\alpha+\beta'+1) \, ,\quad && p_6=-\alpha \beta' \, , &&\\
    q_2&=z_1-z_2 \, ,&& q_4=- \beta' \, , && q_5= \beta \, .\nonumber
\end{alignat}
%
%
\paragraph{Basis}
The GKZ system in question has rank $3$. 
The column vector of standard monomials reads
\begin{align}
    \Std = \lrbrk{
        \begin{array}{c}
            \pd{1} \\ \pd{2} \\ 1
        \end{array}
    } \, .
    \label{eq:appell-std}
\end{align}

\paragraph{Macaulay matrix}
For brevity, let us only consider the Pfaffian matrix $P_1$ in the direction $z_1$.
According to \eqref{eq:pfaff-eq-quotient}, our task is to express
\begin{align}
    \pd{1} \circ \Std = \lrbrk{
        \begin{array}{c}
            \pd{1}^2 \\ \pd{1} \pd{2} \\ \pd{1}
        \end{array}
    }
    \label{eq:appell-partial-std}
\end{align}
as a linear combination of $\Std$ over the rational function field $\QQ(\alpha,\beta,\beta',\gamma)(z_1,z_2)$.

Following the steps of \Algref{alg:macaulay-matrix}, we build the Macaulay matrix of degree $\Degree$ \eqref{eq:macaulay-degree} as follows.
First we act with $\Der = \brc{1,\pd{1},\pd{2}}$, the list of all possible derivatives of degree $\Degree \le 1$,
on the two generators $\brc{\gen_1, \gen_2}$ from \eqref{eq:appell-gens}
Then we bring the ensuing expressions into normally ordered form \brk{see \eqref{eq:normal-form}} and rearrange the result into a matrix.
The rows of the matrix are labeled by the list of operators $\brc{\gen_1, \pd{1} \gen_1, \pd{2} \gen_2, \gen_2, \pd{1} \gen_2, \pd{2} \gen_2}$,
while the columns correspond to the monomials $\Mons$ appearing in the normally
ordered form of these operators. Explicitly, we get the Macaulay matrix
of degree%
\footnote{
    Had we used all three generators of the ideal, it would have been enough to consider the degree $\Degree = 0$ Macaulay matrix.
    However, the shortened setup \eqref{eq:appell-gens} considered here provides a better showcase for the steps of \Algref{alg:macaulay-matrix}.
} $\Degree = 1$
\begin{align}
    \begin{blockarray}{cccccccccc}
        &  \partial_2   \partial_1^{ 2} &    \partial_2^{ 2}   \partial_1&   \partial_2^{ 3} &   \partial_1^{ 2} &   \partial_2  \partial_1&   \partial_2^{ 2} &  \partial_1&  \partial_2&  1 \\
        \begin{block}{c[cccccc|ccc]}
            \gen_1 & 0& 0& 0& 0&  p_{2}&  p_{3}&  p_{4}&  p_{5}&  p_{6} \\
            \pd{1}\gen_1 & p_{2}&  p_{3}& 0&  p_{4}&   p_{5}+ p_{2,1}&  p_{3,1}&   p_{6}+ p_{4,1}&  p_{5,1}&  p_{6,1} \\
            \pd{2}\gen_1 & 0&  p_{2}&  p_{3}& 0&   p_{4}+ p_{2,1}&   p_{5}+ p_{3,2}&  p_{4,2}&   p_{6}+ p_{5,2}&  p_{6,2} \\
            \gen_2& 0& 0& 0& 0&  q_{2}& 0&  q_{4}&  q_{5}& 0 \\
            \pd{1}\gen_2& q_{2}& 0& 0&  q_{4}&   q_{5}+ q_{2,1}& 0&  q_{4,1}&  q_{5,1}& 0 \\
            \pd{2}\gen_2 &0&  q_{2}& 0& 0&   q_{4}+ q_{2,2}&  q_{5}&  q_{4,2}&  q_{5,2}& 0 \\
        \end{block}
    \end{blockarray} \, ,
    \label{eq:appell-macaulay}
\end{align}
%
where $p_{i,j}=\partial_{z_j} p_i$ and $q_{i,j}=\partial_{z_j} q_i$. As dictated by \eqref{eq:macaulay-split}, we decompose the Macaulay matrix into
two blocks: the second block $\Mstd$ has columns labeled by $\Std$, and the first one $\Mext$ by the rest of the derivatives
\begin{align}
    \Ext = \Mons \setminus \Std = \bigbrk{
        \pd{2} \pd{1}^{2}, \, \pd{2}^{2} \pd{1}, \, \pd{2}^{3}, \, \pd{1}^{2}, \, \pd{2} \pd{1}, \, \pd{2}^{2}
    }^T.
    \label{eq:appell-ext}
\end{align}
Explicitly, we obtain the blocks
\begin{align}
    \Mext = \lrbrk{
        \begin{array}{cccccc}
            0& 0& 0& 0& z_2^- {z}_{1}&   z_2 z_2^- \\
              z_2^-  {z}_{1}&   z_2 z_2^-& 0&  -\beta'  {z}_{1}&    (-\beta' - \alpha - 2)  {z}_{2} + \gamma + 1& 0 \\
            0&   z_2^-  {z}_{1}&   z_2 z_2^-& 0&   -\beta'  {z}_{1} + z_2^-&    (-\beta' - \alpha - 3)  {z}_{2}+  \gamma+ 1 \\
            0& 0& 0& 0&   {z}_{1}- {z}_{2}& 0 \\
              {z}_{1}- {z}_{2}& 0& 0&  - \beta'&   \beta+ 1& 0 \\
            0&   {z}_{1}- {z}_{2}& 0& 0&   - \beta'- 1&  \beta \\
        \end{array}
    },
\end{align}
where $z_2^- \defas 1 - z_2$, and
%
\begin{align}
    \Mstd = \lrbrk{
        \begin{array}{ccc}
             -  \beta'  {z}_{1}&    (  - \beta'  - \alpha- 1)  {z}_{2}+ \gamma&  -  \alpha  \beta' \\
              (  - \alpha- 1)  \beta'& 0& 0 \\
            0&    (  - \alpha- 1) ( \beta'  +1)& 0 \\
             - \beta'&  \beta& 0 \\
            0& 0& 0 \\
            0& 0& 0 \\
        \end{array}
    }.
\end{align}

\paragraph{Pfaffian matrix}
We now turn to the linear system \eqref{eq:C1sys,eq:C2sys}.
The coefficient matrices $C^\prime_\Ext$ and $C^\prime_\Std$, defined in  \eqref{eq:partial-std-split}, are obtained by expressing \eqref{eq:appell-partial-std} in terms of the sets \eqref{eq:appell-std}
and \eqref{eq:appell-ext}:
\begin{align}
    \pd{1} \circ \Std = \lrbrk{
        \begin{array}{c}
            \pd{1}^2 \\ \pd{1} \pd{2} \\ \pd{1}
        \end{array}
    }
    =
    \lrbrk{
        \begin{array}{cccccc}
            0&0&0&1&0&0 \\
            0&0&0&0&1&0 \\
            0&0&0&0&0&0
        \end{array}
    }
    \lrbrk{
        \begin{array}{c}
            \pd{2}   \pd{1}^{ 2} \\    \pd{2}^{ 2}   \pd{1}\\   \pd{2}^{ 3} \\   \pd{1}^{ 2} \\   \pd{2}  \pd{1}\\   \pd{2}^{ 2}
        \end{array}
    }
    +
    \lrbrk{
        \begin{array}{ccc}
            0&0&0 \\
            0&0&0 \\
            1&0&0
        \end{array}
    }
    \lrbrk{
        \begin{array}{c}
            \pd{1} \\  \pd{2} \\  1
        \end{array}
    }
    \equiv
    C^\prime_\Ext \cdot \Ext + C^\prime_\Std \cdot \Std \, .
\end{align}
%
Let $C=(c_{ij} \,|\, 1 \leq i \leq 3, 1 \leq j \leq 6)$ be a matrix of unknowns.
We require that $C$ satisfies the system \eqref{eq:pfaff-eq-ring}.
In this case,
\begin{align}
    C \cdot \lrbrk{
        \begin{array}{c}
            \gen_1 \\ \pd{1} \gen_1 \\ \pd{2} \gen_1 \\ \gen_2 \\ \pd{1} \gen_2 \\ \pd{2} \gen_2
        \end{array}
    }
    + P_1 \cdot \Std
    = \pd{1} \circ \Std \, .
\end{align}
Upon collecting terms proportional to $\Ext$, we arrive at the
system \eqref{eq:C1sys} to be solved with respect to $C$. 
By the rank condition \eqref{eq:C1sys-solvability}, we note that the degree $\Degree=1$ Macaulay matrix \eqref{eq:appell-macaulay} does indeed lead to a solution.
Once $C$ is determined, \eqref{eq:C2sys} yields the Pfaffian matrix
%
%
\begin{align}
    P_1 = \lrbrk{
        \begin{array}{ccc}
            \frac{(-(\alpha+\beta-\beta^{\prime}+1) z_1+\gamma-\beta^{\prime})
            z_2+(\alpha+\beta+1) z_1^2-\gamma z_1}{\brk{1 - z_1} z_1 \brk{z_1 - z_2}} &
            \frac{-\beta z_2^2+\beta z_2}{\brk{1 - z_1} z_1 \brk{z_1 - z_2}} &
            \frac{\beta \alpha}{\brk{1 - z_1} z_1}
            \\
            \frac{\beta^{\prime}}{z_1-z_2} & \frac{-\beta}{z_1-z_2} & 0
            \\
            1 & 0 & 0
        \end{array}
    }.
    \label{eq:appell-pfaffian}
\end{align}
\end{example}

\hfill $\blacksquare$

\begin{example}\rm 
\label{ex:test5}
Consider the matrix
    $$
    A_5=\left(
    \begin{array}{ccccccccccc}
        1 & 1 & 1 & 1 & 1 & 1 & 1 & 1 & 1 & 1 & 1 \\
        1 & 0 & 0 & 0 & 0 & 1 & 1 & 1 & 0 & 0 & 0 \\
        0 & 1 & 0 & 0 & 0 & 0 & 0 & 0 & 1 & 1 & 0 \\
        0 & 0 & 1 & 0 & 0 & 1 & 0 & 0 & 0 & 0 & 1 \\
        0 & 0 & 0 & 1 & 0 & 0 & 1 & 0 & 1 & 0 & 0 \\
        0 & 0 & 0 & 0 & 1 & 0 & 0 & 1 & 0 & 1 & 1 \\
    \end{array}
    \right).
    $$
    The solution rank of this GKZ system is $r=13$, and a standard basis obtained
    via the method of \cite{Hibi-Nishiyama-Takayama-2017} reads:
    $$\Std=
    (
        \pd{9}\pd{11}^2,
        \pd{9}^2,
        \pd{10}^2,
        \pd{8}\pd{11},
        \pd{9}\pd{11},
        \pd{10}\pd{11},
        \pd{11}^2,
        \pd{7},
        \pd{8},
        \pd{9},
        \pd{10},
        \pd{11},
        1
    )^T.
    $$
    The Macaulay matrix $\Mext$ of degree $D = 2$ is a $189 \times 113$ matrix
    of rank $113$ (over $\ZZ/65537\ZZ$).
    The Macaulay matrix is obtained in \code{0.558 sec}%
    \footnote{
        Our implementation is in the \soft{Risa/Asir} system and the language \soft{C}.
        It is executed on a machine with an Intel Core i7-10700K processor
        of \code{3.8GHz} and \code{16Gb} of RAM. We do not use multicore functionality.
        All timings are taken on this machine unless specified otherwise.
        \done{\comment{The pentagon calculation was executed on Henrik's laptop}}
    }.
    It took \code{0.603 sec} to obtain the numerical Pfaffian matrix in the direction
    $z_6$ for fixed numerical values of all $z$ variables and $\beta$ parameters.
    \hfill $\blacksquare$
\end{example}

\begin{example}\rm 
\label{ex:test9}
    Consider the matrix
    $$
    A_6=\left(\begin{array}{ccccccccccccccc}
        1 & 1 & 1 & 1 & 1 & 1 & 1 & 1 & 1 & 1 & 1 & 1 & 1 & 1 & 1 \\
        1 & 0 & 0 & 0 & 0 & 0 & 1 & 1 & 1 & 0 & 0 & 0 & 0 & 0 & 0 \\
        0 & 1 & 0 & 0 & 0 & 0 & 0 & 0 & 0 & 1 & 1 & 1 & 0 & 0 & 0 \\
        0 & 0 & 1 & 0 & 0 & 0 & 1 & 0 & 0 & 0 & 0 & 0 & 1 & 1 & 0 \\
        0 & 0 & 0 & 1 & 0 & 0 & 0 & 1 & 0 & 1 & 0 & 0 & 0 & 0 & 1 \\
        0 & 0 & 0 & 0 & 1 & 0 & 0 & 0 & 1 & 0 & 1 & 0 & 1 & 0 & 0 \\
        0 & 0 & 0 & 0 & 0 & 1 & 0 & 0 & 0 & 0 & 0 & 1 & 0 & 1 & 1 \\
    \end{array}\right) \ ,
    $$
    whose solution rank is $r=33$.
    Our choice of standard basis is
    \begin{eqnarray*}
        \Std&=&\big(
        \pd{15}^{ 3} ,  \pd{5}^{ 3} ,   \pd{5}  \pd{6}  \pd{11},   \pd{6}^{ 2}   \pd{11},  \pd{3}   \pd{6}^{ 2} ,  \pd{4}   \pd{6}^{ 2} ,  \pd{6}^{ 3} ,  \pd{11}  \pd{15},  \pd{14}  \pd{15},  \pd{15}^{ 2} ,  \pd{4}^{ 2} , \\
        & & \quad  \pd{5}  \pd{11},  \pd{2}  \pd{5},  \pd{3}  \pd{5},  \pd{4}  \pd{5},  \pd{5}^{ 2} ,  \pd{6}  \pd{11},  \pd{2}  \pd{6},  \pd{6}  \pd{14},  \pd{3}  \pd{6},  \pd{6}  \pd{15}, \\
        & & \quad \pd{4}  \pd{6},  \pd{5}  \pd{6},  \pd{6}^{ 2} , \pd{11}, \pd{2}, \pd{14}, \pd{3}, \pd{15}, \pd{4}, \pd{5}, \pd{6}, 1\big)^T.
    \end{eqnarray*}
    Using the homogeneity propery discussed in \secref{ssec:rescaling},
    we may fix the following variables: $z_1=z_7=z_8=z_9=z_{10}=z_{12}=z_{13}=1$.
    The block $\Mext$ of the Macaulay matrix of degree $\Degree = 2$%
    \done{\seva{is it true?}}%
    turns out to be a sparse $945\times 958$ matrix, whose rank is $534$
    by the probabilistic method.
    Therefore, there are many rows that are not needed for solving the
    system \eqref{eq:C1sys}.
    Again, using the  probabilistic method over the finite field $\ZZ/65537\ZZ$,
    we determine a maximal set of independent rows of the matrix $\Mext$.
    We find that a $534\times 958$ matrix $N_\Ext$ is enough
    to solve a smaller version of the system \eqref{eq:C1sys}, namely $C_\Ext^{\prime}=C\cdot N_\Ext$.
    As was mentioned in the Smaller Macaulay Matrix paragraph of \secref{ssec:macaulay-discussion}, we can further
    reduce this new system as follows:
    since there are exactly%
    \done{\seva{recheck: is it 23 or 27?}}%
    $27$ independent row vectors in $C_\Ext^{\prime}$,
    we may choose only a subset of the row vectors in $N_\Ext$, whose span includes the independent row vectors of $C_\Ext^\prime$.

    The Macaulay matrix is then obtained in \code{15.549 sec}. 
    It took \code{0.66 sec} to calculate the Pfaffian matrix in the $z_6$
    direction for fixed numerical values of variables $z$ and parameters
    $\beta$. We leave the problem of full functional reconstruction of the
    Pfaffian matrices of this example for future work.
    \done{I wrote a code to generate an input for Mathematica. The auto-generated input is 2022-04-20-hexagon-test9-numpf.m. The Mathematica parser makes the job of replacing variables to numbers. The parser does this job faster than the replacement command of asir, which is slow and I have to improve it in a future. Anyway, the timing is now become 0.66 seconds! Note also that z6 is also specialized to a number. It is an advantage of Macaulay matrix.
Enclose me by done after you read.} 

\hfill $\blacksquare$
\end{example}

We have proposed an efficient Macaulay matrix method to construct Pfaffian
systems. 
It can be applied
to relatively large systems which are not feasible by Gr\"obner basis
methods in the ring of differential operators.
The latter methods are usually feasible for systems of up to, approximately, rank $10$.

\section{Macaulay matrix method and generalized Feynman integrals}
\label{sec:macaulay-feynman}

The relation between Feynman integrals and GKZ-hypergeometric functions has been studied in~\cite{Nasrollahpoursamami:2016,Feng:2019bdx, Watanabe:2013ova, Fleischer:2003rm,Klemm:2019dbm}.
Here, we show applications of GKZ systems combined with Macaulay matrices to Feynman integrals.

Let $0 < \e,\d \ll 1$, $d_0 \in 2 \cdot \NN$, $L \in \NN$ and $\nu \defas (\nu_1,\ldots,\nu_n) \in \ZZ^n$.
Moreover, fix the exponents of the Euler integral \eqref{f_Gamma(z)} to
\eq{
    \label{GFI_beta}
    \beta =
    (\e, -\e\d, \ldots, -\e\d) - (d_0/2, \nu_1, \ldots, \nu_n) \, .
}
We define a \emph{generalized Feynman integral} 
as
\begin{eqnarray}
    \label{gen_feyn_int}
    I(d_0,\nu;z) &\defas&
    c(d_0,\nu) f_\Gamma(\beta) \> ,
\end{eqnarray}
where%
\begin{eqnarray}
    f_\Gamma(\beta) &\defas&
    \int_\Gamma \G(z;x)^{\e-d_0/2} \, 
    x_1^{\nu_1+\e\d} \, \cdots \, x_n^{\nu_n+\e\d} \, 
    \frac{\dd x}{x} \ , \\
    \label{cd0nu}
    c(d_0,\nu) &:=&
    \frac
    { \Gamma(d_0/2 - \e) }
    { \Gamma \big( (L+1)(d_0/2-\e)-|\nu|-n\e\d \big) \prod_{i=1}^n \Gamma(\nu_i+\e\d) } 
    \quad , \quad
    |\nu| := \nu_1 + \ldots + \nu_n \ .
\end{eqnarray}
The polynomial $\G$ takes the form
\eq{
    \G(z;x) = \sum_{i=1}^N \, 
    z_i \,  x^{\alpha_i} 
    \quad , \quad
    \alpha_i \in \NN_0^n \, ,
}
and we set $a_i \defas (1, \alpha_i) \in \NN_0^{n+1}$ as usual.
The introduction of the $\delta$ parameter in \eqref{GFI_beta} follows from the analytic regularization of \cite{speer1969theory, Speer1971}

Using the notation of equations \eqref{omega_q} and \eqref{omega_q_Gamma}, we can equally well define the generalized Feynman integral as a pairing between a twisted cycle $\Gamma$ and a twisted form
\eq{
    \label{omega_feynman}
   \omega_{d_0/2, \, \nu} \defas
   c(d_0, \nu) \times
   \G(z;x)^{-d_0/2} x^\nu \frac{\dd x}{x} \, .
}
In particular, we have
\eq{
   \label{eq:generalized_Feynman_integral_definition}
    I(d_0,\nu;z) = \langle \omega_{d_0/2, \, \nu} \rangle_\Gamma \, .
}

The generalized Feynman integral reduces to the Lee-Pomeransky (LP) representation \cite{Lee:2013hzt} of an $L$-loop  integral in $d=d_0-2\e$ dimensions, with propagator powers $\nu_i$ when:
\eq{
\Gamma = {(0,+\infty)^n} \ , 
\quad 
    \d \to 0 \ ,
}
and
\eq{
    \label{zi_limits}
    z_i \in \NN \, \vee \, z_i \in \text{kinematic variables } (m_i^2, p_i^2, p_i \cdot p_j)
    \ ,
} 
where $m_i$ and $p_i$ stand for masses and external momenta, so that
the polynomial $\G$ becomes the LP polynomial built from Symanzik (or graph) polynomials $\mathcal{U}, 
\mathcal{F}$, as $\G = \mathcal{U} + \mathcal{F}$.\footnote{
The limit $\delta \to 0$ of a generalized Feynman integral may yield to ill-defined expressions of the form
$\frac{\dd x_i x^{\nu_i-1}_i}{ \Gamma(\nu_i)}$, when
$\exists \,  \nu_i \in \mathbb{Z}_{\leq 0}$. In these cases, 
we adopt the replacement $\frac{\dd x_i x^{\nu_i-1}_i}{ \Gamma(\nu_i)} \to \partial^{-\nu_i}_i \left( \delta_{D}(x_i) \right)$, with $\delta_{D}(\bullet)$ being the Dirac delta function, as proposed in \cite{LeeTalkMoriond,Lee:2014tja,SameshimaPhD}.
}

Let us observe that within the LP polynomial $\G(z;x)$, each $z_i$ may not necessarily be independent from each other. 
This is different from the polynomial $g(z;x)$ appearing in the Euler integral $f_\Gamma(z)$, for which each $z_i$ is considered independent.
Their independence ensures a non-degenerate correspondence between the monomials in $x_i$ variables and the partial differential operators in the $z_i$ variables - a crucial property to establish the isomorphism between twisted de Rham cohomology group and $\D$-modules. 
This observation implies that Feynman integrals are {\it restrictions} of GKZ integrals, obtained from the latter by choosing suitable values of the variables $z_i$ 
\cite{delaCruz:2019skx,Klausen:2019hrg,Klausen:2021yrt,Tellander:2021xdz,Feng:2019bdx,Vanhove:2018mto}.
We note that the $z_i$ limits in \eqref{zi_limits} are generically not smooth at the level of the Pfaffian system for a GFI.
In a forthcoming publication \cite{Chestnov:2023kww}, we show that it is nevertheless possible to construct the Pfaffian system with all $z_i$ variables identified with their proper kinematic counterparts.

\subsection{Massless one-loop diagrams}
\label{ssec:one-loop}
Let us consider a massless one-loop diagram $G$ with $n$ external legs and $n$ internal edges.
We use the symbol $p_i$ for each external momentum and $x_i$ for the Schwinger parameter of the $i$-th edge.
The first and the second Symanzik polynomials read
\eq{
\mathcal{U}_G&=x_1+\dots+x_n \, ,\\
\mathcal{F}_G&=\sum_{\substack{1\leq i<j\leq n\\ (i,j)\neq (1,2),\dots,(n-1,n),(1,n)}} \sigma_{ij} \, x_ix_j\label{eqn:F_G} \, ,
}
where we set $\sigma_{ij}:=-(p_i+\dots+p_{j-1})^2$.
Since the Lee-Pomeransky polynomial $\mathcal{G}=\mathcal{U}_G+\mathcal{F}_G$ consists of $\binom{n}{2}$ terms, it gives rise to an $(n+1)\times \binom{n}{2}$ $A$-matrix.
In view of the discussion of \secref{ssec:rescaling}, the number of reduced variables is
\eq{
    \binom{n}{2}-n-1=\frac{n(n-3)}{2} \, ,
}
which coincides with the number of independent Mandelstam variables $s_{ij}$ (as was also noted in \cite{Vanhove:2018mto}).
In sum, we conclude that the master integrals for $G$ are precisely subject to the GKZ system of PDEs.

Moreover, the discussion above can be easily extended to the cases when external masses $m_k^2=p_k^2$ are non-zero but treated as independent variables. 
For example, in the case of the one-mass $n$-gon diagram, all external masses are zero but one, say $p_k^2 \ne 0$, therefore a new term $-p_k^2 x_k x_{k+1}$ is added to the second Symanzik polynomial \eqref{eqn:F_G}. 
The number of reduced variables, in this case, reads $\frac{n(n-3)}{2}+1$ which is equal to
$\{\sigma_{ij}\}\cup\{ p_k^2\}$.

\subsection{Example: One-loop box}
\label{sec:one_loop_box}

In the following two sections, we illustrate the Macaulay matrix method for obtaining Pfaffian systems for Feynman integrals.
We have chosen examples where the number of GKZ variables correspond to number of independent kinematic variables.
The general case will be investigated in \cite{Chestnov:2023kww}. \\

\noindent
\begin{minipage}{12cm}
\paragraph{Setup}
Let us first derive the Pfaffian system for the one-loop massless box diagram.
This example was studied in Lee-Pomeransky representation using twisted cohomology in \cite{Mizera:2019vvs}.
The kinematics are
\eq{
    \sum_{i=1}^4 p_i = 0
    \quad , \quad
    p_1^2 = \cdots = p_4^2 = 0
    \quad , \quad
    s = 2 p_1 \cdot p_2
    \quad , \quad
    t = 2 p_2 \cdot p_3 \, .
}
\end{minipage}
\begin{minipage}{3cm}
\includegraphics[scale=0.15]{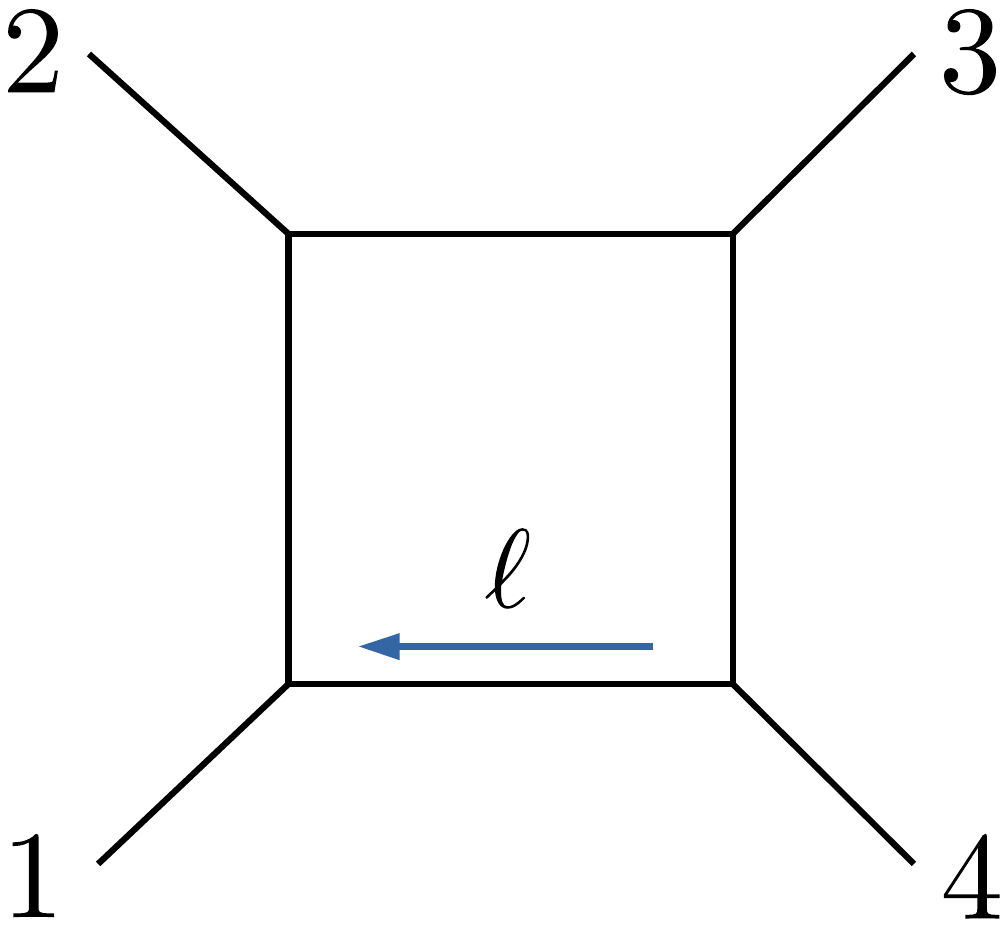}
\end{minipage}

\noindent
The inverse propagators in momentum space are given by%
\footnote{
    Here $\ell^2 = \ell_0^2 - \ell_1^2 - \cdots -\ell_d^2$.
}
\eq{
    \prop_1 = -\ell^2
    \ , \
    \prop_2 = -(\ell-p_1)^2 
    \ , \
    \prop_3 = -(\ell-p_1-p_2)^2 
    \ , \
    \prop_4 = -(\ell-p_1-p_2-p_3)^2 \, .
}
The corresponding generalized Feynman integral reads 
\eq{
    \label{Id0nu_box}
    I(d_0,\nu;\wtz) =
    c(d_0,\nu) (-s)^{d_0/2-\e-|\nu|-4\e\d}
    \int_\Gamma \G(\wtz;x)^{\e-d_0/2} x_1^{\nu_1+\e\d} \cdots x_4^{\nu_4+\e\d} \frac{\dd x}{x} \, ,
}
where
\eq{
    \label{G_box}
    \G(\wtz;x) &=
    \sum_{i=1}^4 \wtz_i x_i + \wtz_5 x_1 x_3 + \wtz_6 x_2 x_4 
    \ , \quad {\rm with} \quad 
    A_4 =
    \begin{pmatrix}
        1 & 1 & 1 & 1 & 1 & 1 \\
        1 & 0 & 0 & 0 & 1 & 0 \\
        0 & 1 & 0 & 0 & 0 & 1 \\
        0 & 0 & 1 & 0 & 1 & 0 \\
        0 & 0 & 0 & 1 & 0 & 1
    \end{pmatrix} \, .
}
The $s$-dependent prefactor in \eqref{Id0nu_box} comes from rescaling integration variables as $x_i \to x_i/(-s)$.
Within the LP-representation, the coefficients $z_i$ in \eqref{G_box} take the explicit form
\eq{
    \wtz_1 &= \cdots = \wtz_5 = 1
    \ , \
    \wtz_6 = \frac{t}{s} =: z \ ,
}
leaving just one external kinematic variable $z$. 
This counting is compatible with the number of GKZ variables after using the homogeneity property discussed in \secref{ssec:rescaling}. Indeed, out of the $N=6$ variables $z_i$ appearing in the generalized Feynman integral, $n+1=4+1=5$ of them  can be rescaled, thereby yielding $z_1 = \cdots = z_5 = 1$ and $z_6 = z$. 

\paragraph{Basis}

Following \cite{Henn:2014qga}, we choose the canonical basis
\eq{
    \label{basis_box_int}
    e =
    (-s)^{\e}
    \Big(z (-s) I(4,0,1,0,2), \ (-s) I(4,1,0,2,0), \ \e z (-s)^2 I(4,1,1,1,1)
    \Big)
}
corresponding to the cohomology basis (see \eqref{omega_feynman})
\eq{
    e^{(\mathrm{dR})} = \Lambda \cdot
    \left(
    \frac{x_4}{x_1 x_3 \G(z;x)^2} \dd x \ , \
    \frac{x_3}{x_2 x_4 \G(z;x)^2} \dd x \ , \
    \frac{1}{\G(z;x)^2} \dd x
    \right) \, ,
    \label{basis_box_D}
}
where $\Lambda$ is a diagonal matrix containing prefactors:
\begin{subequations}
\eq{
    \Lambda_{11} &= \frac
    {(-s)^{-4\e\d}z\Gamma(2-\e)}
    {\Gamma(1-2\e-4\e\d) \Gamma(\e\d)^2 \Gamma(1+\e\d) \Gamma(2+\e\d)} \ , \\[7pt]
    \Lambda_{22} &= \frac
    {(-s)^{-4\e\d}\Gamma(2-\e)}
    {\Gamma(1-2\e-4\e\d) \Gamma(\e\d)^2 \Gamma(1+\e\d) \Gamma(2+\e\d)} \ , \\[7pt]
    \Lambda_{33} &= \frac
    {\e (-s)^{-4\e\d}z\Gamma(2-\e)}
    {\Gamma(-2\e-4\e\d)\Gamma(1+\e\d)^4}
\ . } 
\label{Lambda_box}
\end{subequations}

\noindent
Applying the correspondence in \eqref{P(q)} (or rather \eqref{P(q)_simplex}, 
which employs homogeneity) 
to the basis $e^{(\mathrm{dR})}$,
we obtain the $\D$-module basis 
\eq{
 e^{(\D)} = \Lambda \cdot \bigg( 
        & \frac{\d\e-1}{(\e-1)\e} \partial - \frac{z}{(\e-1)\e} \partial^2 \ , \nonumber \\
        & \frac{(1-3\d)(4\d-1)\e}{\e-1} +
           \frac{z(7\d\e-2\e-1)}{(\e-1)\e} \partial -
           \frac{z^2}{(\e-1)\e} \partial^2 \ , \\ \nonumber
           & \frac{4\d\e-\e-1}{(\e-1)\e} \partial - \frac{z}{(\e-1)\e} \partial^2
 \bigg) \ ,
}
where $\partial := \partial_z$ .

\paragraph{Macaulay and Pfaffian matrices}

Using the Macaulay matrix method, we construct the Pfaffian system in the basis of standard monomials,
\eq{
    \label{Pfaffian_box_std}
    \partial \circ
    \begin{pmatrix}
        \partial^2 \\
        \partial \\
        1
    \end{pmatrix}
    =
    \begin{pmatrix}
        \partial^3 \\
        \partial^2 \\
        \partial
    \end{pmatrix}
    =
    \Pstd \cdot
    \begin{pmatrix}
        \partial^2 \\
        \partial \\
        1
    \end{pmatrix} \ .
} 
Afterwards, we perform a gauge transformation to the basis $e^{(\D)}$, so that $P$ can be read off the equation $\partial \circ e^{(\D)} = P \cdot e^{(\D)}$.

The Macaulay data for the system \eqref{Pfaffian_box_std} is found to be
\begin{subequations}
\eq{ 
    \Mext  &= \begin{pmatrix} z^2(z+1) \end{pmatrix} \ ,\\[4pt]
    \Mstd^T &= 
    \begin{pmatrix}
        z \big((6 \d +2) \e +z (6 \d  \e +\e +3)+3\big) \\
        (3 \d  \e + \e + 1)^2 + z \big ( (9 \d + 2) \e ^2\d + 6 \d  \e + \e + 1 \big) \\
        (4 \d + 1) \e ^3 \d ^2 
    \end{pmatrix} \ , \\[4pt]
    C' &=
    \begin{pmatrix}
        1 & 0 & 0 & 0 \\
        0 & 1 & 0 & 0 \\
        0 & 0 & 1 & 0
    \end{pmatrix} \ , \
    C_{\Ext}' = \begin{pmatrix} 1 & 0 & 0 \end{pmatrix} \ ,  \
    C_{\Std}' = 
    \begin{pmatrix} 
        0 & 0 & 0 \\
        1 & 0 & 0 \\
        0 & 1 & 0
    \end{pmatrix} \ , \\[4pt]
    \Mons &= \Ext \sqcup \Std = \{\partial^3\} \sqcup \{\partial^2, \partial, 1\} \ .
}
\end{subequations}
Setting 
$
    C = \begin{pmatrix} c_{11} & c_{12} & c_{13} \end{pmatrix}
$
and solving $C_{\Ext}' = C \cdot \Mext$, we obtain
\eq{
    C = \begin{pmatrix} \frac{1}{z^2(z+1)} & 0 & 0 \end{pmatrix} \ .
}
The Pfaffian matrix is then
\eq{
    \Pstd &= C_{\Std}' - C \cdot \Mstd \\&=
    \begin{pmatrix}
        \Pstd_{11} & \Pstd_{12} & \Pstd_{13} \\
        1 & 0 & 0 \\
        0 & 1 & 0
    \end{pmatrix} \ ,
}
where
\begin{subequations}
\eq{
    \Pstd_{11} &= 
    -\frac{\e  \big( (z+1)6 \d +z+2 \big) + 3 (z+1)}{z (z+1)} \ , \\
    \Pstd_{12} &= 
    -\frac{(3 \d  \e +\e +1)^2+z \big( (9 \d +2) \e ^2 \d +6 \d  \e +\e +1\big)}{z^2 (z+1)} \ , \\
    \Pstd_{13} &= 
    -\frac{(4 \d +1) \e ^3 \d ^2 }{z^2 (z+1)} \ .
}
\end{subequations}
Finally, according to the algorithm of \secref{ssec:basis-change}, it is possible to build a suitable gauge transformation matrix $G$ such that
\eq{
    P &= (\partial G + G \cdot \Pstd) \cdot G^{-1} \\
      &= \e
    \begin{pmatrix}
        -\frac{1}{z} & 0 & 0 \\
        0 & 0 & 0 \\
        -\frac{2}{z(z+1)} & \frac{2}{z+1} & -\frac{1}{z(z+1)}
    \end{pmatrix} ,
    \label{Pfaffian_box_e}
}
where the limit $\d \to 0$ has been taken.
The matrix \eqref{Pfaffian_box_e} is canonical and in agreement with \soft{LiteRed}.
\hfill $\blacksquare$

\subsection{Example: One-loop pentagon}
\label{sec:one_loop_pentagon}

\begin{minipage}{12cm}
\paragraph{Setup}
We consider a one-loop massless pentagon integral with one massive leg.
The kinematics are 
\eq{
    \sum_{i=1}^5 p_i = 0 
    \quad , \quad
    p_1^2 = \cdots = p_4^2 = 0 
    \quad , \quad
    p_5^2 \defas p^2
    \quad , \quad
    s_{ij} \defas 2 p_i \cdot p_j \ .
}
\end{minipage}
\begin{minipage}{3cm}
\includegraphics[scale=0.12]{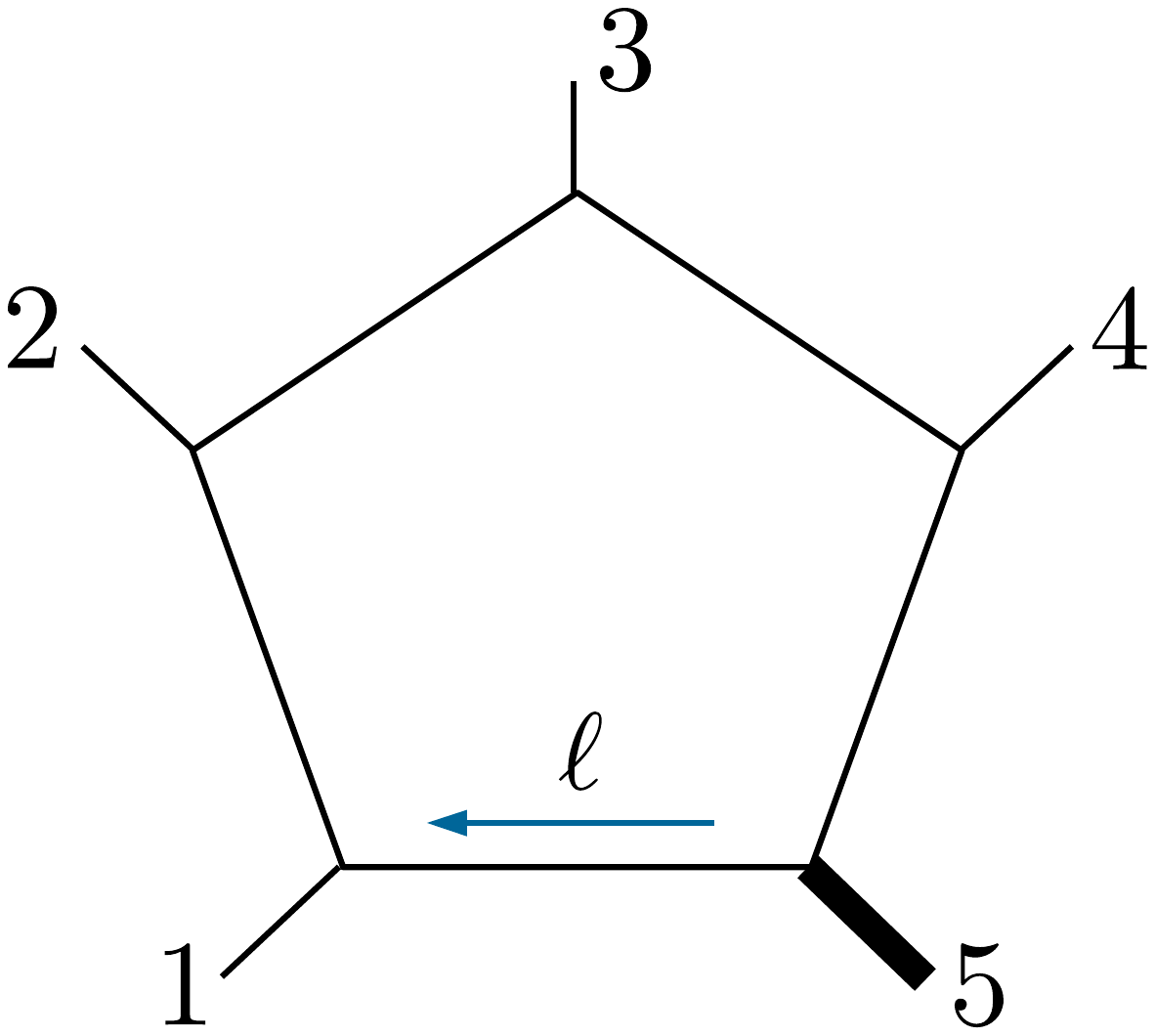}
\end{minipage}

\noindent
Note that the identity $(p_1+p_2+p_3+p_4)^2 = (-p_5)^2 = p^2$ imposes a relation on the Mandelstam variables $s_{ij}$.
The propagators are given by
\begin{alignat}{3}
    & \prop_1 = -\ell^2
    \ , \
    \prop_2 = -(\ell-p_1)^2
    \ && , \
    && \prop_3 = -(\ell-p_1-p_2)^2  \ ,
    \\
    & \prop_4 = -(\ell-p_1-p_2-p_3)^2
    \ && , \
    && \prop_5 = -(\ell-p_1-p_2-p_3-p_4)^2 \ .
    \nonumber
\end{alignat}

\noindent
We can write the generalized Feynman integral as
\eq{
    \label{I_pentagon}
    I(d_0,\nu;\wtz) =
    c(d_0,\nu) \, (-s_{12})^{d_0/2-\e-|\nu|-5\e\d}
    \int_0^\infty \G(\wtz;x)^{\e-d_0/2} \, x_1^{\nu_1+\e\d} \cdots x_5^{\nu_5+\e\d}
    \frac{\dd x}{x}
    \ ,
}
with
\eq{
    \G(\wtz;x) =
    \sum_{i=1}^5 \wtz_i x_i +
    \wtz_6 x_1 x_3 + \wtz_7 x_1 x_4 + \wtz_8 x_1 x_5 + 
    \wtz_9 x_2 x_4 + \wtz_{10} x_2 x_5 + \wtz_{11} x_3 x_5
    \ ,
}
and the corresponding $A$ matrix $A_5$ is given in \exref{ex:test5}.\\
\noindent
The $s_{12}$-dependent prefactor in \eqref{I_pentagon} comes from rescaling integration variables by $x_i \to x_i/(-s_{12})$.
In the Lee-Pomeransky representation, the monomial coefficients are given by
\eq{
    &
    \wtz_1 = \ldots = \wtz_6 = 1
    \ ,  \
    \wtz_7 = 1+y_2+y_4
    \ , \
    \wtz_8 = y_1
    \ , \
    \wtz_9 = y_4 \ , \\
    &
    \wtz_{10} = -1+y_1-y_2-y_3
    \ ,  \
    \wtz_{11} = -1+y_1-y_2-y_3-y_4-y_5
    \ , 
    \nonumber
}
where the $y_i$ are ratios of kinematic variables:
\eq{
    y_1 = \frac{p^2}{s_{12}}
    \ , \
    y_2 = \frac{s_{13}}{s_{12}}
    \ , \
    y_3 = \frac{s_{14}}{s_{12}}
    \ , \
    y_4 = \frac{s_{23}}{s_{12}}
    \ , \
    y_5 = \frac{s_{24}}{s_{12}}
    \ .
}
The Lee-Pomeransky polynomial hence contains 5 monomial coefficients different from unity.
Similar to the previous example, we can obtain equally many non-unity GKZ variables via homogeneity.
In particular, we start with $N=11$ generic variables $z_i$, after which we rescale $n+1=5+1=6$ GKZ variables as $z_1 = \ldots z_6 = 1$, leaving us with 5 variables $z_7, \ldots, z_{11}$.

\paragraph{Basis}

Solving an IBP system numerically, we can swiftly identify a set of master integrals.
Simplifying notation to $I(d_0,\nu)$, we choose
\begin{alignat}{3}
   \nonumber
   e_1 &= (-s_{12})^\e I(2,0,0,1,0,1)
   && e_2 &&= (-s_{12})^\e I(2,0,1,0,0,1) \\ \nonumber
   e_3 &= (-s_{12})^\e I(2,0,1,0,1,0)
   && e_4 &&=  (-s_{12})^\e I(2,1,0,0,0,1) \\ \nonumber
   e_5 &= (-s_{12})^\e I(2,1,0,0,1,0) 
   && e_6 &&= (-s_{12})^\e I(2,1,0,1,0,0) \\ \label{I_basis_pentagon}
   e_7 &= \e (-s_{12})^\e I(4,1,0,1,0,1) 
   && e_8 &&= \e (-s_{12})^{\e+1} I(4,0,1,1,1,1) \\ \nonumber
   e_9 &= \e (2\e-1) (-s_{12})^\e I(6,1,0,1,1,1) \qquad
   && e_{10} &&= \e (-s_{12})^{\e+1} I(4,1,1,0,1,1) \\ \nonumber
   e_{11} &= \e (-s_{12})^{\e+1} I(4,1,1,1,0,1) 
   && e_{12} &&= \e (-s_{12})^{\e+1} I(4,1,1,1,1,0) \\ \nonumber
   e_{13} &= \e^2 (-s_{12})^{\e+1} I(6,1,1,1,1,1) \, . && &&
\end{alignat}
The dimensions $d_0$ are chosen so as to have non-negative $r$-vectors, as per the ending remark of  \secref{ssection:ints_in_D-mod}.
Using \eqref{P(q)} (or formula \eqref{P(q)_simplex} employing homogeneity), we find the $\D$-module basis corresponding to $e$.
We write $e^{(\D)} = \Lambda \cdot {e'}^{(\D)}$ for a diagonal matrix $\Lambda_{ij}$ containing prefactors,
\eq{
    \Lambda_{ii} \defas
    c \big( d_0^{(i)}, \nu^{(i)} \big) \, 
    (-s_{12})^{d_0^{(i)}/2 - \e - \left|\nu^{(i)} \right| - 5 \e \d} \times
    \Lambda'_{i}
    \ ,
}
where $(d_0^{(i)},\nu^{(i)})$ refers to the indices of the $i$-th integral in \eqref{I_basis_pentagon} and%
\footnote{
    The numerators in $\Lambda'_i$ come from \eqref{I_basis_pentagon}, while the denominators arise due to taking derivatives of $\G^\e$ when constructing $P$ such that $P\langle\omega_0\rangle = \langle\omega_q\rangle$, in the notation of \eqref{P(q)}.
}
\eq{
    &
    \Lambda'_{1} = (-s_{12})^\e \frac{1}{\e} \ \ , \ \
    \Lambda'_{2} = (-s_{12})^\e \frac{1}{\e} \ , \ \
    \Lambda'_{3} = (-s_{12})^\e \frac{1}{\e} \ \ , \ \
    \Lambda'_{4} = (-s_{12})^\e \frac{1}{\e} \ \ , \ \
    \\ \nonumber
    &
    \Lambda'_{5} = (-s_{12})^\e \frac{1}{\e} \ \ , \ \
    \Lambda'_{6} = (-s_{12})^\e \frac{-1}{\e} 
    \Lambda'_{7} = \e (-s_{12})^\e \frac{1}{(\e-1)\e} \ \ , \ \
    \Lambda'_{8} = \e (-s_{12})^{\e+1} \frac{1}{(\e-1)\e} \ \ , \ \
    \\ \nonumber
    &
    \Lambda'_{9} = \e(2\e-1)(-s_{12})^\e \frac{-1}{(\e-2)(\e-1)\e} \ \ , \ \
    \Lambda'_{10} = \e (-s_{12})^{\e+1} \frac{1}{(\e-1)\e} 
    \Lambda'_{11} = \e (-s_{12})^{\e+1} \frac{-1}{(\e-1)\e} \ \ , \ \
    \\ \nonumber
    &
    \Lambda'_{12} = \e (-s_{12})^{\e+1} \frac{-1}{(\e-1)\e} \ \ , \ \
    \Lambda'_{13} = \e^2 (-s_{12})^{\e+1} \frac{-1}{(\e-2)(\e-1)\e} 
    \ .
}
Moreover,
\begin{alignat}{2}
\nonumber
&
e'^{(\D)}_1 &&= 
\partial_{11} 
\ , \ 
 e'^{(\D)}_2 = 
\partial_{10} 
\ , \ 
 e'^{(\D)}_3 = 
\partial_9 
\ , \ 
 e'^{(\D)}_4 = 
\partial_8 
\ , \ 
 e'^{(\D)}_5 = 
\partial_7 
 \\[5pt] \nonumber 
 & 
 e'^{(\D)}_6 &&= 
\e (5 \d +1)+z_7 \partial_7+z_8 \partial_8+z_9 \partial_9+z_{10} \partial_{10}+z_{11} \partial_{11} 
 \\[5pt] \nonumber 
 & 
 e'^{(\D)}_7 &&= 
(4 \d  \e +\e +1) \partial_{11} + z_{11} \partial_{11}^2+z_9 \partial_9 \partial_{11}+z_{10} \partial_{10} \partial_{11} 
 \\[5pt] 
 & 
 e'^{(\D)}_8 &&= 
\partial_9 \partial_{11} 
 \\[5pt] \nonumber 
 & 
 e'^{(\D)}_9 &&= 
\d \e \left(4 \d \e +1\right)\partial_{11} +
\d  z_{11} \e  \partial_{11}^2 +
z_7 \left(4 \d \e + \e +1\right) \partial_7 \partial_{11}  +
z_9 \left(5 \d \e + \e +2\right) \partial_9 \partial_{11} +
\d  z_{10} \e  \partial_{10} \partial_{11} 
\\[5pt] \nonumber
&  && +
z_7 z_{11} \partial_7 \partial_{11}^2+z_9 z_{11} \partial_9 \partial_{11}^2+z_9^2 \partial_9^2 \partial_{11}+z_7 z_9 \partial_7 \partial_9 \partial_{11}+z_7 z_{10} \partial_7 \partial_{10} \partial_{11}+z_9 z_{10} \partial_9 \partial_{10} \partial_{11} 
 \\[5pt] \nonumber
 & 
 e'^{(\D)}_{10} &&= 
\partial_7 \partial_{10} 
 \\[5pt] \nonumber 
 & 
 e'^{(\D)}_{11} &&= 
(5 \d  \e +\e +1) \partial_{10} + z_{10} \partial_{10}^2+z_7 \partial_7 \partial_{10}+z_8 \partial_8 \partial_{10}+z_9 \partial_9 \partial_{10}+z_{11} \partial_{11} \partial_{10} 
 \\[5pt] \nonumber 
 & 
 e'^{(\D)}_{12} &&= 
(5 \d  \e +\e +1) \partial_9 + z_9 \partial_9^2+z_7 \partial_7 \partial_9+z_8 \partial_8 \partial_9+z_{10} \partial_{10} \partial_9+z_{11} \partial_{11} \partial_9 
 \\[5pt] \nonumber 
 & 
 e'^{(\D)}_{13} &&= 
(4 \d  \e +\e +2) \partial_{11} \partial_9 + z_9 \partial_{11} \partial_9^2+z_{10} \partial_{10} \partial_{11} \partial_9+z_{11} \partial_{11}^2 \partial_9 \, .
\end{alignat}

\paragraph{Macaulay and Pfaffian matrices}

The Macaulay matrix $\Mext$ has dimensions $189 \times 113$, so in this case there are $113$ exterior monomials $\Ext$.
The basis of 13 standard monomials is given by
\eq{
    \estd =
    \big(
        \partial_9 \partial_{11}^2 , \,
        \partial_9^2 , \,
        \partial_{10}^2 , \,
        \partial_8 \partial_{11} , \,
        \partial_9 \partial_{11} , \,
        \partial_{10} \partial_{11} , \,
        \partial_{11}^2 , \,
        \partial_7 , \,
        \partial_8 , \,
        \partial_9 , \,
        \partial_{10} , \,
        \partial_{11} , \,
        1
    \big)^T
}
We identify $133$ independent rows of $\Mext$ by row reducing it numerically.
We therefore construct an unknown matrix $C$ of dimensions $13 \times 113$ which must satisfy 
$C_{\Ext,i}'-C \cdot N_{\Ext} = 0$, 
where $N_{\Ext}$ contains the $133$ independent rows of $\Mext$.
This linear system can be solved in reasonable time on a laptop using \textsc{FiniteFlow} \cite{url-FiniteFlow}.
The Pfaffian matrix in direction $i$ then follows from 
$C_{\Std,i}'-C \cdot N_{\Std} = \Pstd_i$,
where $N_{\Std}$ is built from $\Mstd$ by taking the same rows as in $N_{\Ext}$.

Using the basis change algorithm outlined in \secref{ssec:basis-change}, we perform a gauge transformation to relate $\Pstd$ to $P_i^{(e)}$, written in the basis $e$.
At this stage, we may safely take the limit $\d \to 0$.
We have verified that the resulting Pfaffian matrices are in agreement with \textsc{LiteRed} \cite{Lee:2012cn, Lee:2013mka}.
\hfill $\blacksquare$ \\ 

A similar analysis can be extended to diagrams with either more legs or more loops. For the purposes of this first work on this subject, we limited the application to one-loop integrals up to six external legs. 
In fact, \exref{ex:test9} refers to the one-loop massless hexagon.
Application to higher-loop cases will be discussed in future works.

\section{Linear relations for generalized Feynman integrals}
\label{sec:linear-relations}
Let $e = (I_1, \ldots, I_r)$ be a basis of master integrals, where each $I_i$ is a generalized Feynman integral of the form \eqref{gen_feyn_int}, i.e depending on generic variables $z_i$.
The vector $e$ will be regarded as the column vector in the sequel; when the distinction between row and column vectors is clear from context, we will omit the transpose symbol ${\bullet}^T$.

Denote by $e' = (I'_1, \ldots, I'_{r'})$ a set of integrals in the same family as $e$.
In both $e$ and $e'$, we replace the integrand factors
$x^{\nu+\e\d}$ by $x_1^{\nu_1 + \e\d_1} \cdots x_n^{\nu_n+\e\d_n}$
for a set of new parameters $\d_i$.

\begin{theorem}
    \label{th:recurrence}
    \hfil
    \begin{enumerate}
        \item There exists a matrix $U \in \mathbb{Q}^{r' \times r}(z)$ such that
        \eq{
            \label{epUe}
            e' = U \cdot e
            \ .
        }
        \item Let $L(e') = (\ell_1 I'_1, \ldots, \ell_{r'} I'_{r'})$, where the $\ell_i$ are differential operators w.r.t. $z$ with rational function coefficients.
        There exists a matrix $V \in \mathbb{Q}^{r' \times r}(z)$ such that
        \eq{
            \label{LepVe}
            L(e') = V \cdot e
            \ .
        }
    \end{enumerate}
    Moreover, there are construction algorithms for $U$ and $V$.
\end{theorem}

\begin{remark}
    \rm
    This theorem can be regarded as an analogue of IBP relations for generalized Feynman integrals.
    Recall the rank $r$ for generalized Feynman integrals is possibly larger than for conventional Feynman integrals.
    However, once we set the GKZ variables $z_i$ equal to their physically relevant values in e.g. $e' = U \cdot e$, some of the master integrals in $e$ might vanish or become equal, in which case we arrive at a conventional IBP relation.
\end{remark}

Although the theorem can be proved by a method analogous to Algorithm 1 of \cite{Matsubara-Heo-Takayama-2020b}, here we present a different approach based on Pfaffian matrices and the matrix factorial \cite{Tachibana-Goto-Koyama-Takayama-2020}.
We begin by constructing recurrence relations in the general framework of GKZ hypergeometic systems, after which we specialize to the case of generalized Feynman integrals.

\subsection{Recurrence relations for GKZ systems}

For convenience, in this section we will work with a rescaled version of the Euler integral \eqref{f_Gamma(z)}:
\eq{
    f_\Gamma(z) \to f(\beta) :=
    \frac{1}{\Gamma(\beta_0+1)} \int_\Gamma g(z;s)^{\beta_0} x_1^{-\beta_1} \cdots x_n^{-\beta_n} \frac{\dd x}{x}
    \ ,
}
where $\beta$ is assumed to be non-resonant.
We have the relation $\pd{i}f(\beta)=f(\beta-a_i)$ with this rescaling.
Since $z$ and $\Gamma$ will stay fixed in what follows, we only emphasize the dependence on $\beta$ in $f(\beta)$.
The integrals in $e$ and $e'$ are constant multiples of $f(\beta)$ with
\eq{  \label{eq:beta_to_e_delta}
  \beta=(\e,-\delta_1 \e,\ldots, -\delta_n \e)-(d_0/2,\nu) \, .
}
Let $s = \{s_1, \ldots, s_r\}$ be a basis of $\mathcal{D}_N/H_A(\beta)$ consisting of monomials in $\partial_i$.
The vector $F(\beta) = \big(s_1 f(\beta), \ldots, s_r f(\beta) \big)$ is then a solution to a Pfaffian system.
Now, because $\partial_i s_j = s_j \partial_i$ and $\partial_i f(\beta) = f(\beta-a_i)$ we have that
\eq{
    \label{rec_minus_a}
    \partial_i F(\beta) = P_i(\beta) F(\beta) = F(\beta-a_i)
    \ .
}
In other words, $P_i(\beta)$ yields a difference equation for $F(\beta)$.
Moreover, since $\beta$ is non-resonant, the matrix $P_i(\beta+a_i)$ is invertible, in which case we also have
\eq{
    \label{rec_plus_a}
    Q_i(\beta) F(\beta) = F(\beta+a_i)
    \quad , \quad
    Q_i(\beta) := P_i(\beta+a_i)^{-1}
    \ .
}
To make the analogy clear between equations \eqref{rec_minus_a} and \eqref{rec_plus_a}, let us introduce the operator $\partial^{-1}_i$ acting as
\eq{
    \label{partial_inv}
    \partial^{-1}_i F(\beta) = Q_i(\beta) F(\beta) = F(\beta+a_i)
    \ .
}

We define the \emph{falling matrix factorial} as
\eq{
    [P_i(\beta)]_j :=
    P_i(\beta-(j-1)a_i) \cdots P_i(\beta-a_i)P_i(\beta)
    \quad , \quad
    j > 0
    \ ,
}
and the \emph{rising matrix factorial} as
\eq{
    \big( Q_i(\beta)\big)_j :=
    Q_i(\beta+(j-1)a_i) \cdots Q_i(\beta+a_i)Q_i(\beta)
    \quad , \quad
    j > 0
    \ .
}
These matrix products can be swiftly computed using rational reconstruction over finite fields.

When the subscript $j$ is replaced by an integer vector $\kappa \in \NN_0^N$, we extend the definition of matrix factorials%
\footnote{
    When $i=N$ in the product we use the convention $\sum_{j=N+1}^N \kappa_j a_j = 0$.
} to:
\eq{
    \label{P(beta)_kappa}
    [P(\beta)]_\kappa
    & :=
    \prod_{i=1}^N \
    \Big[P_i \Big( \beta - \sum_{j=i+1}^N \kappa_j a_j  \Big) \Big]_{\kappa}
    \ ,
    \\
    \big( Q(\beta) \big)_\kappa
    & :=
    \prod_{i=1}^N \
    \Big(Q_i \Big( \beta - \sum_{j=i+1}^N \kappa_j a_j  \Big) \Big)_{\kappa}
    \ .
}

\begin{lemma} \label{lem:matrix_fac}
For $\kappa \in \NN_0^N$ we have the recurrence relations
\eq{
    \label{lem:matrix_fac_falling}
    F\Big( \beta - \sum_{i=1}^N \kappa_i a_i \Big)
    &=
    [P_i(\beta)]_\kappa \, F(\beta) = \partial^\kappa F(\beta)
    \ ,
    \\
    F\Big( \beta + \sum_{i=1}^N \kappa_i a_i \Big)
    &=
    \big( Q_i(\beta) \big)_\kappa \, F(\beta) = \partial^{-\kappa} F(\beta)
    \ .
    \label{lem:matrix_fac_rising}
}
\end{lemma}

\begin{proof}

By induction, let us derive that
\eq{
    \label{matrix_fac_induction}
    \partial^j_i F(\beta) =
    F(\beta - j a_i) =
    [P_i(\beta)]_j \, F(\beta)
    \ .
}
The case $j=1$ is shown in equation \eqref{rec_minus_a}.
Suppose that \eqref{matrix_fac_induction} holds for $j-1$.
Then we have
\begin{subequations}
\begin{alignat}{2}
    \partial_i \left( \partial^{j-1}_i F(\beta) \right)
    &=
    \partial_i F(\beta-(j-1)a_i) && \qquad \text{\eqref{matrix_fac_induction}}
    \\
    &=
    P_i(\beta-(j-1)a_i) \, F(\beta-(j-1)a_i) && \qquad \text{(Pfaffian equation)}
    \\
    &=
    P_i(\beta-(j-1)a_i) \, [P_i(\beta)]_{j-1} \, F(\beta) && \qquad \text{\eqref{matrix_fac_induction}}
    \\
    &=
    [P_i(\beta)]_j \, F(\beta) && \qquad \text{\eqref{P(beta)_kappa}}
    \ .
\end{alignat}
\end{subequations}
Applying \eqref{matrix_fac_induction} to each $(\kappa)_j$, we obtain the relation \eqref{lem:matrix_fac_falling} for the falling matrix factorial.
The case of the rising matrix factorial \eqref{lem:matrix_fac_rising} is proved in a similar fashion.

\end{proof}
\begin{remark}
    \rm
    Comparing
    $
        F\big((\beta-a_i)-a_j\big) =
        P_i(\beta-a_j) \, P_j(\beta) \, F(\beta)
    $
    to
    $
        F\big((\beta-a_j)-a_i\big) =
        P_j(\beta-a_i) \, P_i(\beta) \, F(\beta) \, ,
    $
    we note that the commutation relation $P_i(\beta-a_j) \, P_j(\beta) = P_j(\beta-a_i) \, P_i(\beta)$ holds.
    The explicit matrix factorial representation of \eqref{lem:matrix_fac_falling} is hence not unique.
    The same statement holds for \eqref{lem:matrix_fac_rising}.
\end{remark}

Observing that $\big( F(\beta) \big)_i = s_i f(\beta)$, we can write $F(\beta)$ as
\eq{
    F(\beta) = \big(f(\beta-A \cdot k_1), \ldots, f(\beta-A \cdot k_r) \big)
    \ ,
}
where the vectors $k_i \in \NN_0^N$ are fixed according to $s_i = \partial^{k_i}$.
Suppose that we want to obtain a recurrence relation for $f(\beta-A \cdot k_0)$ given some choice of $k_0 \in \ZZ^n$.
Notice that $f(\beta-A \cdot k_0)$ is the first element of $F(\beta-A\cdot(k_0-k_1))$.
We propose to obtain the recurrence relation by \Algref{alg:RMM}, whose
correctness follows from \namedref{Lemma}{lem:matrix_fac}.
Note that the algorithm does not perform differentiation with respect to $z_i$, meaning that it can still produce recurrence relations when the $z_i$ are fixed to numbers in the Pfaffian matrices.

\begin{algorithm}
    \begin{tabbing}
        \underline{Input}:
        Vector $k_0 \in \ZZ^N$,
        indeterminate vector $\beta$,
        monomial basis $s = \{s_1, \ldots, s_r\}$ where $s_i = \partial^{k_i}$, $k_i \in \NN_0^N$.
        \\[7pt]
        \underline{Output}:
        Recurrence relation
        $f(\beta-A \cdot k_0) = \sum_{i=1}^r u_i(\beta) f(\beta-A \cdot k_i)$,
        $u_i \in \QQ(\beta)$
    \end{tabbing}
    \begin{algorithmic}[1]
        \State Construct $P_i(\beta)$, $i=1,\ldots,r$, w.r.t basis $s$ by
        calling \Algref{alg:macaulay-matrix}
        \State Decompose $k_0 - k_1 = \kappa^+ - \kappa^- =: \kappa$ where $\kappa^\pm \in \NN_0^N$.
        \State Compute the matrix factorial
        $
            \big[P(\beta - \sum_{\kappa_j<0} \kappa_j \, a_j) \big]_{\kappa^+} \,
            \big(Q(\beta)\big)_{\kappa^-} \,
            F(\beta)
        $
        \State Output the first element of step 3.
    \end{algorithmic}
    \caption{: Recurrence relation by Macaulay matrix}
    \label{alg:RMM}
\end{algorithm}

\subsection{Recurrence relations for generalized Feynman integrals}

We can employ \Algref{alg:RMM} to find relations among generalized Feynman integrals by specializing $\beta$ to \eqref{eq:beta_to_e_delta}.
The matrices $Q_i = P_i^{-1}$ exist for this generic choice of $\e$ and $\d_i$'s.
In certain cases, we can even take a limit $\delta_i \rightarrow 0$ to obtain conventional IBP relations as we observe in the example below.\\

\subsection{Example: One-loop bubble}
\label{ex:1-loop-2-point}

\noindent
\begin{minipage}{12cm}
\paragraph{Setup}
We use \Algref{alg:RMM} to obtain a linear relation for the one-loop bubble integral with one massive propagator (see also Example $1.1$ of \cite{Klausen:2019hrg}), whose denominators are defined as
\eq{
    \prop_1 = -\ell^2 + m^2
    \quad , \quad
    \prop_2 = -(\ell+p)^2
    \ .
}
\end{minipage}
\begin{minipage}{3cm}
\includegraphics[scale=0.15]{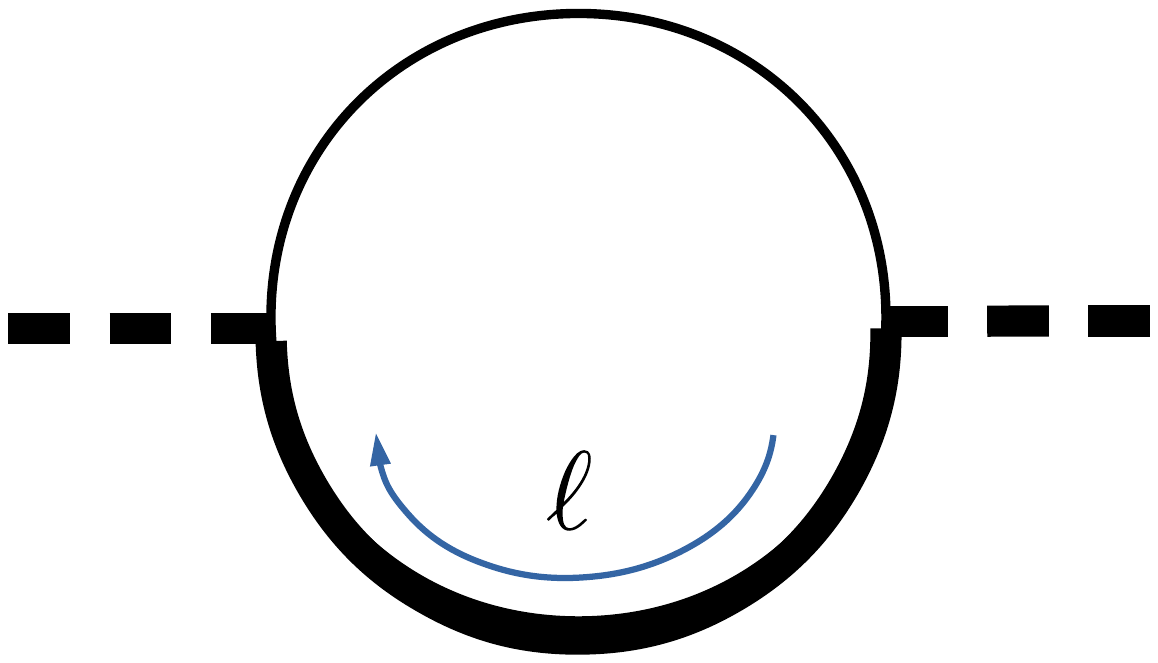}
\end{minipage}

\noindent
From the Lee-Pomeransky representation we read off
\eq{
    \G(z;x) = z_1 x_1 + z_2 x_2 + z_3 x_1 x_2 + z_4 x_1^2
    \quad , \quad
    \wtz = (1, 1, m^2-p^2, m^2)
    \ ,
}
and
\eq{
    A =
    \begin{pmatrix}
    1 & 1 & 1 & 1 \\
    1 & 0 & 1 & 2 \\
    0 & 1 & 1 & 0
    \end{pmatrix}
    \ .
}
We pick master integrals $e = \big(I(4,1,1), \> I(4,2,0) \big)$ and wish to decompose $e' = \big(I(4,1,2)\big)$.
For this diagram, the choice $\d_1 = \d_2 := \d$ is allowed.

In order to apply \Algref{alg:RMM}, we require $k_i$
corresponding to a relevant set of standard monomials. 
We use auxiliary vectors $k'_i$ and $\alpha$
such that $k_i$ correspond to master integrals $e$.
Their definitions are given by \eqref{eq:def_of_k_prime} and \eqref{alpha_max}
in the proof of \thmref{th:recurrence}.

\paragraph{Step 0: Input}

The master integrals correspond to the choice
\begin{alignat}{3}
    k'_1 &= (1,1,0,0)
    \quad && \text{s.t} \quad
    A \cdot k'_1 &&= (2,1,1)
    \ ,
    \\
    k'_2 &= (2,0,0,0)
    \quad && \text{s.t} \quad
    A \cdot k'_2 &&= (2,2,0)
    \ .
\end{alignat}
Given that these $k_i'$ only have positive entries, we have $\alpha = (0,0,0,0)$ using \eqref{alpha_max}.
Then the $A \cdot \alpha$ terms drop, and we have
\begin{alignat}{3}
    k_0 &= (0,1,1,0)
    \quad && \text{s.t} \quad
    A \cdot k_0 &&= (2,1,2)
    \ ,
    \\
    \beta &= \beta'
    \quad && \text{with} \quad
    \beta' &&= (\e,-\e\d,-\e\d)
    \ ,
\end{alignat}
using notation defined in \eqref{f_RMM}.
The vector $k_0$ represents the integral we want to reduce, $I(4,1,2)$.

\paragraph{Steps 1 and 2}

Calling \Algref{alg:macaulay-matrix}, we obtain Pfaffian matrices $P_i$, $i=1,\ldots,4$.
Moreover, we decompose
\begin{subequations}
\eq{
    k_0 - k_1 &= (-1,0,1,0) \\&=
    (0,0,1,0) - (1,0,0,0)   \\&=
    \kappa_+ - \kappa_-     \defas
    \kappa
    \ .
}

\end{subequations}
Note that $k_1=k'_1+\alpha=k'_1$ in this case.

\paragraph{Step 3}

Next, we are instructed to compute the matrix factorial corresponding to the operator $\partial_3^{+1} \partial_1^{-1}$, where the $\partial_3$ stems from $\kappa_+$ and the $\partial_1$ stems from $\kappa_-$.
Explicitly, we calculate
\eq{
    P_3(\beta+a_1) Q_1(\beta)
    \ .
}

\paragraph{Step 4: Output}

The recurrence relations follow from the 1st element of the matrix factorial above, apart from an adjustment of $\Gamma$-prefactors.
This adjustment is due to the constants $c$ and $c^{(i)}$ in equations \eqref{f_RMM_LHS} and \eqref{f_RMM_RHS}.

Given all of the above, \Algref{alg:RMM} produces the result
\begin{eqnarray}
    \nonumber
    I(4,1,2)
    &=&
    \frac
    { (-1)\big((4{\e}{\d}^{2}+(6{\e}-2){\d}+2{\e}-1){p}^{2}+(5{\e}{\d}^{2}+(6{\e}-1){\d}+2{\e}-1){m}^2\big) }
    { ({\d}+1)({\e}{\d}+1)({p}^2-{m}^{2})^2 }
    I(4,1,1)
    \\ && +
    \frac
    { (-1){m}^{2}(3{\d}+2) }
    { ({\d}+1)({p}^2-{m}^{2})^2 }
    I(4,2,0)
    \\[7pt] &=&
    \frac{(1-2\e)(p^2+m^2)}{(p^2-m^2)^2} I(4,1,1) -
    \frac{2m^2}{(p^2-m^2)^2} I(4,2,0)
    \quad \text{as} \quad \d \to 0
    \ ,
    \nonumber
\end{eqnarray}
in agreement with \textsc{LiteRed}.
\hfill $\blacksquare$

\paragraph{Critical case.}
We observe that our method may not give linear relations
for some $A$ and {\it non-generic} choices of $\e$ and $\delta_i$'s.
Put $\delta_1 = \ldots = \delta_n := \delta$
and take
\eq{
    A_2 =
    \begin{pmatrix}
        1 & 1 & 1 & 1 \\
        0 & 1 & 1 & 2 \\
        0 & 1 & 0 & 0
    \end{pmatrix}.
}
Define the linear form $b(\beta) := \beta_2 - \beta_3$, where $\beta_j$ denotes the $j$-th component of $\beta$.
We have $b(a_1) = b(a_2) = 0, \, b(a_3) = 1, \, b(a_4) = 2$.
Then $b(\beta)$ is a primitive supporting function of the facet $\{a_1, a_2\}$ of the convex polytope $\Delta_A$.
It is known that the supporting function gives a denominator
of $\pd{i}^{-1}$ modulo $H_A(\beta)$, $i=3,4$ \cite{SSTip}.
In fact, we have
$$
\big(   -    2  {z}_{3}  {z}_{4}  \partial_{4}  -    4  {z}_{1}  {z}_{4}  \partial_{3}+  (  -  2  {\beta}_{1}+  {\beta}_{2}+ {\beta}_{3})  {z}_{3} \big)  \partial_{3} \equiv    ( - 1)  (  {\beta}_{2}- {\beta}_{3})  (    2  {\beta}_{1}- {\beta}_{2}- {\beta}_{3})
\ \mathrm{mod}\, H_A(\beta)
\ ,
$$
and
$$
\big(   -    2  {z}_{3}  {z}_{4}  \partial_{3}  -   {z}_{3}^{ 2}   \partial_{1}+  (  -  2  {\beta}_{2}+   2  {\beta}_{3}+ 2)  {z}_{4} \big)  \partial_{4} \equiv    ( - 1)  (  {\beta}_{2}- {\beta}_{3})  (   {\beta}_{2}- {\beta}_{3}- 1)
\ \mathrm{mod}\, H_A(\beta)
\ .
$$
Note that $\partial_3$ and $\partial_4$ stand for $P_i$.
We have $b\big( (\e, -\e\d, -\e\d) \big)=0$, for which reason $Q_i$,
$i=3,4$ acquires a $0$ in the denominator.
Since $Q_i$ acts similarly to $\partial_i^{-1}$,
the recurrence relations involving $Q_i$ do not exist.
Moreover, $b(\beta)=0$ has the effect of breaking several isomphism theorems for $\mathcal{D}$-modules
(see, e.g., \cite{SSTip}, \cite[\S 4.5]{SST} and \cite{newBateman} for recent references on isomorphism theorems of GKZ systems). 
\hfill $\blacksquare$
\\ 

\noindent
We close this section with proving \thmref{th:recurrence}.
\begin{proof} (\thmref{th:recurrence}.)
We begin by noting that the matrices $Q_i$ exist since $P_i$ are invertible for generic parameters $\d_i$.

Without loss of generality, suppose that $r'=1$.
Then $e' = (I'_1)$ is given by an integral $I'_1 = I(d_0,\nu)$ in the notation of \secref{sec:macaulay-feynman}.
We week to reduce $I'_1$ in terms of master integrals $e = (I_1,\ldots,I_{r})$, where we use the notation
$I_i = I(d_0^{(i)} , \nu^{(i)})$.

Take vectors $k'_i \in \ZZ^N$ such that
\eq{  \label{eq:def_of_k_prime}
    A \cdot k'_i = \big( d_0^{(i)}/2, \nu^{(i)} \big)
    \ .
}
We decompose $k'_i$ into a sum of positive and negative parts: $k'_i = \kappa^{+}_i - \kappa^{-}_i$ for
$\kappa^{\pm}_i \in \NN_0^N$.
Then we may define the vector
\eq{
    \label{alpha_max}
    \alpha =
    \Big(
        \max_{1 \, \leq \, i \, \leq \, r'} \kappa^{-}_{i,1}
        \, , \ldots ,
        \max_{1 \, \leq \, i \, \leq \, r'} \kappa^{-}_{i,N}
    \Big)
    \ .
}
where $\kappa^-_{i,j}$ is the $j$-th component of $\kappa^-_i$.
Put%
\footnote{
    $\alpha$ is constructed in order to have a common shift for all $k'_i$ such that their entries become non-negative.
}
$k_i = k'_i + \alpha \in \NN_0^N$ and take $k_0 \in \ZZ^N$ such that $A \cdot k_0 = A \cdot \alpha+(d_0/2,\nu)$.
We apply \Algref{alg:RMM} using input $k_0$, $\beta = \beta' + A \cdot \alpha$ for $\beta' = (\e, -\e\d_1, \ldots, -\e\d_n)$ and $s_i = \partial^{k_i}$.
The output of the algorithm takes the form
\eq{
    \label{f_RMM}
    f(\beta - A \cdot k_0) =
    \sum_{i=1}^{r'} u_i(\beta) f(\beta - A \cdot k_i)
    \quad , \quad
    u_i \in \QQ(\beta)
    \ .
}
Let us inspect the expressions for $f$ on both sides of \eqref{f_RMM}.
On the LHS,
\eq{
    f(\beta-A \cdot k_0) &=
    f\left(\beta'-(d_0/2,\nu)\right) \\&=
    I(d_0,\nu) \, \big / \, c
    \label{f_RMM_LHS}
    \ .
}
So we obtain the generalized Feynman integral $I'_1$ apart from constant prefactors $c = c(d_0,\nu)$ defined in \eqref{cd0nu}.
On the RHS,
\begin{subequations}
\eq{
    f \left(\beta - A \cdot k_i \right) &=
    f \left(\beta - \big( d_0^{(i)}/2, \nu^{(i)} \big) - A \cdot \alpha \right) \\&=
    f \left(\beta' - \big( d_0^{(i)}/2, \nu^{(i)} \big) \right) \\&=
    I \big( d_0^{(i)},\nu^{(i)} \big) \, \big / \, c^{(i)}
    \label{f_RMM_RHS}
    \ .
}
\end{subequations}

In other words, we obtain $I_i$ apart from $\Gamma$-prefactors $c^{(i)} = c^{(i)}\big( d_0^{(i)}, \nu^{(i)} \big)$.
The coefficients $u_i(\beta)$ multiplied by $\Gamma$-prefactors
give the matrix $U$.
Thus, we conclude Statement 1. of \thmref{th:recurrence}.

Statement 2 can be proven by noting that $\pd{i}$ induces the parameter shift 
and applying Statement 1.
\end{proof}

\section{Decomposition via cohomology intersection numbers}
\label{sec:intersection}
 Relations between Feynman integrals, equivalent to IBP identities, and, more generally, identities for Euler-Mellin integrals, equivalent to contiguity relations, can be derived by means of intersection theory for twisted de Rham cohomologies~\cite{Mastrolia:2018uzb, Frellesvig:2019kgj, Frellesvig:2019uqt, Frellesvig:2020qot}. According to the mentioned algorithm, the decomposition of any given integrals in terms of an independent basis of MIs can be obtained from the projection of the twisted differential form appearing in the integrand of the integral to decompose into a basis of differential forms that generate a de Rham twisted cohomology group, {\it via} intersection numbers. \\
For the case of generalized Feynman integrals
\eqref{eq:generalized_Feynman_integral_definition}, the covariant derivative \eqref{covariant_derivative_GKZ} reads
\begin{equation}
\label{covariant_derivative_Feynman_integrals}
    \nabla_{x} = \dd_x + \epsilon \frac{\dd_x \mathcal{G} }{\mathcal{G}} \wedge +  \, \epsilon \, \delta \,  \sum_{i=1}^{n} \, \frac{\dd x_i}{x_i} \wedge,
\end{equation}
and we denote the associated $n$-th \emph{de Rham cohomology group} as $\mathbb{H}^{n}$ (see also \eqref{de_Rham_cohomology_group}). We can also introduce a \emph{dual} covariant derivative $\nabla_x^{\vee} = \nabla_x \big|_{\epsilon \to - \epsilon}$ and let $\mathbb{H}^{n \vee}$ be the $n$-th \emph{(dual) de Rham cohomology group} associated to it. The \emph{cohomology intersection number}
\begin{equation}
     \langle \bullet, \bullet \rangle_{ch}: \, \mathbb{H}^{n}  \times \mathbb{H}^{n \vee} \to \mathbb{C}(z),
 \end{equation}
 is a natural pairing between the elements of the two groups.\\

Let $\{ e_i \}_{i=1}^{r}$ be a basis for $\mathbb{H}^{n}$ and $\{h_i \}_{i=1}^{r}$ a basis for $\mathbb{H}^{n \vee}$; the decomposition of any twisted form $\varphi \in \mathbb{H}^{n}$ in terms of $\{ e_i \}_{i=1}^{r}$ can be obtained via chomology intersection numbers according to the master decomposition formula~\cite{Mastrolia:2018uzb, Frellesvig:2019kgj, Frellesvig:2019uqt, Frellesvig:2020qot},
\begin{equation}
    \varphi = \sum_{i=1}^{r} c_i \, e_i \,
    \ , 
    \qquad {\rm with} \quad 
    c_i := \sum_{j=1}^r \, 
    \langle \varphi, h_j \rangle_{ch} \, \left( I^{-1}_{ch} \right)_{ji}, \quad 
    {\rm and} \qquad
    \left(I_{ch}\right)_{ij} := \langle e_i, h_j \rangle_{ch}\, .
    \label{eq:master_deco}
\end{equation}
This formula implies the decomposition of (generalized) Feynman integrals in terms of master integrals, upon the identification in \eqref{eq:generalized_Feynman_integral_definition}.\\
~\\
Looking at \eqref{eq:master_deco} we infer that {\it two} distinct sets of intersection numbers are required, 
namely $ \{ \langle \varphi, h_i \rangle_{ch} \}_{i=1}^{r}$ and $\{ \left(I_{ch}\right)_{ij} \}_{i,j=1}^{r}$. \ 
Therefore, in order to apply the decomposition formula \eqref{eq:master_deco}, it is required the determination 
of the matrix $I_{ch}$ and of the vector $\langle \varphi, h_i \rangle_{ch}$.

According to the algorithm proposed in \cite{Matsubara-Heo-Takayama-2020b}, given a basis $\{ e_i  \}_{i=1}^{r}$ for $\mathbb{H}^{n}$ ($\{ h_i \}_{i=1}^{r}$ for $\mathbb{H}^{n \vee}$)  and $\{P_i \}_{i=1}^{N}$ the associated set of Pfaffian matrices (resp $ \{ P^{\vee}_i \}_{i=1}^{N}$), the cohomology intersection matrix  $ \left( I_{ch} \right)_{ij} = \langle e_i, h_j \rangle_{ch} $ can be obtained as a {\it rational solution} of the so called 
{\it secondary equation},
\begin{equation}
    \partial_i I_{ch} = P_i \cdot I_{ch} + I_{ch} \cdot \left( P_i^{\vee } \right)^{T},
    \qquad \quad i =1, \dots, N ;
\label{eq:secondary_equation}
\end{equation}
which is a (system of partial) differential equation(s) for the cohomology intersection matrix $I_{ch}$, controlled by the Pfaffian matrices
$\{ P_i \}_{i=1}^{N}$ and its dual $\{ P_i^{\vee} \}_{i=1}^{N}$. 
In particular, any non-zero rational solution of the secondary equation $I$ is related to $I_{ch}$ as
\begin{eqnarray}
    I_{ch} = \kappa_0 \, I
\end{eqnarray}
up to a constant $\kappa_0$,
corresponding to a boundary value, to be independently provided, see theorem 2.1 of \cite{matsubaraheo2019algorithm}.
\\

Remarkably, also the vector $\langle \varphi, h_i \rangle_{ch}$ can be obtained by a second application of the same algorithm to a suitably chosen set of differential forms. 
For this purpose,
it is sufficient to build the \emph{auxiliary} basis%
\footnote{
    The set $\{ e_i \}_{i=1}^{r-1} \bigcup \{ \varphi\} $ fails to be a basis when $\varphi=0$.
    In the following discussion, we implicitly assume that the cohomology class $\varphi$ is not zero and the set $\{ e_i \}_{i=1}^{r-1} \bigcup \{ \varphi\} $ is a basis.
    This can be achieved when $\varphi$ is represented by a {\it monomial differential form} $\omega_{d_0/2,\nu}$.
}
, 
$\{ e^{\text{aux}}_i \}_{i=1}^{r} := \{ e_i \}_{i=1}^{r-1} \bigcup \{ \varphi\} $,
as well as the \emph{auxiliary} dual basis
$\{ h^{\text{aux}}_i \}_{i=1}^{r} := \{ h_i \}_{i=1}^{r}$, and the corresponding Pfaffian matrices, say 
 $\{ P^{\text{aux}}_i \}_{i=1}^{N}$ and 
 $\{ P^{\vee \, \text{aux}}_i \}_{i=1}^{N}$, where, owing to the dual bases choice, 
 $P^{\vee \, \text{aux}}_i \equiv 
  P^{\vee}_i $, $\forall i \, | \, 1 \leq i \leq N$. 
Then, the cohomology intersection matrix for the auxiliary bases, $I^{\text{aux}}_{ch}$, whose elements are defined as 
$ \left( I^{\text{aux}}_{ch} \right)_{ij} = \langle e^{\text{aux}}_i, h^{\text{aux}}_j \rangle = \langle e^{\text{aux}}_i, h_j \rangle $,
can be obtained as rational solution of this matrix differential equation,
\begin{equation}
\begin{split}
    \partial_i I^{\text{aux}}_{ch} &  = P^{\text{aux}}_i \cdot I^{\text{aux}}_{ch} + I^{\text{aux}}_{ch} \cdot \left( P_i^{\vee \, \text{aux}} \right)^{T} \\
    & = P^{\text{aux}}_i \cdot I^{\text{aux}}_{ch} + I^{\text{aux}}_{ch} \cdot \left( P_i^{\vee} \right)^{T},
    \qquad \quad i =1, \dots, N \,.
\end{split}
\label{eq:secondary_equation_auxiliary}
\end{equation}
The wanted vector of intersection numbers 
$\langle \varphi, h_j \rangle_{ch}$,
can be read off the $r^{th}$-row of 
$I^{\text{aux}}_{ch}$, because 
\begin{equation}
    \left( I^{\text{aux}}_{ch} \right)_{rj}=\langle \varphi, h^{\text{aux}}_j \rangle_{ch} = \langle \varphi, h_j \rangle_{ch} , \qquad \quad j=1, \dots, r  .
\end{equation}
Actually the solution of the auxiliary secondary equation,
\begin{eqnarray}
    I_{ch}^{\rm aux} = \kappa_0^{\rm aux} I^{\rm aux}
\end{eqnarray}
also requires the independent knowledge of the constant $\kappa_0^{\rm aux}$.

The proposed twofold procedure, yields the determination of all the cohomology intersection numbers needed in \eqref{eq:master_deco},
up to the knowledge of the ratio of two boundary constants 
$\kappa_0^{\rm aux}/\kappa_0$.

Let us observe that the master decomposition formula \eqref{eq:master_deco} requires the knowledge of the inverse matrix $I_{ch}^{-1}$. By using the secondary equation \eqref{eq:secondary_equation}, 
one can show that $I_{ch}^{-1}$ can be directly determined as rational solution of 
the following matrix differential equation,
\begin{eqnarray}
    \partial_i (I_{ch}^{-1}) = 
    -(P_i^{\vee})^{T} \cdot (I_{ch}^{-1}) - 
    (I_{ch}^{-1}) \cdot P_i .
\end{eqnarray}

\paragraph{Determination of coefficients.}
Let us apply the master decomposition  \eqref{eq:master_deco} to a vector formed by $r$ differential forms $e_1, e_2, \ldots, e_{r-1}, \varphi$, reading as,
\begin{equation}
\label{eq:masterdecomatrix}
    \begin{pmatrix}
    e_1\\
    \vdots\\
    e_{r-1}\\
    \varphi
    \end{pmatrix}
    =
    I_{ch}^{\rm aux} \cdot I_{ch}^{-1}
    \begin{pmatrix}
    e_1\\
    \vdots\\
    e_{r-1}\\
    e_r
    \end{pmatrix}
    \ .
\end{equation}
By construction, the product 
$I_{ch}^{\rm aux} \cdot I_{ch}^{-1}$ has the following form,
\begin{equation}\label{eqn:172}
    I_{ch}^{\rm aux} \cdot I_{ch}^{-1}=
    \left(
    \begin{array}{ccc|c}
    &&&0\\
    &{{\mathbb I}_{r-1}}&  &\vdots \\
    &&&0\\
    \hline
    c_1&\cdots&c_{r-1}&c_r
    \end{array}
    \right)
     \ .
\end{equation}
where ${\mathbb I}_{r-1}$ is the identity matrix in the $(r-1)\times(r-1)$ space,
and the entries of the last row are the coefficients of the decomposition \eqref{eq:master_deco}.
On the other side, owing to the solution of the secondary equations, one has,
\begin{eqnarray}
   I_{ch}^{\rm aux} \cdot I_{ch}^{-1} 
   = {\kappa_0^{\rm aux} \over \kappa_0} \, I^{\rm aux} \cdot I^{-1} \ .
\end{eqnarray}
This relation among matrices can be exploited to fix the value of the ratio 
$\kappa_0^{\rm aux}/\kappa_0$. In fact,
any of the elements on the diagonal, say $(I_{ch}^{\rm aux} \cdot I_{ch}^{-1})_{kk}$ , for any $k \in \{1, \ldots, r-1 \}$, amounts to 1, therefore,
\begin{eqnarray}
    {\kappa_0^{\rm aux} \over \kappa_0} \, (I^{\rm aux} \cdot I^{-1})_{kk} 
    = 1 \ , 
\end{eqnarray}
yielding the relation 
\begin{equation}
    {\kappa_0\over \kappa_0^{\rm aux} }
    = ( I^{\rm aux} \cdot I^{-1} )_{kk} \ ,
\end{equation}
namely fixing the ratio of the boundary constants from one of the elements of the product of two known matrices.
Finally, \eqref{eq:masterdecomatrix} becomes,
\begin{equation}
    \begin{pmatrix}
    e_1\\
    \vdots\\
    e_{r-1}\\
    \varphi
    \end{pmatrix}
    =
    {1 \over (I^{\rm aux} \cdot I^{-1})_{kk}}
    I^{\rm aux}I^{-1}
    \begin{pmatrix}
    e_1\\
    \vdots\\
    e_{r-1}\\
    e_r
    \end{pmatrix}
    \ ,
\end{equation}
hence the coefficients of decomposition 
of the differential form $\varphi$ 
in \eqref{eq:master_deco} reads as,
\begin{eqnarray}
\varphi = \sum_{i=1}^{r} c_i \, e_i \,
    \ , 
    \qquad {\rm with} \quad 
c_i = {1 \over (I^{\rm aux} \cdot I^{-1})_{kk}}
    (I^{\rm aux} \cdot I^{-1})_{r i} \ ,
\end{eqnarray}
for any chosen 
$k \in \{1, \ldots, r-1 \}$. 
This result bypasses the determination
of the complete cohomology intersection matrices, 
$I_{ch}$ and $I_{ch}^{\rm aux}$,
and relies only on the knowledge of the Pfaffian matrices, $P$ and $P^{\rm aux}$,
which control the secondary equations.
The determination of their {\it rational} solutions, $I$ and $I^{\rm aux}$, can be efficiently combined with the rational function reconstruction over finite fields, earlier discussed. 

Let us remark that, if needed, the individual expressions of the normalization constants $\kappa_0$, and 
$\kappa_0^{\rm aux}$ can be determined by means of
\cite[Theorem 8.1]{matsubaraheo2019euler}.

\subsection{Example: One-loop box}
We reconsider the example of \secref{sec:one_loop_box}; we aim to decompose:
\begin{equation}
    \e (-s)^{\epsilon} \, (-s)^3 z I(4,2,1,1,1;z),
\end{equation}
in terms of the basis given by \eqref{basis_box_int}.
The associated differential forms are:
\begin{equation}
    \varphi = \frac{\epsilon (-s)^{-4 \epsilon \delta} \, z \, \Gamma(2-\epsilon)}{\Gamma({-}1 {-} 2 \epsilon {-} 4 \epsilon \delta) \,  \Gamma(1+\epsilon \delta)^3 \,  \Gamma( 2 + \delta \epsilon)} \, \frac{x_1}{\G(z;x)^2} \, \dd x \, .
\label{eq:varphi_basis_massless_box}
\end{equation}
and, 
\begin{equation}
    \left( e_1, \, e_2, \, e_3 \right) = \eqref{basis_box_D} \, ,
\label{eq:phi_basis_massless_box}
\end{equation}
while we explicitly choose the basis of dual forms as:
\begin{equation}
   \left( h_1,  \, h_2, \, h_3 \right) =  \left( e_1, \, e_2, \, e_3 \right) \big|_{\epsilon \to - \epsilon} = \eqref{basis_box_D} \big|_{\epsilon \to - \epsilon} \, .
\end{equation}
The Pfaffian matrix associated to \eqref{basis_box_int} reads:
\begin{equation}
    P = \left(
\begin{array}{ccc}
 - \frac{ \epsilon \left(\delta ^2 (12 z+11)+7 \delta  (z+1)+z+1 \right)}{(3 \delta
   +1) z (z+1)} & -\frac{\delta ^2 \epsilon }{(3 \delta+1)(z+1)} & \frac{\delta ^2
   \epsilon  (\delta  (z+2)+1)}{2 (3 \delta +1) z
   (z+1) (\delta  \epsilon +1)} \\
 \frac{\delta ^2 \epsilon }{(3 \delta +1)
   z \left(z+1\right)} & -\frac{\delta ^2
   \epsilon }{(3\delta+1) (z+1)} &
   -\frac{\delta ^2 \epsilon  (\delta +2 \delta 
   z+z)}{2 (3 \delta +1) z (z+1) (\delta 
   \epsilon +1)} \\
 -\frac{2 (2 \delta +1) \epsilon  (\delta 
   \epsilon +1)}{(3 \delta +1) z (z+1)} & \frac{2
   (2 \delta +1) \epsilon  (\delta  \epsilon
   +1)}{(3 \delta +1) (z+1)} & -\frac{\epsilon 
   \left(\delta ^2 (5 z+7)+\delta  (2
   z+5)+1\right)}{(3 \delta +1) z (z+1)} \\
\end{array}
\right),
\end{equation}
while its dual is:
\begin{equation}
    P^{\vee} =  P \big|_{\epsilon \to - \epsilon}.
\end{equation}
A {\it rational solution} of \eqref{eq:secondary_equation} obtained via \cite{integrableconnection} is
\begin{equation}
\label{eq:C_mat_can_massless_box}
    I = 
\left(
\begin{array}{ccc}
 -\frac{(2 \delta +1) (4 \delta +1)}{\delta } &
   \delta  & -2 (\delta  \epsilon -1) \\
 \delta  & -\frac{(2 \delta +1) (4 \delta
   +1)}{\delta } & -2 (\delta  \epsilon -1) \\
 2 (\delta  \epsilon +1) & 2 (\delta  \epsilon
   +1) & -\frac{4 \left(10 \delta ^2+6 \delta
   +1\right) (\delta  \epsilon -1) (\delta 
   \epsilon +1)}{\delta ^3} \\
\end{array}
\right) .
\end{equation}
The {\it auxiliary} basis reads:
\begin{equation}
    \left( e^{\text{aux}}_1, \, e^{\text{aux}}_2, \, e^{\text{aux}}_3 \right) = \left( e_1, \, e_2, \, \varphi \right).
\end{equation}
The expression for $P^{\text{aux}}$ is too lengthy to be reported here. 
By solving \eqref{eq:secondary_equation_auxiliary}, we obtain $I^{\text{aux}}$. The entries of its last row read:
\begin{subequations}
\eq{
\left(I^{\text{aux}} \right)_{31} = &  \frac{2 (3 \delta  \epsilon +\epsilon +1) ((4 \delta +2) \epsilon +1) \left(\left(6 \delta ^2+5 \delta +1\right) \epsilon +\delta  z (2 \delta  \epsilon +\epsilon +1)\right)}{\delta  z (2 \delta  \epsilon
   +\epsilon +1)^2} \, ,\\
\left(I^{\text{aux}} \right)_{32}  = & \frac{2 ((4 \delta +2) \epsilon +1) \left(2 \left(6 \delta ^2+5 \delta +1\right) \epsilon ^2+(5 \delta +2) \epsilon +1\right)}{(2 \delta  \epsilon +\epsilon +1)^2} \, , \\
\left(I^{\text{aux}} \right)_{33}  =  & -\frac{4 (\delta  \epsilon -1) ((4 \delta +2) \epsilon +1)}{\delta ^3 z (2 \delta  \epsilon +\epsilon +1)^2}
\left(
\delta  \left(6 \delta ^2+5 \delta +1\right) \epsilon  (3 \delta  \epsilon +\epsilon +1)+z \left(66 \delta ^4 \epsilon ^2+\delta ^3 \epsilon  (91 \epsilon +50) \right. \right. \nonumber \\
& \left. \left.
 +\delta ^2 \left(47 \epsilon ^2+50 \epsilon
   +10\right)+\delta  \left(11 \epsilon ^2+17 \epsilon +6\right)+(\epsilon +1)^2 \right) \right)
 .
}
\end{subequations}
For the example under consideration, we find:
\begin{equation}
    \frac{\kappa_0}{\kappa^{\text{aux}}_0}=1.
\end{equation}
All the cohomology intersection numbers needed for the decomposition formula \eqref{eq:master_deco} are therefore computed:
\begin{subequations}
\eq{
 \langle \varphi, h_j \rangle_{ch} & =  \left( I^{\text{aux}}_{ch} \right)_{3j} \, ,
    \qquad \quad j=1, \dots, 3 \ ; \\
 \langle e_i, h_j \rangle_{ch} & = \left( I_{ch}   \right)_{ij} \, , 
 \qquad \quad i,j=1, \dots, 3 \ ,
}
\end{subequations}
and, in the $\delta \to 0$ limit, the resulting decomposition in terms of master forms reads,
\begin{equation}
    \varphi = -\frac{2 \epsilon  (2 \epsilon +1)}{z (\epsilon +1)} \cdot e_1 + 0 \cdot e_2 +( 2 \epsilon +1) \cdot e_3 \, ,
\end{equation}
which is in agreement with the IBP decomposition obtained with \textsc{LiteRed}.
To conclude, let us also observe that, the individual normalization constants $\kappa_0=\kappa_0(\delta,\epsilon)$, can be determined by means of
\cite[Theorem 8.1]{matsubaraheo2019euler}, as 
\begin{equation}
    \kappa_0=\frac{\sin^4(\delta\e \pi)\sin\left(2(2\delta+1)\e \pi\right)}{\varepsilon^22\delta(2\delta+1)^2(\e ^2\delta^2-1)\pi^4\sin(\pi\e) } \ ,
\end{equation}
which, combined with \eqref{eq:C_mat_can_massless_box}, gives the cohomology intersection matrix $I_{ch}$.
\hfill $\blacksquare$

\section{Conclusion}

In this work, we presented a simplified algorithm for the determination of Pfaffian matrices from Macaulay matrix, making use of the theory of ${\cal D}$-modules for Gel'fand-Kapranov-Zelevinsky systems.

Then, we introduced two algorithms for the derivation of {\it linear relations} for GKZ-hypergeometric integrals making use of Pfaffians: 
{\it i)} the first one uses the properties of the {\it matrix factorial}, within the holonomic gradient method;
{\it ii)} the second one uses the direct integral decomposition {\it via} cohomology {\it intersection numbers} and the rational solution of the {\it secondary equation} for cohomology intersection matrices.
In the latter case, we derived a novel variant of the master formula for the projection of differential forms onto a set of independent forms.

Our investigation exploited the {\it isomorphism} between de Rham twisted cohomology groups, whose elements are differential forms, and ${\cal D}$-modules, whose elements are partial differential operators in a Weyl algebra. 
Our results can be applied to GKZ-hypergeometric functions as well as to Feynman integrals, which can be considered {\it restrictions} of the former class of functions. 
Within our analyses the number of master integrals, corresponding to the dimension of the de Rham cohomology group, is related to the rank of the GKZ system, computed by polyhedra combinatorics in terms of volumes of polytopes.
We showcased a few applications to simple mathematical functions and one-loop integrals.

Within the standard approaches, linear relations among integrals, 
such as contiguity relation for GKZ-hypergeometric functions and 
integration-by-parts identities for Feynman integrals, are employed to derive systems of partial differential equations for a chosen set of independent function. Instead, in the current work, we reversed the perspective, and showed how Pfaffian matrices for system of partial differential equations, built from Maculay matrices, can be used to derive linear relations for integrals.

We hope that our study could offer a novel, more complete view on integral relations which emerge from the application of differential operators acting both on internal and external variables, and could provide additional insights to the investigation of de Rham twisted co-homology theory in computational (quantum) field theory.

\section*{Acknowledgements}
We wish to thank Tiziano Peraro for interesting discussions and for assistance with the use of {\sc FiniteFlow}.\\
We acknowledge the stimulating discussions stemming from the workshop {\it MathemAmplitudes 2019: Intersection Theory \& Feynman Integrals}, U. Padova,
\href{https://pos.sissa.it/383/}{[MA2019]}.
\\
F.G. is supported by Fondazione Cassa di Risparmio di Padova e Rovigo (CARIPARO). The work of M.K.M is supported by Fellini - Fellowship for Innovation at INFN funded by the European Union's Horizon 2020 research and innovation programme under the Marie Sk{\l}odowska-Curie grant agreement No 754496.
V.C. is supported by the {\it Diagrammalgebra} Stars Wild-Card Grant UNIPD.
S.M. and N.T. are supported in part by the JST CREST Grant Number JP19209317.
S.M. is supported by JSPS KAKENHI Grant Number 22K13930.
N.T. is supported in part by JSPS KAKENHI Grant Number JP21K03270.
\appendix

\section{Further details on the homogeneity property}
\label{sec:rescaling-details}
Here we provide additional details and formulas for \secref{ssec:rescaling}.
We begin by introducing simplified matrix notation to reduce the number of
indices in the formulas below.

Firstly, we give a generalization of the multivariate exponent from
\eqref{eq:multivar-exp}. Given an $\brk{n + 1} \times N$ matrix
$A \defas \brk{a_1, \ldots, a_N}$ and a list of variables
$t \defas \brk{t_1, \ldots, t_{n + 1}}$,
we construct the following list of exponentials%
\footnote{
    In \mathematica{} this can be achieved by \code{Inner[Power, $t$, $A$, Times]}.
}:
\begin{align}
    t^A \defas \brk{t^{a_1}, \ldots, t^{a_N}} \, .
    \label{eq:matrix-exp}
\end{align}
This matrix exponentiation obeys the rule $\brk{t^A}^B = t^{A B}$.

Secondly, we define an item-wise product of two lists having the same lengths. Given two
lists $t$ and $s$ of lengths $n$, we define%
\footnote{
   Using \mathematica{} syntax, we can write \code{Inner[Times, $s$, $t$, List]}.
}
\begin{align}
    s \times t \defas \brk{s_1 t_1, \ldots, s_n t_n} \, ,
\end{align}

Using the formulas above, we compactly rewrite the homogeneity property
\eqref{f_Gamma_homogeneity} of Euler integrals as follows:
\begin{align}
     \eulerInt\brk{t^A \times z} = t^{\beta} \eulerInt\brk{z} \, .
     \label{eq:homoegeneity-matrix}
\end{align}
Now we pick a set
\begin{align}
    \sigma \subset \brc{1,\ldots,N}, \quad
    \bsigma \defas \brc{1,\ldots,N} \setminus \sigma \, ,
\end{align}
where $\sigma$ is of size $\abs{\sigma} = n+1$, and $\eta$ is its
complement.
Denote by $A_\sigma$ the submatrix of matrix $A$ constructed from the
columns labeled by $\sigma$.
Given a vector $v$, we likewise construct $v_\sigma$. 
Using this notation, we can separate the following matrix product into two terms:
\begin{align}
    A v = A_\sigma v_\sigma + A_{\bsigma} v_{\bsigma} \, .
    \label{eq:Av-split}
\end{align}
We shall require $\sigma$ to be chosen such that the $A_\sigma$
submatrix is invertible, i.e. $\mathrm{det}(A_\sigma) \neq 0$.

Now observe that for the special choice of the rescaling parameters
$t = z_{\sigma}^{-A_{\sigma}^{-1}}$ in \eqref{eq:homoegeneity-matrix},
we obtain the following factorized representation of a given Euler integral:
\begin{align}
    \eulerInt\brk{z}
    = z_{\sigma}^{-A_{\sigma}^{-1} \beta} \, \eulerInt\bigbrk{z_{\sigma}^{-A_{\sigma}^{-1} A} \times z} \, .
    \label{eq:euler-int-rescaled}
\end{align}
The key feature of the RHS is that all arguments labeled by
$\sigma$ are set to $1$:
\begin{align}
    \bigbrc{z_{\sigma}^{-A_{\sigma}^{-1} A} \times z}_j =
    \begin{cases}
        1 & j \in \sigma \, .
        \\
        z_{\sigma}^{-A_{\sigma}^{-1} a_j} \, z_j & j \in \bsigma{} \, .
    \end{cases}
\end{align}
This encourages us to define the rescaled variables
\begin{align}
    w \defas z_\sigma^{-A_\sigma^{-1} A_{\bsigma}} \times z_{\bsigma} \, ,
    \label{eq:w-def}
\end{align}
as well as the rescaled Euler integral
\begin{align}
    \label{g_Gamma(w)}
    \eulerIntResc\brk{w} \defas
    z_{\sigma}^{A_{\sigma}^{-1} \beta} \, \eulerInt\brk{z} \, .
\end{align}
With this notation we are now ready to formulate the main proposition of this
section.
\begin{proposition}
    A generic partial derivative operator
    $\partial_z^v \defas \partial_{z_1}^{v_1} \partial_{z_2}^{v_2} \cdots \partial_{z_N}^{v_N}$
    acts on the Euler integral \eqref{eq:euler-int-rescaled} as follows:
    \begin{align}
        \label{P(q)_simplex}
        \partial_z^v \eulerInt\brk{z} =
        z_{\sigma}^{-A_{\sigma}^{-1} \brk{\beta + A v}} \,
        \brk{-1}^{\abs{v_\sigma}} \,
        \Bigbrk{
            w^{-v_{\bsigma}} \,
            \bigbrk{A_{\sigma}^{-1} \beta + A_{\sigma}^{-1} A_{\bsigma} \theta_w}_{v_{\sigma}} \,
            \bigsbrk{\theta_w}_{v_{\bsigma}} \, \eulerIntResc\brk{w}
        }\Big|_{w = \eqref{eq:w-def}} \, ,
    \end{align}
    where $\abs{v_\sigma} \defas \sum_{i \in \sigma} v_i$.
    \label{prp:simplex}
\end{proposition}
\begin{proof}
    \done{\seva{explain somewhere that $A_\eta \theta_w$ is a column vector}}
    Let us start with a few useful formulas. The splitting of the vector $v$
    into $v_\sigma$ and $v_{\bsigma}$ implies the corresponding factorization
    of the differential operator:
    \begin{align}
        \partial_z^{v} = \partial_{z_\sigma}^{v_\sigma} \,
        \partial_{z_{\bsigma}}^{v_{\bsigma}} \, .
    \end{align}
    This operator can also be expressed in terms of the
    lowering factorial of an Euler operator like so:
    \begin{align}
        \partial_z^v = z^{-v} \bigsbrk{\theta_z}_{v} \, .
    \end{align}
    Finally, the raising and lowering factorials%
    \footnote{
        Recall the definitions
        $\brk{a}_b \defas a \brk{a + 1} \ldots \brk{a + b - 1}$
        and
        $\sbrk{a}_b \defas a \brk{a - 1} \ldots \brk{a - b + 1}$.
        In vector case, these definitions are applied component-wise:
        $\brk{u}_v \defas \brk{u_1}_{v_1} \ldots \brk{u_N}_{u_N}$.
    } are related as follows:
    \begin{align}
        \bigbrk{a}_b = \brk{-1}^{\abs{b}} \, \bigsbrk{-a}_b \, .
    \end{align}

    Using these formulas, it follows that the action of 
    $\partial_{z_{\bsigma}}^{v_{\bsigma}}$ on the Euler integral \eqref{g_Gamma(w)} is
    \begin{align}
        \partial_{z_{\bsigma}}^{v_{\bsigma}} \eulerInt\brk{z}
        = z_{\sigma}^{-A_{\sigma}^{-1} \beta} \, \partial_{z_{\bsigma}}^{v_{\bsigma}} \, \eulerIntResc\bigbrk{z_{\sigma}^{-A_{\sigma}^{-1} A_{\bsigma}} \times z_{\bsigma}}
        = z_{\sigma}^{-A_{\sigma}^{-1} \brk{\beta + A_{\bsigma} v_{\bsigma}}}
        \, \Bigbrc{
            \partial_w^{v_{\bsigma}} \eulerIntResc\brk{w}
        }
        \ ,
    \end{align}
    where the curly brackets indicate the substitution $w =
    \eqref{eq:w-def}$, once the expression inside of them was evaluated.
    Proceeding with the action of the $\sigma$-part of $\partial^v_z$, we have 
    \begin{align}
        \partial_{z_{\sigma}}^{v_\sigma} \partial_{z_{\bsigma}}^{v_{\bsigma}} \eulerInt\brk{z}
        &= z_\sigma^{-v_\sigma} \bigsbrk{\theta_{z_\sigma}}_{v_\sigma} \partial_{z_{\bsigma}}^{v_{\bsigma}} \eulerInt\brk{z}
        \\
        &= z_\sigma^{-v_\sigma} z_{\sigma}^{-A_{\sigma}^{-1} \brk{\beta + A_{\bsigma} v_{\bsigma}}} \,
        \bigsbrk{-A_{\sigma}^{-1}\brk{\beta + A_{\bsigma} v_{\bsigma}} + \theta_{z_\sigma}}_{v_\sigma}
        \Bigbrc{
            \partial_w^{v_{\bsigma}} \eulerIntResc\brk{w}
        }
        \\
        &= z_{\sigma}^{-A_{\sigma}^{-1} \brk{\beta + A v}} \,
        \Bigbrc{
            \bigsbrk{-A_{\sigma}^{-1}\brk{\beta + A_{\bsigma} v_{\bsigma}} - A_\sigma^{-1} A_{\bsigma} \theta_w}_{v_\sigma}
            w^{-v_{\bsigma}}
            \bigsbrk{\theta_w}_{v_{\bsigma}} \eulerIntResc\brk{w}
        }
        \\
        &= z_{\sigma}^{-A_{\sigma}^{-1} \brk{\beta + A v}} \,
        \Bigbrc{
            w^{-v_{\bsigma}} \,
            \bigsbrk{-A_{\sigma}^{-1}\brk{\beta + A_{\bsigma} \theta_w}}_{v_\sigma}
            \bigsbrk{\theta_w}_{v_{\bsigma}} \eulerIntResc\brk{w}
        }
        \\
        &= z_{\sigma}^{-A_{\sigma}^{-1} \brk{\beta + A v}} \,
        \brk{-1}^{\abs{v_\sigma}}
        \Bigbrc{
            w^{-v_{\bsigma}} \,
            \bigbrk{A_{\sigma}^{-1}\brk{\beta + A_{\bsigma} \theta_w}}_{v_\sigma}
            \bigsbrk{\theta_w}_{v_{\bsigma}} \eulerIntResc\brk{w}
        }
        \ ,
    \end{align}
    where $A_\eta \theta_w$ is a column vector, obtained from matrix
    multiplication of $A_\eta$ and a column vector $\theta_w$.
\end{proof}
Let us close this discussion with one application of \prpref{prp:simplex},
namely the action of a generic partial derivative operator on the cohomology
class $\sbrk{\dd x / x}$:
\begin{align}
    \partial_z^v \bullet \lrsbrk{\frac{dx}{x}} \Big|_{z_i=1, \, i \in \sigma} =
    \brk{-1}^{\abs{v}}
    \Bigbrk{
        w^{-v_{\bsigma}} \,
        \bigbrk{
            A_\sigma^{-1} \brk{\beta + A_{\bsigma} \theta_w}
        }_{v_\sigma}
        \bigsbrk{\theta_w}_{v_{\bsigma}} \bullet
        \lrsbrk{\frac{dx}{x}}
    }\Big|_{w = z_{\bsigma}}.
    \label{F(q)_rest}
\end{align}

\section{Holonomic \texorpdfstring{$\D$}{}-modules in a nutshell}
\label{sec:D-modules-appendix}

In this appendix, we summarize basics of the theory of holonomic systems
utilized in our study of Feynman integrals. We cite introductory textbooks
rather than the original papers or comprehensive textbooks.

\subsection{Holonomic ideals}
Let $\D$ be the Weyl algebra with polynomial coefficients:
\begin{equation}
    \D := 
 \CC\langle z_1, \ldots, z_N, \pd{1}, \ldots, \pd{N} \rangle
= \CC[z_1, \ldots, z_N]\langle \pd{1}, \ldots, \pd{N} \rangle
\end{equation}
where the commutators are $[z_i,z_j]=0$, $[\pd{i},\pd{j}]=1$, and $[\pd{i},z_j] = \delta_{ij}$.
When we need to specify the number of variables, we also denote this algebra by
$\D_N$. Any element $L \in \D$ can be written in the normally ordered form
\begin{equation}
 L=\sum_{(p,q) \in E}  c_{p,q} z^p \pd{}^q \, ,
  \quad z^p=z_1^{p_1} \cdots z_N^{p_N},
        \pd{}^q = \pd{1}^{q_1} \cdots \pd{N}^{q_N}, c_{p,q} \in \CC
        \label{eq:poly-norm-ord}
\end{equation}
via the commutator relations.
The principal symbol $\initial_{({\bf 0},{\bf 1})}(L)$ of $L$ is defined as the
sum of the highest order differential operators in $L$:
\begin{equation}
 \initial_{({\bf 0},{\bf 1})}(L)
= \sum_{(p,q) \in E, \mbox{ $|q|$ is max}} c_{p,q} z^p \xi^q \in \CC[z,\xi]
      \label{eq:principal-symbol}
\end{equation}
where $|q| \defas q_1+\cdots+q_N$, and the commutative variable $\xi$ denotes the
derivative $\partial$ when a given expression is written in the normally ordered
form. We see that the principal symbol is the sum of the highest weight
terms with respect to the weight ${\bf 0}$ for $z$ and the weight
${\bf 1}=(1, \ldots, 1)$ for $\pd{}$.

Let $\Ideal$ be a left ideal in $\D$. The left $\D$-module $\D/\Ideal$ is called {\it holonomic}
when the dimension of the ideal of principal symbols
\begin{equation}
\initial_{({\bf 0},{\bf 1})}(\Ideal) = \CC \{ \initial_{({\bf 0},{\bf 1})}(L)\,|\,
      L \in \Ideal \}
\end{equation}
equals $N$ (see also \cite[Th 1.4.12, p.32]{SST}).
In this case, the left ideal $\Ideal$ is called a \emph{holonomic ideal}.
For general $\D$-modules, the definition of a holonomic $\D$-module requires the notion
of filtration, which lies out of the scope of this short review.
We refer the interested reader to \cite[\S 6.7, \S 6.8]{dojo} and references therein
for further details.

\begin{example}\rm
    Take $N = 2$ and let $\Ideal$ be a left ideal generated by $z_1$ and $\pd{2}$.
    The ideal of principal symbols is then generated by $z_1$ and $\xi_2$.
    The dimension of the zero set of $z_1=\xi_2=0$ in $\CC^4$ is $2$, and so
    we conclude that $\D/\Ideal$ is holonomic.
\end{example}

\subsection{Standard monomials}

Let $\WeylR$ be the rational Weyl algebra
\begin{equation}
\WeylR := \CC(z_1, \ldots, z_N) \otimes_{\CC[z_1,\ldots,z_N]} \D
  = \CC(z_1, \ldots, z_N) \langle \pd{1}, \ldots, \pd{N} \rangle \, .
\end{equation}
When $\Ideal$ is a holonomic ideal in $\D$,
the left ideal $\WeylR \, \Ideal$ of $\WeylR$ is a zero-dimensional ideal.
In other words, the dimension of $\WeylR/(\WeylR \, \Ideal)$ as a vector space
over the rational function field $\CC(z_1, \ldots, z_N)$ is finite.
Similar to the polynomial case \eqref{eq:poly-norm-ord}, any element $L$ of the
rational Weyl algebra $\WeylR$ can be brought into the normally ordered form
\begin{equation}
 L = \sum_{q \in E} c_q(z) \pd{}^q \, , \quad c_q(z) \in \CC(z_1, \ldots,z_N) \, .
\end{equation}

Let $\prec$ be a term order among the monomials in $\pd{i}$'s.
The largest monomial $\pd{}^q$ in $L$ \brk{stripped of its rational function coefficient}
is called the initial monomial $\initial_\prec(L)$ of $L$. 
Monomials in derivatives naturally form the polynomial ring $\CC[\pd{1}, \ldots, \pd{N}]$.
A finite subset $G$ of $\WeylR$ is called a Gr\"obner basis of $\WeylR \, \Ideal$
with respect to $\prec$ when the initial monomial ideal
\begin{equation}
\initial_\prec(\WeylR \, \Ideal) = \CC \{ \initial_\prec(L)\,|\, L \in \WeylR \, \Ideal\}
 \subset \CC[\pd{1}, \ldots, \pd{N}]
\end{equation}
is generated by $\initial_\prec(L)$, $L \in G$.
A monomial $\pd{}^\alpha$ is called a {\it standard monomial} with respect to $G$
when $\pd{}^\alpha$ does not belong to the initial monomial ideal
\brk{$\initial_\prec\brk{\WeylR \, \Ideal}$ in our case}.
Within the theory of Gr\"obner basis in $\WeylR$,
the zero-dimensionality of $\WeylR \, \Ideal$ is equivalent
to the finitness of the set of standard monomials.
In the following, an element $\pd{}^\alpha$ of $\initial_\prec(\WeylR \, \Ideal) \subset \CC[\pd{1}, \ldots, \pd{N}]$
is also denoted by $\xi^\alpha$ to emphasize that it is an element of the polynomial ring.
Let us now give a few examples of standard monomials for different ideals.

\begin{example} \rm
    When $N=2$ and $\Ideal=\langle z_1, \pd{2} \rangle \subset \D$,
    we have $\WeylR \, \Ideal = \langle 1 \rangle$ because $\frac{1}{z_1} z_1 = 1$.
    Then the set of standard monomials is empty.
\end{example}

\noindent
\begin{minipage}{11cm}
\begin{example}\rm
    Put
    \eq{ 
        L_1 &= ({z}_{1}  {z}_{2}- 1) \, \pd{1}+  {\beta}_{0}  {z}_{2} 
        \\
        L_2 &= z_2(   {z}_{1}   {z}_{2} - 1) \, \pd{2}+    (  - {\beta}_{1}+ {\beta}_{0})  {z}_{1}  {z}_{2}+ {\beta}_{1} \, ,
    }
    where $\beta_0$ and $\beta_1$ are complex numbers. These two operators annhilate the function
    $(1-z_1z_2)^{-\beta_0} z_2^{\beta_1}$.
    By a Gr\"obner basis computation in the Weyl algebra $\D$, we find that
    the ideal of principal symbols is generated by the two elements
    $   \xi_{1} -   {z}_{2}^{ 2}   \xi_{2}$ and
    $   {z}_{1}  \xi_{1} -  {z}_{2}  \xi_{2}$.
    Since it defines a $2$-dimensional variety in $\CC^4$, we conclude that the
    module $\D/\langle L_1, L_2 \rangle$ is holonomic.
    The set $\{L_1,L_2\}$ is a Gr\"obner basis in $\WeylR$ with any term order $\prec$
    and the initial ideal is generated by $\xi_1$ and $\xi_2$.
    The set of standard monomials is simply $\brc{\xi_1^0 \xi_2^0} \equiv \brc{1}$, corresponding to the point $(0,0)$ in the figure to the right, where the axes represent powers of $\xi_1,\xi_2$.\\
    \end{example}
    \end{minipage}
\begin{minipage}{4cm}    
        \setlength{\unitlength}{2mm}
        \begin{picture}(22,22)
            \put(0,10){\vector(1,0){22}}
            \put(10,0){\vector(0,1){22}}
            \put(9,8){(0,0)}
            \put(10,10){\circle*{1}}
            \put(14,8){(1,0)}
            \put(15,10){\circle{1}}
            \put(15,10){\line(0,1){11}}
            \put(8,13){(0,1)}
            \put(10,15){\circle{1}}
            \put(10,15){\line(1,0){11}}

            \put(10,15){\circle{1}}\put(10,20){\circle{1}}
            \put(15,15){\circle{1}}\put(15,20){\circle{1}}
            \put(20,10){\circle{1}}\put(20,15){\circle{1}}\put(20,20){\circle{1}}
        \end{picture}
\end{minipage}

\begin{theorem} \emph{(see e.g. \cite[\S 6.9]{dojo})}
    Let $\Ideal$ be a holonomic ideal in $\D$. Then $\WeylR \, \D$ is a zero-dimensional ideal
    in $\WeylR$.
    Conversely, let $\Jideal$ be a zero-dimensional ideal in $\WeylR$.
    Then $\D \cap \Jideal$ is a holonomic ideal.
\end{theorem}
Note that the second construction $\D \cap \Jideal$ from $\Jideal$ is a highly non-trivial.
Algorithms realising such a construction are called Weyl closure algorithms.

We close this subsection with a note on the relation between Pfaffian equations
and Gr\"obner bases.
Let $G$ be a reduced Gr\"obner basis of a zero dimensional ideal 
$\WeylR \Ideal$
in the rational Weyl algebra $\WeylR$.
Since it is zero-dimensional, the set of the standard monomials
is a finite set.
We denote it by $\{ s_1, \ldots, s_r \}$.
Computing the normal form of $\pd{i} s_j$ by the Gr\"obner basis, 
we obtain
$$
  \pd{i} s_j = \sum_{k=1}^r P_{ij}^k s_k \quad {\rm mod}\, \WeylR G
$$
The matrix $P_i=(P_{ij}^k)_{j,k}$ is the Pfaffian matrix for the operator $\pd{i}$.

Conversely, assume that we are given a basis 
$S=\{s_1, \ldots, s_r \}$
of $\WeylR/\WeylR \Ideal$
as a vector space over $\CC(z_1, \ldots, z_N)$.
Here the $s_i$'s are monomials.
Let $P_i$ be the Pfaffian matrix for $\pd{i}$.
Then we have $\pd{i} s_j = \sum_{k=1}^r (P_i)_{jk} s_k$
modulo $\WeylR \Ideal$.
We assume that there exists a term order $\prec$ such that
$s \prec t$ for all monomials $s \in S$ and $t \in S^c$
where $S^c$ is the set of the monomials which do not belong to $S$.
If such an order exists, it can be found by solving a system of linear
inequalities \cite[\S 2.1]{SST}. See also a note below.
Then the set
$G=\{ g_{ij}:=\pd{i} s_j - \sum_{k=1}^r (P_i)_{jk} s_k \,|\, \pd{i}s_j \not\in S\}$
is a Gr\"obner basis with repect to the term order $\prec$,
and
the set $S$ is the set of the standard monomials for the Gr\"obner basis.
Let us prove it.
Since $u \succeq 1$ for any monomial $u$,
we have $ut \succeq t$ for any monomial $t$.
Therefore, we have $ut \in S^c$ for $t \in S^c$ and any monomial $u$ 
by the condition of the order.
It implies that a Gr\"obner basis of the monomial ideal $\langle S^c \rangle$
is $\{ \pd{i} s_j \,|\, \pd{i} s_j \not\in S\}$,
and $S$ is the set of standard monomials for that basis.
It follows from the condition on the order that we have
$\initial_\prec(g_{ij}) = \pd{i} s_j \not\in S$ and 
$\initial_\prec(G) \subset \initial_\prec(\WeylR \Ideal)$.
Since $\mathop{\rm dim} \CC[\pd{1}, \ldots, \pd{N}]/\initial_\prec(G) = r$
by $\initial_\prec(G) = \langle S^c \rangle$,
$\initial_\prec(G) = \initial_\prec(\WeylR \Ideal)$ holds.
Thus, $G$ is a Gr\"obner basis for the order $\prec$.

Note that it is necessary to pose the condition on the order
and the basis $S$.
For example, consider $N=1$ and $M=\WeylR/\WeylR(\pd{1}-z_1)$.
$S=\{ \pd{1} \}$ is a vector space basis of $M$ over $\CC(z_1)$
and the Pfaffian system with respect to $S$ is
$\pd{1}-(z_1+1/z_1)$, because
$\pd{1}(\pd{1}-z_1)=\pd{1}^2-z_1 \pd{1} -1$ and $\frac{1}{z_1} \pd{1}-1$
belong to the ideal $\WeylR(\pd{1}-z_1)$.
However, the reduced G\"obner basis of the ideal for any term order is
$\{ \pd{1}-z_1 \}$.
In fact, the condition $\pd{1} \prec 1$ does not hold for any term order.

Finally, we explain a procedure to check 
if there exists a term order $\prec$ such that
$s \prec t$ for all monomials $s \in S$ and $t \in S^c$.
Let $w \in \RR_{\geq 0}^N$ be a weight vector.
We solve a system of linear inequlities with respect to $w$
$$
  w \cdot \alpha \leq w \cdot \beta \ 
  \mbox{ for all }
  s=\pd{}^\alpha \in S \mbox{ and } t=\pd{}^\beta \in S'
$$
where $S'$ is the set of the minimal generators of the monoid $S^c$.
If the solution space is an $N$-dimensional cone, then take $w$ from the interior
of the cone and define the order by $\prec_w$, where we may take any tie breaker
$\prec$. 
Here, we mean $\pd{}^a \prec_w \pd{}^b$ if and only if $w \cdot a < w \cdot b$
or ($w \cdot a = w \cdot b$ and $ \pd{}^a \prec \pd{}^b$).
If the cone is empty, there does not exist such an order.
When the cone has dimension less than $N$, take $w(1) \in \RR_{\geq 0}^N$ 
from the relative interior of the cone
and also take the set $P$ of $(\alpha, \beta)$'s such that
$ w(1) \cdot \alpha = w(1) \cdot \beta$.
We solve a system of linear inequalities 
$$
  w \cdot \alpha \leq w \cdot \beta \ \mbox{ for  all $(\alpha,\beta) \in P$}.
$$
Take $w(2) \in \RR_{\geq 0}^N$ from the relative interior of the solution cone.
If this $w(2)$ satisfies $w(2) \cdot \alpha = w(2) \cdot \beta$
for all $(\alpha,\beta) \in P$, there exists no term order we want.
If not, take $\prec_{w(2)}$ as a tie breaker. 
We repeat this procedure.
Since any term order can be expressed by a matrix of weight vectors \cite{Robbiano:1985},
we can check if there exists a term order 
$s \prec t$ for all monomials $s \in S$ and $t \in S^c$
and we can construct it if it exists.

\subsection{Proof of Theorem \ref{thm:3.1}}
In this subsection, we prove \thmref{thm:3.1}.
We follow the notation of \appref{ssec:D-module to DEQ}.
Let $J$ be an ideal of $\CC[\theta]$.
It naturally lifts to a left ideal of $\mathcal{R}$ generated by the elements of $J$ which we will denote by $\mathcal{R}J$.
Recall that a term order $\prec$ on the ring $\CC[\theta]$ naturally corresponds to an order among monomials $\partial^k$ by the correspondence $\theta^k\leftrightarrow\partial^k$.
Thus, $\prec$ induces a term order on the ring $\mathcal{R}$.

\begin{lemma}\label{lem:Frob}
The set of $\prec$-standard monomials of the ideal $J$ coincides with that of $\mathcal{R}J$ through the correspondence $\theta^k\leftrightarrow\partial^k$.
\end{lemma}
\begin{proof}
Let us recall an identity
\begin{equation}\label{eqn:theta-to-pd}
    z_i^k\partial_i^k=\theta_i(\theta_i-1)\cdots(\theta_i-k+1)
\end{equation}
for any $k\in\mathbb{N}$.
For any $k=(k_1,\dots,k_N)\in\mathbb{N}^N$, we set
\begin{equation}
    [\theta]_k:=\prod_{i=1}^N\theta_i(\theta_i-1)\cdots(\theta_i-k_i+1).
\end{equation}
We claim that a $\prec$-Gr\"obner basis $G\subset J$ of $J$ is again a $\prec$-Gr\"obner basis of $\mathcal{R}J$.
Let us take $P,Q\in G$ and write them as
\begin{equation}
    P=\underline{a\theta^k}+\cdots,\ \ \ Q=\underline{b\theta^\ell}+\cdots
\end{equation}
where $a,b\in\CC$ and underlines indicate the leading terms.
Regarding $P$ and $Q$ as elements of the ring $\mathcal{R}$, we obtain an expansion
\begin{equation}
    P=\underline{az^k\pd{}^k}+\cdots,\ \ \ Q=\underline{bz^\ell\pd{}^\ell}+\cdots.
\end{equation}
We set $c_i={\rm max}(k_i,\ell_i)$ and $c=(c_1,\dots,c_N)$.
The S-pair of $P$ and $Q$ in the ring $\mathcal{R}$ (cf. \cite[\S 6.1]{dojo}) is
\begin{equation}
 bz^{\ell}\partial^{c-k}P-az^k\pd{}^{c-\ell}Q=z^{k+\ell-c}(b[\theta]_{c-k}P-a[\theta]_{c-\ell}Q).   
\end{equation}
The second factor clearly belongs to the ideal $J$ and the S-pair of $P$ and $Q$ is reduced to zero by $G$ in the ring $\mathcal{R}$.
This shows that $G$ is a $\prec$-Gr\"obner basis of $\mathcal{R}J$.
Thus, the set of leading monomials of $J$ is identical to that of $\mathcal{R}J$ through the correspondence $\theta^k\leftrightarrow\partial^k$.
This proves the lemma.
\end{proof}

\noindent
The following theorem is known as {\it Gr\"obner deformation}.
\begin{theorem}[a simplified version of Theorem 3.1.3 of \cite{SST}]\label{thm:GD}
For any parameter $\beta$ generic enough, one has an identity
\begin{equation}
{\rm in}_\prec(\mathcal{R} H_A(\beta))=\mathcal{R}\langle E_1,\dots,E_{n+1},{\rm in}_{\prec}(\mathcal{I}_{\square})\rangle.
\end{equation}
\end{theorem}

Now, let us prove \thmref{thm:3.1}.
Since $z^k$ is invertible in the ring $\mathcal{R}$, \thmref{thm:GD} shows that ${\rm in}_\prec(\mathcal{R} H_A(\beta))$ is identical to the ideal $\mathcal{R}\Ideal'$.
Clearly, the right-hand side is a lift of an ideal $\Ideal'$ in $\CC[\theta]$.
By \namedref{Lemma}{lem:Frob}, we can conclude that the set of standard monomials of $\Ideal'$ is identical to that of ${\rm in}_\prec(\mathcal{R} H_A(\beta))$ through the correspondence $\theta^k\leftrightarrow\partial^k$.

\subsection{Integration and restriction}
Now we discuss two important properties of the holonomic $\D$-modules:
the so-called integration and restriction constructions.

Let $\D_{M+N}$ be the ring of differential operators in $M + N$ variables:
$$
\CC \langle x_1, \ldots, x_M, z_1, \ldots, z_N ,
    \pd{x_1}, \ldots, \pd{x_M}, \pd{z_1}, \ldots, \pd{z_N} \rangle,
$$
and $\Ideal$ be a left holonomic ideal of $\D_{M+N}$.
We denote
$
\CC \langle x_1, \ldots, x_M,
    \pd{x_1}, \ldots, \pd{x_M} \rangle
$
by $\D_M$ and \\
$
\CC \langle z_1, \ldots, z_N ,
          \pd{z_1}, \ldots, \pd{z_N} \rangle
$
by $\D_N$.

The \emph{integration} of the $\D_{M+N}$-module $L=\D_{M+N}/\Ideal$ for the $x$
variables is defined by
\begin{equation} \label{eq:d-integral}
  \frac{\D_{M+N}}{\sum_{i=1}^M \pd{x_i}\D_{M+N}} \otimes_{\D_{M+N} } L
 = \frac{\D_{M+N}}{\Ideal + \sum_{i=1}^M \pd{x_i}\D_{M+N}},
\end{equation}
which is a left $\D_N$-module \brk{see \cite[\S 5.5]{SST} or \cite[\S 6.10]{dojo}
for related algorithms and examples}.
The integration of a $\D$-module is an algebraic counterpart
of the integration in calculus (for example, see \cite[Th 5.5.1, p.227]{SST}).
It also gives a different approach to the twisted cohomology theory 
(see, e.g., \cite[Th 5.5.11, p.233]{SST}).
The integration module of $L$ is sometimes denoted with an integral sign:
\begin{equation}  \label{eq:d-integral-symbol}
\int_\pi L,
\end{equation}
where $\pi$ denotes the natural projection from the $\brk{z,x}$-space to the $z$-space.

The \emph{restriction} of the $\D$-module $L$ to a linear subspace
$z_1=c_1, \ldots, z_N=c_N$, for some constant $c_i \in \CC$, is defined by
\begin{equation} \label{eq:d-restriction}
  \frac{\D_{M+N}}{\sum_{j=1}^n (z_j-c_j) \D_{M+N}} \otimes_{\D_{M+N} } L
 = \frac{\D_{M+N}}{\Ideal + \sum_{j=1}^n (z_j-c_j) \D_{M+N}}.
\end{equation}
It is an algebraic counterpart of the restriction of solutions of
differential equations
to the linear subspace \brk{see \cite[Prop. 5.2.14, p.207]{SST}}.

\begin{theorem} \emph{(see e.g. \cite[Th 6.10.8, p.303]{dojo}.)}
    If $L$ is a holonomic $\D_{M+N}$-module, then the integration and the
    restriction defined above are holonomic $\D_N$-modules.
\end{theorem}

Finally, let us note that the integration and the restriction constructions
commute:
\begin{eqnarray}
& &  \frac{\D_{n}}{\sum_{j=1}^n (z_j-c_j) \D_{n}} \otimes_{\D_N}
   \left(\frac{\D_{M+N}}{\sum_{i=1}^M \pd{x_i}\D_{M+N}} \otimes_{\D_{M+N} } L \right) \label{eq:integration_first} \\
&=&  \frac{\D_{M}}{\sum_{i=1}^M \pd{x_i}\D_{M}} \otimes_{\D_M} \left(\frac{\D_{M+N}}{\sum_{j=1}^n (z_j-c_j) \D_{M+N}}\otimes_{\D_{M+N} } L \right) \label{eq:restriction_first}\\
&=& \frac{\D_{M+N}}{
  \Ideal+\sum_{j=1}^n (z_j-c_j) \D_{M+N}+\sum_{i=1}^M \pd{x_i}\D_{M+N}}.
\end{eqnarray}
Restrictions of generalized Feynman integrals will be discussed in the
forthcomming paper \cite{Chestnov:2023kww}.

\bibliographystyle{nb}
\bibliography{biblio}

\begin{thebibliography}{100}
\providecommand{\href}[2]{#2}
\providecommand{\arxivref}[2]{\href{http://arxiv.org/abs/#1}{#2}}
\providecommand{\doiref}[2]{\href{http://dx.doi.org/#1}{#2}}
\providecommand{\nbbstauthor}[1]{#1}
\providecommand{\nbbstjournal}[1]{\textsf{#1}}
\providecommand{\nbbsttitle}[1]{\textit{``#1''}}
\providecommand{\nbbsturl}[1]{\texttt{#1}}
\providecommand{\nbbsteprint}[1]{\texttt{#1}}
\providecommand{\nbbststyle}{\raggedright\small\parskip0pt}
\nbbststyle

\bibitem{Weinzierl:2022eaz}
\nbbstauthor{S.~Weinzierl},
\nbbsttitle{{Feynman Integrals}},
\nbbsteprint{\arxivref{2201.03593}{arxiv:2201.03593}}.

\bibitem{Chetyrkin:1981qh}
\nbbstauthor{K.~G.~Chetyrkin and F.~V.~Tkachov},
\nbbsttitle{{Integration by Parts: The Algorithm to Calculate beta Functions in
  4 Loops}},
\nbbstjournal{\doiref{10.1016/0550-3213(81)90199-1}{Nucl.~Phys.~B192,~159~(1981)}}.

\bibitem{Laporta:2001dd}
\nbbstauthor{S.~Laporta},
\nbbsttitle{{High precision calculation of multiloop Feynman integrals by
  difference equations}},
\nbbstjournal{\doiref{10.1016/S0217-751X(00)00215-7,
  10.1142/S0217751X00002157}{Int.~J.~Mod.~Phys.~A15,~5087~(2000)}},
\nbbsteprint{\arxivref{hep-ph/0102033}{hep-ph/0102033}}.

\bibitem{Barucchi:1973zm}
\nbbstauthor{G.~Barucchi and G.~Ponzano},
\nbbsttitle{{Differential equations for one-loop generalized feynman
  integrals}},
\nbbstjournal{\doiref{10.1063/1.1666327}{J.~Math.~Phys.~14,~396~(1973)}}.

\bibitem{KOTIKOV1991158}
\nbbstauthor{A.~Kotikov},
\nbbsttitle{{Differential equations method. New technique for massive Feynman
  diagram calculation}},
\nbbstjournal{\doiref{https://doi.org/10.1016/0370-2693(91)90413-K}{Physics~Letters~B~254,~158
  ~(1991)}},
\href{http://www.sciencedirect.com/science/article/pii/037026939190413K}{\nbbsturl{http://www.sciencedirect.com/science/article/pii/037026939190413K}}.

\bibitem{KOTIKOV1991123}
\nbbstauthor{A.~Kotikov},
\nbbsttitle{Differential equation method. The calculation of N-point Feynman
  diagrams},
\nbbstjournal{\doiref{https://doi.org/10.1016/0370-2693(91)90536-Y}{Physics~Letters~B~267,~123
  ~(1991)}},
\href{http://www.sciencedirect.com/science/article/pii/037026939190536Y}{\nbbsturl{http://www.sciencedirect.com/science/article/pii/037026939190536Y}}.

\bibitem{Bern:1993kr}
\nbbstauthor{Z.~Bern, L.~J.~Dixon and D.~A.~Kosower},
\nbbsttitle{{Dimensionally regulated pentagon integrals}},
\nbbstjournal{\doiref{10.1016/0550-3213(94)90398-0}{Nucl.~Phys.~B412,~751~(1994)}},
\nbbsteprint{\arxivref{hep-ph/9306240}{hep-ph/9306240}}.

\bibitem{Remiddi:1997ny}
\nbbstauthor{E.~Remiddi},
\nbbsttitle{{Differential equations for Feynman graph amplitudes}},
\nbbstjournal{Nuovo~Cim.~A110,~1435~(1997)},
\nbbsteprint{\arxivref{hep-th/9711188}{hep-th/9711188}}.

\bibitem{Gehrmann:1999as}
\nbbstauthor{T.~Gehrmann and E.~Remiddi},
\nbbsttitle{{Differential equations for two loop four point functions}},
\nbbstjournal{\doiref{10.1016/S0550-3213(00)00223-6}{Nucl.~Phys.~B580,~485~(2000)}},
\nbbsteprint{\arxivref{hep-ph/9912329}{hep-ph/9912329}}.

\bibitem{Henn:2013pwa}
\nbbstauthor{J.~M.~Henn},
\nbbsttitle{{Multiloop integrals in dimensional regularization made simple}},
\nbbstjournal{\doiref{10.1103/PhysRevLett.110.251601}{Phys.~Rev.~Lett.~110,~251601~(2013)}},
\nbbsteprint{\arxivref{1304.1806}{arxiv:1304.1806}}.

\bibitem{Henn:2014qga}
\nbbstauthor{J.~M.~Henn},
\nbbsttitle{{Lectures on differential equations for Feynman integrals}},
\nbbstjournal{\doiref{10.1088/1751-8113/48/15/153001}{J.~Phys.~A48,~153001~(2015)}},
\nbbsteprint{\arxivref{1412.2296}{arxiv:1412.2296}}.

\bibitem{Papadopoulos:2014lla}
\nbbstauthor{C.~G.~Papadopoulos},
\nbbsttitle{{Simplified differential equations approach for Master Integrals}},
\nbbstjournal{\doiref{10.1007/JHEP07(2014)088}{JHEP~1407,~088~(2014)}},
\nbbsteprint{\arxivref{1401.6057}{arxiv:1401.6057}}.

\bibitem{Mastrolia:2014wca}
\nbbstauthor{P.~Mastrolia, M.~Argeri, S.~Di~Vita, E.~Mirabella, J.~Schlenk,
  U.~Schubert and L.~Tancredi},
\nbbsttitle{{Magnus and Dyson Series for Master Integrals}},
\nbbstjournal{\doiref{10.22323/1.211.0007}{PoS~LL2014,~007~(2014)}}.

\bibitem{Tarasov:1996br}
\nbbstauthor{O.~V.~Tarasov},
\nbbsttitle{{Connection between Feynman integrals having different values of
  the space-time dimension}},
\nbbstjournal{\doiref{10.1103/PhysRevD.54.6479}{Phys.~Rev.~D54,~6479~(1996)}},
\nbbsteprint{\arxivref{hep-th/9606018}{hep-th/9606018}}.

\bibitem{Lee:2009dh}
\nbbstauthor{R.~N.~Lee},
\nbbsttitle{{Space-time dimensionality D as complex variable: Calculating loop
  integrals using dimensional recurrence relation and analytical properties
  with respect to D}},
\nbbstjournal{\doiref{10.1016/j.nuclphysb.2009.12.025}{Nucl.~Phys.~B830,~474~(2010)}},
\nbbsteprint{\arxivref{0911.0252}{arxiv:0911.0252}}.

\bibitem{Laporta:2003jz}
\nbbstauthor{S.~Laporta},
\nbbsttitle{{Calculation of Feynman integrals by difference equations}},
\nbbstjournal{Acta~Phys.~Polon.~B34,~5323~(2003)},
\nbbsteprint{\arxivref{hep-ph/0311065}{hep-ph/0311065}},
in: \nbbsttitle{{Matter to the deepest: Recent developments in physics of
  fundamental interactions. Proceedings, 27th International Conference of
  theoretical physics, Ustron, Poland, September 15-21, 2003}},
pp.~5323-5334.

\bibitem{Larsen:2015ped}
\nbbstauthor{K.~J.~Larsen and Y.~Zhang},
\nbbsttitle{{Integration-by-parts reductions from unitarity cuts and algebraic
  geometry}},
\nbbstjournal{\doiref{10.1103/PhysRevD.93.041701}{Phys.~Rev.~D93,~041701~(2016)}},
\nbbsteprint{\arxivref{1511.01071}{arxiv:1511.01071}}.

\bibitem{vonManteuffel:2014ixa}
\nbbstauthor{A.~von~Manteuffel and R.~M.~Schabinger},
\nbbsttitle{{A novel approach to integration by parts reduction}},
\nbbstjournal{\doiref{10.1016/j.physletb.2015.03.029}{Phys.~Lett.~B744,~101~(2015)}},
\nbbsteprint{\arxivref{1406.4513}{arxiv:1406.4513}}.

\bibitem{Peraro:2016wsq}
\nbbstauthor{T.~Peraro},
\nbbsttitle{{Scattering amplitudes over finite fields and multivariate
  functional reconstruction}},
\nbbstjournal{\doiref{10.1007/JHEP12(2016)030}{JHEP~1612,~030~(2016)}},
\nbbsteprint{\arxivref{1608.01902}{arxiv:1608.01902}}.

\bibitem{Klappert:2019emp}
\nbbstauthor{J.~Klappert and F.~Lange},
\nbbsttitle{{Reconstructing Rational Functions with $\texttt{FireFly}$}},
\nbbsteprint{\arxivref{1904.00009}{arxiv:1904.00009}}.

\bibitem{Peraro:2019svx}
\nbbstauthor{T.~Peraro},
\nbbsttitle{{FiniteFlow: multivariate functional reconstruction using finite
  fields and dataflow graphs}},
\nbbstjournal{\doiref{10.1007/JHEP07(2019)031}{JHEP~1907,~031~(2019)}},
\nbbsteprint{\arxivref{1905.08019}{arxiv:1905.08019}},
\href{https://github.com/peraro/finiteflow}{\nbbsturl{https://github.com/peraro/finiteflow}}.

\bibitem{cho1995}
\nbbstauthor{K.~Cho and K.~Matsumoto},
\nbbsttitle{{Intersection theory for twisted cohomologies and twisted Riemann's
  period relations I}},
\nbbstjournal{\doiref{10.1017/S0027763000005304}{Nagoya~Math.~J.~139,~67~(1995)}},
\href{http://projecteuclid.org/euclid.nmj/1118775097}{\nbbsturl{http://projecteuclid.org/euclid.nmj/1118775097}}.

\bibitem{matsumoto1994}
\nbbstauthor{K.~Matsumoto},
\nbbsttitle{{Quadratic Identities for Hypergeometric Series of Type $(k,l)$}},
\nbbstjournal{\doiref{10.2206/kyushujm.48.335}{Kyushu~Journal~of~Mathematics~48,~335~(1994)}}.

\bibitem{matsumoto1998}
\nbbstauthor{K.~Matsumoto},
\nbbsttitle{Intersection numbers for logarithmic $k$-forms},
\nbbstjournal{Osaka~J.~Math.~35,~873~(1998)},
\href{https://projecteuclid.org:443/euclid.ojm/1200788347}{\nbbsturl{https://projecteuclid.org:443/euclid.ojm/1200788347}}.

\bibitem{OST2003}
\nbbstauthor{K.~Ohara, Y.~Sugiki and N.~Takayama},
\nbbsttitle{{Quadratic Relations for Generalized Hypergeometric Functions $_p
  F_{p-1}$}},
\nbbstjournal{\doiref{10.1619/fesi.46.213}{Funkcialaj~Ekvacioj~46,~213~(2003)}}.

\bibitem{doi:10.1142/S0129167X13500948}
\nbbstauthor{Y.~Goto},
\nbbsttitle{{Twisted Cycles and Twisted Period Relations for Lauricella's
  Hypergeometric Function $F_C$}},
\nbbstjournal{\doiref{10.1142/S0129167X13500948}{International~Journal~of~Mathematics~24,~1350094~(2013)}},
\nbbsteprint{\arxivref{1308.5535}{arxiv:1308.5535}}.

\bibitem{aomoto2011theory}
\nbbstauthor{K.~Aomoto and M.~Kita},
\nbbsttitle{{Theory of Hypergeometric Functions}},
Springer Monographs in Mathematics,
Springer Japan (2011).

\bibitem{yoshida2013hypergeometric}
\nbbstauthor{M.~Yoshida},
\nbbsttitle{{Hypergeometric Functions, My Love: Modular Interpretations of
  Configuration Spaces}},
Aspects of Mathematics,
Vieweg+Teubner Verlag (2013).

\bibitem{goto2015}
\nbbstauthor{Y.~Goto and K.~Matsumoto},
\nbbsttitle{{The monodromy representation and twisted period relations for
  Appell’s hypergeometric function $F_4$}},
\nbbstjournal{\doiref{10.1215/00277630-2873714}{Nagoya~Math.~J.~217,~61~(2015)}},
\href{https://doi.org/10.1215/00277630-2873714}{\nbbsturl{https://doi.org/10.1215/00277630-2873714}}.

\bibitem{goto2015b}
\nbbstauthor{Y.~Goto},
\nbbsttitle{{Twisted period relations for Lauricella's hypergeometric functions
  $F_{A}$}},
\nbbstjournal{Osaka~J.~Math.~52,~861~(2015)},
\href{https://projecteuclid.org:443/euclid.ojm/1437137622}{\nbbsturl{https://projecteuclid.org:443/euclid.ojm/1437137622}}.

\bibitem{Yoshiaki-GOTO2015203}
\nbbstauthor{Y.~Goto},
\nbbsttitle{{Intersection Numbers and Twisted Period Relations for the
  Generalized Hypergeometric Function $_{m+1}F_{m}$}},
\nbbstjournal{\doiref{10.2206/kyushujm.69.203}{Kyushu~Journal~of~Mathematics~69,~203~(2015)}}.

\bibitem{Mizera:2017rqa}
\nbbstauthor{S.~Mizera},
\nbbsttitle{{Scattering Amplitudes from Intersection Theory}},
\nbbstjournal{\doiref{10.1103/PhysRevLett.120.141602}{Phys.~Rev.~Lett.~120,~141602~(2018)}},
\nbbsteprint{\arxivref{1711.00469}{arxiv:1711.00469}}.

\bibitem{matsubaraheo2019algorithm}
\nbbstauthor{S.-J.~Matsubara-Heo and N.~Takayama},
\nbbsttitle{An algorithm of computing cohomology intersection number of
  hypergeometric integrals},
\nbbstjournal{\doiref{10.1017/nmj.2021.2}{Nagoya~Mathematical~Journal~,~1~(2019)}},
\nbbsteprint{\arxivref{1904.01253}{arxiv:1904.01253}}.

\bibitem{Mastrolia:2018uzb}
\nbbstauthor{P.~Mastrolia and S.~Mizera},
\nbbsttitle{{Feynman Integrals and Intersection Theory}},
\nbbstjournal{\doiref{10.1007/JHEP02(2019)139}{JHEP~1902,~139~(2019)}},
\nbbsteprint{\arxivref{1810.03818}{arxiv:1810.03818}}.

\bibitem{Frellesvig:2019kgj}
\nbbstauthor{H.~Frellesvig, F.~Gasparotto, S.~Laporta, M.~K.~Mandal,
  P.~Mastrolia, L.~Mattiazzi and S.~Mizera},
\nbbsttitle{{Decomposition of Feynman Integrals on the Maximal Cut by
  Intersection Numbers}},
\nbbstjournal{\doiref{10.1007/JHEP05(2019)153}{JHEP~1905,~153~(2019)}},
\nbbsteprint{\arxivref{1901.11510}{arxiv:1901.11510}}.

\bibitem{Mizera:2019gea}
\nbbstauthor{S.~Mizera},
\nbbsttitle{{Aspects of Scattering Amplitudes and Moduli Space Localization}},
\nbbsteprint{\arxivref{1906.02099}{arxiv:1906.02099}}.

\bibitem{Frellesvig:2019uqt}
\nbbstauthor{H.~Frellesvig, F.~Gasparotto, M.~K.~Mandal, P.~Mastrolia,
  L.~Mattiazzi and S.~Mizera},
\nbbsttitle{{Vector Space of Feynman Integrals and Multivariate Intersection
  Numbers}},
\nbbstjournal{\doiref{10.1103/PhysRevLett.123.201602}{Phys.~Rev.~Lett.~123,~201602~(2019)}},
\nbbsteprint{\arxivref{1907.02000}{arxiv:1907.02000}}.

\bibitem{Frellesvig:2020qot}
\nbbstauthor{H.~Frellesvig, F.~Gasparotto, S.~Laporta, M.~K.~Mandal,
  P.~Mastrolia, L.~Mattiazzi and S.~Mizera},
\nbbsttitle{{Decomposition of Feynman Integrals by Multivariate Intersection
  Numbers}},
\nbbsteprint{\arxivref{2008.04823}{arxiv:2008.04823}}.

\bibitem{Mizera:2020wdt}
\nbbstauthor{S.~Mizera},
\nbbsttitle{{Status of Intersection Theory and Feynman Integrals}},
\nbbstjournal{\doiref{10.22323/1.383.0016}{PoS~MA2019,~016~(2019)}},
\nbbsteprint{\arxivref{2002.10476}{arxiv:2002.10476}}.

\bibitem{Frellesvig:2021vem}
\nbbstauthor{H.~A.~Frellesvig and L.~Mattiazzi},
\nbbsttitle{{On the Application of Intersection Theory to Feynman Integrals:
  the univariate case}},
\nbbstjournal{\doiref{10.22323/1.383.0017}{PoS~MA2019,~017~(2022)}},
\nbbsteprint{\arxivref{2102.01576}{arxiv:2102.01576}}.

\bibitem{Mandal:2022vok}
\nbbstauthor{M.~K.~Mandal and F.~Gasparotto},
\nbbsttitle{{On the Application of Intersection Theory to Feynman Integrals:
  the multivariate case}},
\nbbstjournal{\doiref{10.22323/1.383.0019}{PoS~MA2019,~019~(2022)}}.

\bibitem{Mastrolia:2022tww}
\nbbstauthor{P.~Mastrolia},
\nbbsttitle{{From Diagrammar to Diagrammalgebra}},
\nbbstjournal{\doiref{10.22323/1.383.0015}{PoS~MA2019,~015~(2022)}}.

\bibitem{Mizera:2019vvs}
\nbbstauthor{S.~Mizera and A.~Pokraka},
\nbbsttitle{{From Infinity to Four Dimensions: Higher Residue Pairings and
  Feynman Integrals}},
\nbbstjournal{\doiref{10.1007/JHEP02(2020)159}{JHEP~2002,~159~(2020)}},
\nbbsteprint{\arxivref{1910.11852}{arxiv:1910.11852}}.

\bibitem{Weinzierl:2020xyy}
\nbbstauthor{S.~Weinzierl},
\nbbsttitle{{On the computation of intersection numbers for twisted cocycles}},
\nbbsteprint{\arxivref{2002.01930}{arxiv:2002.01930}}.

\bibitem{Kaderli:2019dny}
\nbbstauthor{A.~Kaderli},
\nbbsttitle{{A note on the Drinfeld associator for genus-zero superstring
  amplitudes in twisted de Rham theory}},
\nbbsteprint{\arxivref{1912.09406}{arxiv:1912.09406}}.

\bibitem{Weinzierl:2020nhw}
\nbbstauthor{S.~Weinzierl},
\nbbsttitle{{Correlation functions on the lattice and twisted cocycles}},
\nbbstjournal{\doiref{10.1016/j.physletb.2020.135449}{Phys.~Lett.~B~805,~135449~(2020)}},
\nbbsteprint{\arxivref{2003.05839}{arxiv:2003.05839}}.

\bibitem{Chen:2020uyk}
\nbbstauthor{J.~Chen, X.~Xu and L.~L.~Yang},
\nbbsttitle{{Constructing Canonical Feynman Integrals with Intersection
  Theory}},
\nbbsteprint{\arxivref{2008.03045}{arxiv:2008.03045}}.

\bibitem{Caron-Huot:2021xqj}
\nbbstauthor{S.~Caron-Huot and A.~Pokraka},
\nbbsttitle{{Duals of Feynman integrals. Part I. Differential equations}},
\nbbstjournal{\doiref{10.1007/JHEP12(2021)045}{JHEP~2112,~045~(2021)}},
\nbbsteprint{\arxivref{2104.06898}{arxiv:2104.06898}}.

\bibitem{Caron-Huot:2021iev}
\nbbstauthor{S.~Caron-Huot and A.~Pokraka},
\nbbsttitle{{Duals of Feynman Integrals. Part II. Generalized unitarity}},
\nbbstjournal{\doiref{10.1007/JHEP04(2022)078}{JHEP~2204,~078~(2022)}},
\nbbsteprint{\arxivref{2112.00055}{arxiv:2112.00055}}.

\bibitem{fresn2020quadratic}
\nbbstauthor{J.~Fresán, C.~Sabbah and J.-D.~Yu},
\nbbsttitle{Quadratic relations between periods of connections},
\nbbsteprint{\arxivref{2005.11525}{arxiv:2005.11525}}.

\bibitem{fresn2020quadratic2}
\nbbstauthor{J.~Fresán, C.~Sabbah and J.-D.~Yu},
\nbbsttitle{Quadratic relations between Bessel moments},
\nbbsteprint{\arxivref{2006.02702}{arxiv:2006.02702}}.

\bibitem{Cacciatori:2021nli}
\nbbstauthor{S.~L.~Cacciatori, M.~Conti and S.~Trevisan},
\nbbsttitle{{Co-Homology of Differential Forms and Feynman Diagrams}},
\nbbstjournal{\doiref{10.3390/universe7090328}{Universe~7,~328~(2021)}},
\nbbsteprint{\arxivref{2107.14721}{arxiv:2107.14721}}.

\bibitem{Abreu:2022mfk}
\nbbstauthor{S.~Abreu, R.~Britto and C.~Duhr},
\nbbsttitle{{The SAGEX Review on Scattering Amplitudes, Chapter 3: Mathematical
  structures in Feynman integrals}},
\nbbsteprint{\arxivref{2203.13014}{arxiv:2203.13014}}.

\bibitem{Lee:2013hzt}
\nbbstauthor{R.~N.~Lee and A.~A.~Pomeransky},
\nbbsttitle{{Critical points and number of master integrals}},
\nbbstjournal{\doiref{10.1007/JHEP11(2013)165}{JHEP~1311,~165~(2013)}},
\nbbsteprint{\arxivref{1308.6676}{arxiv:1308.6676}}.

\bibitem{Aluffi:2009aa}
\nbbstauthor{P.~Aluffi and M.~Marcolli},
\nbbsttitle{{Parametric Feynman integrals and determinant hypersurfaces}},
\nbbstjournal{\doiref{10.4310/ATMP.2010.v14.n3.a5}{Adv.~Theor.~Math.~Phys.~14,~911~(2010)}},
\nbbsteprint{\arxivref{0901.2107}{arxiv:0901.2107}}.

\bibitem{Aluffi:2011ep}
\nbbstauthor{P.~Aluffi},
\nbbsttitle{{Chern classes of graph hypersurfaces and deletion-contraction}},
\nbbstjournal{Moscow~Math~J.~12,~671~(2012)},
\nbbsteprint{\arxivref{1106.1447}{arxiv:1106.1447}}.

\bibitem{Bitoun:2017nre}
\nbbstauthor{T.~Bitoun, C.~Bogner, R.~P.~Klausen and E.~Panzer},
\nbbsttitle{{Feynman integral relations from parametric annihilators}},
\nbbstjournal{\doiref{10.1007/s11005-018-1114-8}{Lett.~Math.~Phys.~109,~497~(2019)}},
\nbbsteprint{\arxivref{1712.09215}{arxiv:1712.09215}}.

\bibitem{Bitoun:2018afx}
\nbbstauthor{T.~Bitoun, C.~Bogner, R.~P.~Klausen and E.~Panzer},
\nbbsttitle{{The number of master integrals as Euler characteristic}},
\nbbstjournal{\doiref{10.22323/1.303.0065}{PoS~LL2018,~065~(2018)}},
\nbbsteprint{\arxivref{1809.03399}{arxiv:1809.03399}}.

\bibitem{Smirnov:2010hn}
\nbbstauthor{A.~V.~Smirnov and A.~V.~Petukhov},
\nbbsttitle{{The Number of Master Integrals is Finite}},
\nbbstjournal{\doiref{10.1007/s11005-010-0450-0}{Lett.~Math.~Phys.~97,~37~(2011)}},
\nbbsteprint{\arxivref{1004.4199}{arxiv:1004.4199}}.

\bibitem{GKZ-1989}
\nbbstauthor{I.~M.~Gel'fand, A.~V.~Zelevinski\u{\i} and M.~M.~Kapranov},
\nbbsttitle{Hypergeometric functions and toric varieties},
\nbbstjournal{\doiref{10.1007/BF01078777}{Funktsional.~Anal.~i~Prilozhen.~23,~12~(1989)}},
\href{https://doi.org/10.1007/BF01078777}{\nbbsturl{https://doi.org/10.1007/BF01078777}}.

\bibitem{GKZ-Euler-1990}
\nbbstauthor{I.~M.~Gel'fand, M.~M.~Kapranov and A.~V.~Zelevinsky},
\nbbsttitle{Generalized {E}uler integrals and {$A$}-hypergeometric functions},
\nbbstjournal{\doiref{10.1016/0001-8708(90)90048-R}{Adv.~Math.~84,~255~(1990)}},
\href{https://doi.org/10.1016/0001-8708(90)90048-R}{\nbbsturl{https://doi.org/10.1016/0001-8708(90)90048-R}}.

\bibitem{Matsubara-Heo-Takayama-2020b}
\nbbstauthor{S.-J.~Matsubara-Heo and N.~Takayama},
\nbbsttitle{Algorithms for pfaffian systems and cohomology intersection numbers
  of hypergeometric integrals, Errata in
  \url{http://www.math.kobe-u.ac.jp/OpenXM/Math/intersection2/}},
in: \nbbsttitle{Lecture Notes in Computer Science},
Lecture notes in computer science,
pp.~73--84,
Springer International Publishing (2020).

\bibitem{Nasrollahpoursamami:2016}
\nbbstauthor{E.~{Nasrollahpoursamami}},
\nbbsttitle{{Periods of Feynman Diagrams and GKZ D-Modules}},
\nbbsteprint{\arxivref{1605.04970}{arxiv:1605.04970}}.

\bibitem{Regge:1968rhi}
\nbbstauthor{T.~Regge},
\nbbsttitle{{Algebraic Topology Methods in the Theory of Feynman Relativistic
  Amplitudes}},
in: \nbbsttitle{{Battelle Rencontres}},
pp.~433--458.

\bibitem{Kashiwara:1977nf}
\nbbstauthor{M.~Kashiwara and T.~Kawai},
\nbbsttitle{{Holonomic Systems of Linear Differential Equations and Feynman
  Integrals}},
\nbbstjournal{\doiref{10.2977/prims/1195196602}{Publ.~Res.~Inst.~Math.~Sci.~Kyoto~12,~131~(1977)}}.

\bibitem{Hosono:1993qy}
\nbbstauthor{S.~Hosono, A.~Klemm, S.~Theisen and S.-T.~Yau},
\nbbsttitle{{Mirror symmetry, mirror map and applications to Calabi-Yau
  hypersurfaces}},
\nbbstjournal{\doiref{10.1007/BF02100589}{Commun.~Math.~Phys.~167,~301~(1995)}},
\nbbsteprint{\arxivref{hep-th/9308122}{hep-th/9308122}}.

\bibitem{Hosono:1994ax}
\nbbstauthor{S.~Hosono, A.~Klemm, S.~Theisen and S.-T.~Yau},
\nbbsttitle{{Mirror symmetry, mirror map and applications to complete
  intersection Calabi-Yau spaces}},
\nbbstjournal{\doiref{10.1016/0550-3213(94)00440-P}{Nucl.~Phys.~B~433,~501~(1995)}},
\nbbsteprint{\arxivref{hep-th/9406055}{hep-th/9406055}}.

\bibitem{Fleischer:2003rm}
\nbbstauthor{J.~Fleischer, F.~Jegerlehner and O.~V.~Tarasov},
\nbbsttitle{{A New hypergeometric representation of one loop scalar integrals
  in d dimensions}},
\nbbstjournal{\doiref{10.1016/j.nuclphysb.2003.09.004}{Nucl.~Phys.~B~672,~303~(2003)}},
\nbbsteprint{\arxivref{hep-ph/0307113}{hep-ph/0307113}}.

\bibitem{Kalmykov:2008gq}
\nbbstauthor{M.~Y.~Kalmykov, B.~A.~Kniehl, B.~F.~L.~Ward and S.~A.~Yost},
\nbbsttitle{{Hypergeometric functions, their epsilon expansions and Feynman
  diagrams}},
\nbbsteprint{\arxivref{0810.3238}{arxiv:0810.3238}},
in: \nbbsttitle{{15th International Seminar on High Energy Physics}}.

\bibitem{Vanhove:2018mto}
\nbbstauthor{P.~Vanhove},
\nbbsttitle{{Feynman integrals, toric geometry and mirror symmetry}},
\nbbsteprint{\arxivref{1807.11466}{arxiv:1807.11466}},
in: \nbbsttitle{{KMPB Conference}: {Elliptic Integrals, Elliptic Functions and
  Modular Forms in Quantum Field Theory}},
pp.~415--458.

\bibitem{Kalmykov:2020cqz}
\nbbstauthor{M.~Kalmykov, V.~Bytev, B.~A.~Kniehl, S.-O.~Moch, B.~F.~L.~Ward and
  S.~A.~Yost},
\nbbsttitle{{Hypergeometric Functions and Feynman Diagrams}},
\nbbsteprint{\arxivref{2012.14492}{arxiv:2012.14492}},
in: \nbbsttitle{{Antidifferentiation and the Calculation of Feynman
  Amplitudes}}.

\bibitem{Blumlein:2021hbq}
\nbbstauthor{J.~Bl\"umlein, M.~Saragnese and C.~Schneider},
\nbbsttitle{{Hypergeometric Structures in Feynman Integrals}},
\nbbsteprint{\arxivref{2111.15501}{arxiv:2111.15501}}.

\bibitem{delaCruz:2017zqr}
\nbbstauthor{L.~de~la~Cruz, A.~Kniss and S.~Weinzierl},
\nbbsttitle{{Properties of scattering forms and their relation to
  associahedra}},
\nbbstjournal{\doiref{10.1007/JHEP03(2018)064}{JHEP~1803,~064~(2018)}},
\nbbsteprint{\arxivref{1711.07942}{arxiv:1711.07942}}.

\bibitem{Klausen:2019hrg}
\nbbstauthor{R.~P.~Klausen},
\nbbsttitle{{Hypergeometric Series Representations of Feynman Integrals by GKZ
  Hypergeometric Systems}},
\nbbstjournal{\doiref{10.1007/JHEP04(2020)121}{JHEP~2004,~121~(2020)}},
\nbbsteprint{\arxivref{1910.08651}{arxiv:1910.08651}}.

\bibitem{Faugere-1999}
\nbbstauthor{J.-C.~Faug{\`e}re},
\nbbsttitle{A new efficient algorithm for computing Gr{\"o}bner bases ($F_4$)},
\nbbstjournal{\doiref{https://doi.org/10.1016/S0022-4049(99)00005-5}{Journal~of~Pure~and~Applied~Algebra~139,~61~(1999)}}.

\bibitem{Macaulay-1903}
\nbbstauthor{F.~S.~MacAulay},
\nbbsttitle{Some {F}ormulae in {E}limination},
\nbbstjournal{\doiref{10.1112/plms/s1-35.1.3}{Proc.~Lond.~Math.~Soc.~35,~3~(1903)}},
\href{https://doi.org/10.1112/plms/s1-35.1.3}{\nbbsturl{https://doi.org/10.1112/plms/s1-35.1.3}}.

\bibitem{Hibi-Nishiyama-Takayama-2017}
\nbbstauthor{T.~Hibi, K.~Nishiyama and N.~Takayama},
\nbbsttitle{Pfaffian systems of A-hypergeometric equations I: Bases of twisted
  cohomology groups},
\nbbstjournal{\doiref{10.1016/j.aim.2016.10.021}{Adv.~Math.~(N.~Y.)~306,~303~(2017)}}.

\bibitem{Ohara-Takayama-2015}
\nbbstauthor{K.~Ohara and N.~Takayama},
\nbbsttitle{Pfaffian Systems of A-Hypergeometric Systems II --- Holonomic
  Gradient Method},
\nbbsteprint{\arxivref{1505.02947}{arxiv:1505.02947}}.

\bibitem{url-hgm-refs}
\nbbsttitle{{References for the Holonomic Gradient Method (HGM) and the
  Holonomic Gradient Descent Method (HGD)}},
\href{http://www.math.kobe-u.ac.jp/OpenXM/Math/hgm/ref-hgm.html}{\nbbsturl{http://www.math.kobe-u.ac.jp/OpenXM/Math/hgm/ref-hgm.html}}.

\bibitem{delaCruz:2019skx}
\nbbstauthor{L.~de~la~Cruz},
\nbbsttitle{{Feynman integrals as A-hypergeometric functions}},
\nbbstjournal{\doiref{10.1007/JHEP12(2019)123}{JHEP~1912,~123~(2019)}},
\nbbsteprint{\arxivref{1907.00507}{arxiv:1907.00507}}.

\bibitem{Klausen:2021yrt}
\nbbstauthor{R.~P.~Klausen},
\nbbsttitle{{Kinematic singularities of Feynman integrals and principal
  A-determinants}},
\nbbstjournal{\doiref{10.1007/JHEP02(2022)004}{JHEP~2202,~004~(2022)}},
\nbbsteprint{\arxivref{2109.07584}{arxiv:2109.07584}}.

\bibitem{Tellander:2021xdz}
\nbbstauthor{F.~Tellander and M.~Helmer},
\nbbsttitle{{Cohen-Macaulay Property of Feynman Integrals}},
\nbbstjournal{\doiref{10.1007/s00220-022-04569-6}{Commun.~Math.~Phys.~399,~1021~(2023)}},
\nbbsteprint{\arxivref{2108.01410}{arxiv:2108.01410}}.

\bibitem{Feng:2019bdx}
\nbbstauthor{T.-F.~Feng, C.-H.~Chang, J.-B.~Chen and H.-B.~Zhang},
\nbbsttitle{{GKZ-hypergeometric systems for Feynman integrals}},
\nbbstjournal{\doiref{10.1016/j.nuclphysb.2020.114952}{Nucl.~Phys.~B~953,~114952~(2020)}},
\nbbsteprint{\arxivref{1912.01726}{arxiv:1912.01726}}.

\bibitem{Tachibana-Goto-Koyama-Takayama-2020}
\nbbstauthor{Y.~Tachibana, Y.~Goto, T.~Koyama and N.~Takayama},
\nbbsttitle{Holonomic gradient method for two-way contingency tables},
\nbbstjournal{\doiref{10.2140/astat.2020.11.125}{Algebraic~Statistics~11,~125–153~(2020)}}.

\bibitem{url-asir}
\nbbsttitle{{OpenXM project (including Risa/Asir distribution)}},
\href{http://www.openxm.org}{\nbbsturl{http://www.openxm.org}}.

\bibitem{url-maple}
\nbbstauthor{{Maplesoft, a division of Waterloo Maple Inc.}},
\nbbsttitle{Maple},
\href{https://hadoop.apache.org}{\nbbsturl{https://hadoop.apache.org}}.

\bibitem{url-Mathematica}
\nbbstauthor{W.~R.~Inc.},
\nbbsttitle{Mathematica, {V}ersion 13.0.0},
Champaign, IL, 2021,
\href{https://www.wolfram.com/mathematica}{\nbbsturl{https://www.wolfram.com/mathematica}}.

\bibitem{url-FiniteFlow}
\nbbstauthor{T.~Peraro},
\nbbsttitle{FiniteFlow},
\href{https://github.com/peraro/finiteflow}{\nbbsturl{https://github.com/peraro/finiteflow}}.

\bibitem{Lee:2012cn}
\nbbstauthor{R.~N.~Lee},
\nbbsttitle{{Presenting LiteRed: a tool for the Loop InTEgrals REDuction}},
\nbbsteprint{\arxivref{1212.2685}{arxiv:1212.2685}}.

\bibitem{Lee:2013mka}
\nbbstauthor{R.~N.~Lee},
\nbbsttitle{{LiteRed 1.4: a powerful tool for reduction of multiloop
  integrals}},
\nbbstjournal{\doiref{10.1088/1742-6596/523/1/012059}{J.~Phys.~Conf.~Ser.~523,~012059~(2014)}},
\nbbsteprint{\arxivref{1310.1145}{arxiv:1310.1145}}.

\bibitem{url-mm-data}
\nbbsttitle{{Macaulay Matrix for Feynman Integrals: Programs and Data}},
\href{http://www.math.kobe-u.ac.jp/OpenXM/Math/amp-MM}{\nbbsturl{http://www.math.kobe-u.ac.jp/OpenXM/Math/amp-MM}}.

\bibitem{Adolphson-1994}
\nbbstauthor{A.~Adolphson},
\nbbsttitle{Hypergeometric functions and rings generated by monomials},
\nbbstjournal{\doiref{10.1215/S0012-7094-94-07313-4}{Duke~Math.~J.~73,~269~(1994)}},
\href{https://doi.org/10.1215/S0012-7094-94-07313-4}{\nbbsturl{https://doi.org/10.1215/S0012-7094-94-07313-4}}.

\bibitem{SST}
\nbbstauthor{M.~Saito, B.~Sturmfels and N.~Takayama},
\nbbsttitle{Gr{\"o}bner deformations of hypergeometric differential equations},
Algorithms and computation in mathematics,
Springer (2011),
Berlin, Germany.

\bibitem{polymake:2000}
\nbbstauthor{E.~Gawrilow and M.~Joswig},
\nbbsttitle{{\tt polymake}: a framework for analyzing convex polytopes},
in: \nbbsttitle{Polytopes---combinatorics and computation ({O}berwolfach,
  1997)},
DMV Sem.,
pp.~43--73,
Birkh\"auser, Basel (2000).

\bibitem{AdoSpa}
\nbbstauthor{A.~Adolphson and S.~Sperber},
\nbbsttitle{{$A$}-hypergeometric systems that come from geometry},
\nbbstjournal{\doiref{10.1090/S0002-9939-2011-11073-6}{Proc.~Amer.~Math.~Soc.~140,~2033~(2012)}},
\href{https://doi.org/10.1090/S0002-9939-2011-11073-6}{\nbbsturl{https://doi.org/10.1090/S0002-9939-2011-11073-6}}.

\bibitem{SSTip}
\nbbstauthor{M.~Saito, B.~Sturmfels and N.~Takayama},
\nbbsttitle{Hypergeometric Polynomials and Integer Programming},
\nbbstjournal{\doiref{10.1023/a:1000609524994}{Compositio~Mathematica~115,~231~(1999)}}.

\bibitem{dojo}
\nbbstauthor{T.~Hibi et~al.},
\nbbsttitle{Gr{\"o}bner bases},
2013 edition,
Springer (2014),
Tokyo, Japan.

\bibitem{Watanabe:2013ova}
\nbbstauthor{N.~Watanabe and T.~Kaneko},
\nbbsttitle{{One loop integration with hypergeometric series by using recursion
  relations}},
\nbbstjournal{\doiref{10.1088/1742-6596/523/1/012063}{J.~Phys.~Conf.~Ser.~523,~012063~(2014)}},
\nbbsteprint{\arxivref{1309.3118}{arxiv:1309.3118}}.

\bibitem{Klemm:2019dbm}
\nbbstauthor{A.~Klemm, C.~Nega and R.~Safari},
\nbbsttitle{{The $l$-loop Banana Amplitude from GKZ Systems and relative
  Calabi-Yau Periods}},
\nbbstjournal{\doiref{10.1007/JHEP04(2020)088}{JHEP~2004,~088~(2020)}},
\nbbsteprint{\arxivref{1912.06201}{arxiv:1912.06201}}.

\bibitem{speer1969theory}
\nbbstauthor{E.~R.~Speer},
\nbbsttitle{{Generalized Feynman Amplitudes}},
Princeton University Press (1969).

\bibitem{Speer1971}
\nbbstauthor{E.~R.~Speer and M.~J.~Westwater},
\nbbsttitle{Generic Feynman amplitudes},
\nbbstjournal{Annales~de~l'I.H.P.~Physique~théorique~14,~1~(1971)},
\href{http://eudml.org/doc/75685}{\nbbsturl{http://eudml.org/doc/75685}}.

\bibitem{LeeTalkMoriond}
\nbbstauthor{R.~Lee},
\nbbsttitle{{Modern techniques of multiloop calculations}},
{Talk given at: Moriond QCD and High Energy Interactions},
\href{{https://moriond.in2p3.fr/QCD/2014/ThursdayMorning/Lee.pdf}}{\nbbsturl{{https://moriond.in2p3.fr/QCD/2014/ThursdayMorning/Lee.pdf}}}.

\bibitem{Lee:2014tja}
\nbbstauthor{R.~N.~Lee},
\nbbsttitle{{Modern techniques of multiloop calculations}},
\nbbsteprint{\arxivref{1405.5616}{arxiv:1405.5616}},
in: \nbbsttitle{{49th Rencontres de Moriond on QCD and High Energy
  Interactions}},
pp.~297--300.

\bibitem{SameshimaPhD}
\nbbstauthor{R.~Sameshima},
\nbbsttitle{{On Different Parametrizations of Feynman Integrals}},
{PhD Thesis, City University of New York (CUNY)},
\href{{https://academicworks.cuny.edu/gc\_etds/3376/}}{\nbbsturl{{https://academicworks.cuny.edu/gc\_etds/3376/}}}.

\bibitem{Chestnov:2023kww}
\nbbstauthor{V.~Chestnov, S.~J.~Matsubara-Heo, H.~J.~Munch and N.~Takayama},
\nbbsttitle{{Restrictions of Pfaffian Systems for Feynman Integrals}},
\nbbsteprint{\arxivref{2305.01585}{arxiv:2305.01585}}.

\bibitem{newBateman}
\nbbstauthor{N.~Takayama},
\nbbsttitle{$A$-Hypergeometric Functions},
in: \nbbsttitle{Encyclopedia of Special Functions: The Askey-Bateman Project},
pp.~101–121,
ed.: T.~H.~Koornwinder and J.~V.~Stokman,
Cambridge University Press (2020).

\bibitem{matsubaraheo2019euler}
\nbbstauthor{S.-J.~Matsubara-Heo},
\nbbsttitle{Euler and Laplace integral representations of GKZ hypergeometric
  functions},
\nbbsteprint{\arxivref{1904.00565}{arxiv:1904.00565}}.

\bibitem{integrableconnection}
\nbbstauthor{M.~Barkatou, T.~Cluzeau, C.~El~Bacha and J.-A.~Weil},
\nbbsttitle{{IntegrableConnections, A Maple package for computing closed form
  solutions of integrable connections}},
\href{{http://www.unilim.fr/pages\_perso/thomas.cluzeau/Packages/IntegrableConnections/PDS.html}}{\nbbsturl{{http://www.unilim.fr/pages\_perso/thomas.cluzeau/Packages/IntegrableConnections/PDS.html}}}.

\bibitem{Robbiano:1985}
\nbbstauthor{L.~Robbiano},
\nbbsttitle{Term orderings on the polynomial ring},
in: \nbbsttitle{E{UROCAL} '85, {V}ol. 2 ({L}inz, 1985)},
Lecture Notes in Comput. Sci.,
pp.~513--517,
Springer, Berlin (1985).

\end{thebibliography}

\end{document}